\documentclass[a4paper, 11pt]{article} %mögliche Schriftgrößen: 10pt, 11pt, 12pt, ansonsten {extarticle} für 8pt, 9pt, 10pt, 11pt, 12pt, 14pt, 17pt, 20pt
\usepackage[margin=2.5cm, top=2.5cm, bottom=3cm, footskip=1cm]{geometry}	%definiert Seitenraender

\usepackage{amsmath}		%verbessert das Setzen mathematischer Ausdruecke
\usepackage{amsfonts}		%siehe amsmath
\usepackage{amsthm}
\usepackage[utf8]{inputenc}
\usepackage{color}
\usepackage{tikz}
\usepackage{slashed}
\usepackage{mathrsfs} 
\usetikzlibrary{calc}
\usetikzlibrary{arrows}  
\usetikzlibrary{decorations.markings}
\usetikzlibrary{shapes,snakes}
\usepackage{hyperref} 
\hypersetup{
	bookmarks=true,
	unicode=false,
	pdftoolbar=true,
	pdfmenubar=true,
	pdffitwindow=false,
	pdfstartview={FitH},
	pdftitle={My title},
	pdfauthor={Author},
	pdfsubject={Subject},
	pdfcreator={Creator},
	pdfproducer={Producer},
	pdfkeywords={keyword1} {key2} {key3},
	pdfnewwindow=true,
	colorlinks=true,   %Colour hyperlinks in document
	linkcolor={blue!80!black}, %Specify which colour to use to link equations, sections, etc. are 
	citecolor={CrimsonRed}, %Specify which colour to use for citation
	filecolor={blue!80!black},
	urlcolor={blue},      %Specify hyperlink colour
	hypertexnames=false, %
} 
\usepackage[figure]{hypcap}
\usepackage{amssymb}		%siehe amsmath
\usepackage{amstext}		%ermoeglicht \text und mehr
\usepackage{graphicx}		%einbinden von Grafiken
\usepackage[utf8]{inputenc}	%Zeichensatz
\usepackage[T1]{fontenc}	%Silbentrennung und mehr
\usepackage{pdfpages}		%einbinden von PDF
\usepackage{wrapfig}		%einbinden von Bildern in Text
\usepackage{here}			%laut float-Beschreibung obsolet
\usepackage{float}			%ermoeglicht den H-parameter bei Bildern
\usepackage{fancyhdr}
\usepackage{verbatim}		%ermoeglicht mehrzeilige Kommentare mit \begin{comment} \end{comment}
\usepackage{wasysym}		%erweiterter Zeichensatz
\usepackage{tocloft}		%Gestaltung des Inhaltsverzeichnisses und mehr
\usepackage{siunitx}		%fehlerfreie Darstellung von Werten und Einheiten
\usepackage[arrow, matrix, curve]{xy} %kommutative Diagramme
\usepackage{wrapfig,tabularx,multirow,hyperref,icomma}
\usepackage{mathptmx}
\usepackage{lmodern}
\usepackage{courier}
\usepackage{paralist}
\usepackage{afterpage}
\usepackage{subcaption}
\usepackage[section]{placeins}
{\left.\begin{aligned}}
	{\end{aligned}\right\rbrace}
	
\usepackage{etoolbox}
\usepackage{relsize,exscale}
\usepackage[titletoc]{appendix}
\usepackage{natbib}
\usepackage[bb=boondox]{mathalpha}

\usepackage{nameref}
\definecolor{RoyalBlue}{RGB}{35,32, 57}
\definecolor{lightblue}{RGB}{132,131,143}
\definecolor{CrimsonRed}{RGB}{130,0,29}

\newtheorem{definition}{Definition}[section]

\newtheorem{theorem}[definition]{Theorem}

\newtheorem{proposition}[definition]{Proposition}

\newtheorem{corollary}[definition]{Corollary}

\newtheorem{lemma}[definition]{Lemma}

\theoremstyle{definition}
\newtheorem{remark}[definition]{Remark}

\theoremstyle{definition}
\newtheorem*{notation}{Notation}

\theoremstyle{definition}
\newtheorem{Regularity}[definition]{Assumption}

\theoremstyle{definition}

\newcommand{\res}[1]{\raisebox{0.0ex}{\scalebox{1}{$\underset{#1}{\mathrm{Res}}\;$}}}%\res{f}_{z=a} for residue of f at z=a 

\renewcommand{\Im}{\mathrm{Im}} %imaginary part of a function
\renewcommand{\Re}{\mathrm{Re}} %real part of a function

\newcommand\numberthis{\addtocounter{equation}{1}\tag{\theequation}}
\numberwithin{equation}{section}

\makeatletter
\patchcmd{\chapter}{\if@openright\cleardoublepage\else\clearpage\fi}{}{}{}
\makeatother

\usepackage{graphicx}
\usepackage{calc}
\usepackage{tocloft}
\usepackage[nottoc,notlof]{tocbibind}
\usepackage{emptypage}
\usepackage{hyperref}
\usepackage{bashful}
\usepackage{enumerate}

\usepackage{authblk}

\title{A mathematical analogy between certain hollow vortex and bubble flows -- existence, uniqueness, and identical boundary shapes}

\author[1]{Valentin D. Kunz}
\affil[1]{{\footnotesize Department of Mathematics, Ohio State University, 231 W 18th Ave, Columbus, OH 43210, USA}}
\author[1]{Saleh  Tanveer}
\author[2]{Darren G. Crowdy}
\affil[2]{{\footnotesize Department of Mathematics, Imperial College London, 180 Queen's Gate, London SW7 2AZ, UK }}
\affil[ ]{{\footnotesize
		v$\_$kunz@startmail.com	}}
\date{\today}
\begin{document}

\maketitle

	\begin{abstract}
		We study two physically distinct steady free-boundary problems in a two-dimensional inviscid fluid: a hollow vortex in a constant-vorticity flow subject to simple shear, third order and linear strain at infinity, without surface tension, and an irrotational flow past a constant-pressure bubble subject to surface tension, linear strain, and circulation at infinity.
		Building on the mathematical vortex–bubble correspondence recently developed by Crowdy and Tanveer, we show that both problems are linked by a common mathematical formulation. For the vortex problem, we extend the sequence-space fixed-point approach of Crowdy and Tanveer to obtain new local existence, uniqueness, and analyticity results. For the bubble problem, we develop a separate existence theory, reducing it to a nonlinear singular integral operator fixed-point problem and deriving a scalar solvability condition that defines a real-analytic surface in the physical parameter space. On this surface we obtain locally unique bubble solutions with real-analytic boundaries and holomorphically extendible velocity fields. Finally, we prove that the scalar solvability condition is equivalent to a functional conformal-mapping identity. The scalar condition between the parameters is enough for the vortex and bubble problems to have the same free-boundary shape, despite their different exterior flows.
	\end{abstract}

\section{Introduction}

Free-boundary problems for the incompressible Euler equations arise when part of the boundary of the fluid region is itself unknown and must be determined simultaneously with the flow. Even for steady two-dimensional flows, the resulting problems retain a nonlinear coupling between the geometry of the free boundary and the velocity field. Two classical examples are provided by hollow vortices and gas bubbles in an inviscid liquid. In the former, the geometry of a constant-pressure vortex boundary is determined by the surrounding rotational flow, while in the latter the pressure jump across the interface is balanced by surface tension and hence by its curvature. The present article studies such hollow vortex, subject to shear and third order strain, as well as a bubble subject to strain and circulation. The precise mathematical formulations of these problems are given below.

Beyond their classical role in free-streamline theory, hollow vortices are connected to several broader themes in mathematical physics. Point vortices provide a finite-dimensional Hamiltonian model for concentrated vorticity, whose relative equilibria and stability have been studied extensively, while hollow vortices can be viewed as free-boundary desingularisations of such singular vortex configurations \cite{Aref1983,ChenWalshWheeler2026, Newton2001}. For certain choice of parameters, the particular hollow-vortex problem considered in the present article can be interpreted as a limiting case of a corotating five-vortex configuration in which a central vortex is modelled as a hollow vortex and four distant point vortices of equal circulation are situated at the vertices of a square, as discussed in a companion manuscript by the present authors, submitted for publication. Exact equilibria consisting of a central hollow vortex surrounded by a non-rotating polygonal array of point vortices were previously constructed in \cite{CrowdyRoenby2014}. Related hollow-core models also arise in the theory of quantized vortices in superfluid helium \cite{Donnelly1991,LehmannSchmied2003}. Capillary bubbles, on the other hand, are canonical examples of deformable interfaces in multiphase flow  \cite{CliftGraceWeber1978,LegendreZenit2025, YanEtAl2023}. Thus, both classes of problems connect nonlinear free-boundary phenomena with fundamental questions in vortex dynamics, interfacial mechanics, and transport.
Historically, however, hollow vortices and capillary bubbles have largely been studied independently.

Despite this historical separation, the same free-boundary profiles have repeatedly appeared in physically distinct Euler flows, most notably in vortical flows without surface tension and in irrotational capillary flows. Crapper's classical exact solutions for pure capillary waves in irrotational flow \cite{Crapper1957} were later found to reappear in rotational flows without surface tension: numerical computations of Dyachenko and Hur \cite{DyachenkoHur2019b,DyachenkoHur2019a} first indicated the emergence of Crapper's profiles in an appropriate limiting regime, further numerical and asymptotic evidence was obtained by Hur and Vanden-Broeck \cite{HurVandenBroeck2020}, and Hur and Wheeler \cite{HurWheeler2020} subsequently constructed exact constant-vorticity waves having precisely the same free-surface shapes. Related correspondences were found in other geometries. Crowdy and Roenby \cite{CrowdyRoenby2014} obtained travelling water waves carrying submerged point vortices whose free surfaces are again related to Crapper's capillary-wave profiles. They also found that equilibria for the aforementioned non-rotating point-vortex configuration, with a single hollow vortex at its centre, admit the same class of free-boundary shapes as those found by Wegmann and Crowdy \citep{Crowdy1999, WegmannCrowdy2000} for a capillary bubble  with circulation, but without strain. It was later found by Crowdy, Nelson and Krishnamurthy \cite{CrowdyNelsonKrishnamurthy2021} that this same class of free-boundary curves also describes equilibria of a single rotating hollow vortex without surface tension. Taken together, these results suggested that the recurrence of identical free-boundary profiles across physically distinct flows reflects an underlying mathematical correspondence rather than a collection of isolated coincidences. 

Recent work by Crowdy and Tanveer \cite{CrowdyTanveer2025} has provided a mathematical explanation for these correspondences, showing that certain hollow-vortex problems without surface tension and certain irrotational free-boundary problems with surface tension, despite their quite different physical origins, have a common mathematical framework. These include the  problems of a hollow vortex in a simple shear and linear strain as well as a translating bubble in the presence of surface tension but without circulation, which, individually, have been studied extensively in their respective settings; see, for instance, \cite{BakerSaffmanSheffield1976,LlewellynSmithCrowdy2012,ZannettiFerlautoLlewellynSmith2016} for related hollow-vortex problems and \cite{McLeod1955,VandenBroeckKeller1980,Tanveer1996} for the translating-bubble problem. 
Outside the special cases of vanishing shear for the vortex and vanishing Bernoulli constant for the bubble, exact solutions are not generally available.

The present article considers a further pair of physically distinct problems, which we denote by (i) (vortex) and (ii) (bubble) in the enumeration below. These problems generalize several of the aforementioned radial configurations. For the Vortex problem (i), we adapt the sequence-space argument developed by Crowdy and Tanveer \cite{CrowdyTanveer2025} for a vortex in linear strain. The resulting fixed-point equation contains an additional dependence on an unknown Laurent coefficient and requires a modified contraction argument; cf.\! Remark \ref{rem:11}.

For the Bubble problem (ii), substantially different analytical machinery is required. As in \cite{CrowdyTanveer2025}, the vortex and bubble problems can be embedded into a common conformal-mapping formulation. A solution of this common problem, however, need not satisfy the additional compatibility condition imposed by the bubble's complex velocity potential. We therefore reformulate the bubble problem as a singular-integral fixed-point problem and express the compatibility condition as a scalar solvability condition which defines a real-analytic surface in the physical parameter space. Away from certain exceptional parameter values, this formulation yields locally unique bubble solutions by a contraction argument. At the exceptional values, where the direct fixed-point construction degenerates, we use the common conformal-mapping formulation together with the existence theory for the vortex problem to analytically continue the bubble solutions across the exceptional points.

 In a companion paper \cite{CKT}, we present exact solutions in the case of vanishing  shear (vortex) and vanishing Bernoulli constant (bubble), respectively, elucidate the physical features of the problems considered here, and provide numerical computations that supplement and extend our theoretical results beyond the parameter regime considered in the present article. \\

\textbf{Problems studied in the present Article.} We consider the two
physical problems of
\begin{enumerate}[(i)]
  \item a hollow 2-D inviscid equilibrium vortex in a simple shear and third
  order straining flow, allowing for a linear straining component;
  \item a 2D bubble in an inviscid-irrotational flow with linear strain at
  infinity and surface tension on the boundary, allowing for circulation at infinity, 
\end{enumerate}
These problems are illustrated in Figure \ref{Fig.01} below. Mathematically, both (i)
and (ii) are free boundary value problems. After suitable scaling and
non-dimensionalisation, they can be formulated as follows. We identify
$\mathbb{R}^2 = \mathbb{C}$ and consider a connected region $\Omega \subset
\mathbb{C}$ with sufficiently regular boundary $\partial \Omega \subset \mathbb{C}$ describing a
simple closed curve. We also consider a ``streamfunction'' $\psi : \Omega
\rightarrow \mathbb{R}$. A-priori regularity assumptions on $\partial \Omega$
and $\psi$ will be stated within Section \ref{sec:write-up2.2}, in terms of a
conformal mapping $Z : \mathbb{D} \rightarrow \Omega \cup \{\infty\}$, that maps the unit disc
$\mathbb{D} \subset \mathbb{C}$ onto $\Omega \cup \{\infty\}$.
\begin{enumerate}[(i)]
  \item The pair $(\psi, \Omega)$ satisfies the equations
  \begin{alignat*}{3}
    \Delta \psi & = 2 \gamma, &&  \quad \text{in} \quad \phantom{\partial}\Omega; \numberthis \label{e:1.1}\\
    \psi = 0 \quad \text{and} \quad | \nabla \psi |& = 1,
    && \quad \text{on} \quad \partial \Omega . \label{e:1.2}\numberthis
  \end{alignat*}
  Moreover, if we let $\xi_V : \Omega \rightarrow \mathbb{C}$ denote the
  analytic function defined via
  \begin{equation}
    2 i \partial_z \psi = i \gamma \bar{z} + \xi_V (z), \label{e:01}
  \end{equation}
  then the physical condition in the far-field implies
  \begin{equation}
    \xi_V (z) \sim i (\hat{\alpha} z^3 + (\beta - \gamma) z + \delta + \omega
    z^{- 1})  + \mathcal{O} (z^{- 2}), \quad \text{as} \quad z \rightarrow
    \infty. \label{e:write-up02}
  \end{equation}
  $\hat{\alpha} $ and $\gamma$ are real whereas $\beta$ is generally
  complex; $\omega$ is also real. Physically, $2 \gamma$ is the shear-rate (so $-2\gamma$ is the vorticity). The
  parameters $\hat{\alpha}$ and $\beta$ physically describe the
  third and first order strain components, and are specified a-priori along
  with $\gamma$. In \cite{CKT}, it was shown that we may take $\hat{\alpha}<0$ without loss of generality, which we will do throughout this article.  $\delta$ and $- 2 \pi \omega$, which physically describe drift and
  circulation, will both be determined as part of the solution,
  see Corollary \ref{prop:write-up26}.
  
  \item Let $\kappa : \partial \Omega \rightarrow \mathbb{R}$ denote the
  curvature of $\partial \Omega$. The pair $(\psi, \Omega)$ satisfies the
  equations
  \begin{alignat*}{3}
    \Delta \psi & = 0, &&  \quad \text{in} \quad   \phantom{\partial}\Omega; \label{e:1.5} \numberthis \\
    \psi = 0 \quad \text{and} \quad | \nabla \psi |^2 + \kappa & = \gamma, &&
    \quad \text{on} \quad \partial \Omega, \label{e:1.6} \numberthis
  \end{alignat*}
  subject to the asymptotic condition
  \begin{equation}
    2 i \partial_z \psi \sim \alpha z - \frac{i \Gamma_0}{z} +
    \mathcal{O} (z^{- 2}), \quad \text{as} \quad z \rightarrow \infty .
    \label{e:02}
  \end{equation}
  $\gamma, \alpha$ and $\Gamma_0$ are all real. $\gamma$ represents the
  Bernoulli constant, $\alpha$ describes the strain-rate and $2\pi \Gamma_0$
  describes the circulation. Without loss of generality, we will choose $\alpha >0$ (cf.\! \cite{CKT}).  As will be proved in the ensuing, only two of the
  parameters $(\gamma, \alpha, \Gamma_0)$ may be specified with the remaining
  parameter determined through a scalar constraint. All these parameters are
  part of the problem formulation, see Section \ref{sec:write-up2.2}.
  
  \quad Since $\psi$ is harmonic in $\Omega$, it is the imaginary part of a
  (possibly multivalued) function $W (z)$ holomorphic in $\Omega$, i.e.,
  \begin{equation}
    \psi (z) = \Im (W (z)).
  \end{equation}
  Note, however, $dW/dz$, which is related to the physical velocity, must be
  single-valued. We will refer to $W (z)$ as \textit{complex velocity
  potential}.
\end{enumerate}
In Problems (i) and (ii), the region $\mathbb{C} \setminus \Omega$ is called a
\textit{hollow vortex} and \textit{bubble}, respectively. The region
$\mathbb{C} \setminus \Omega$ is simply connected, in the sense that there exists a univalent conformal map $Z: \mathbb{D} \to \Omega \cup \{\infty\}$. So, we
consider here only a single hollow vortex (resp. bubble). \\

	\begin{figure}[h!]
			\includegraphics[width=\textwidth]{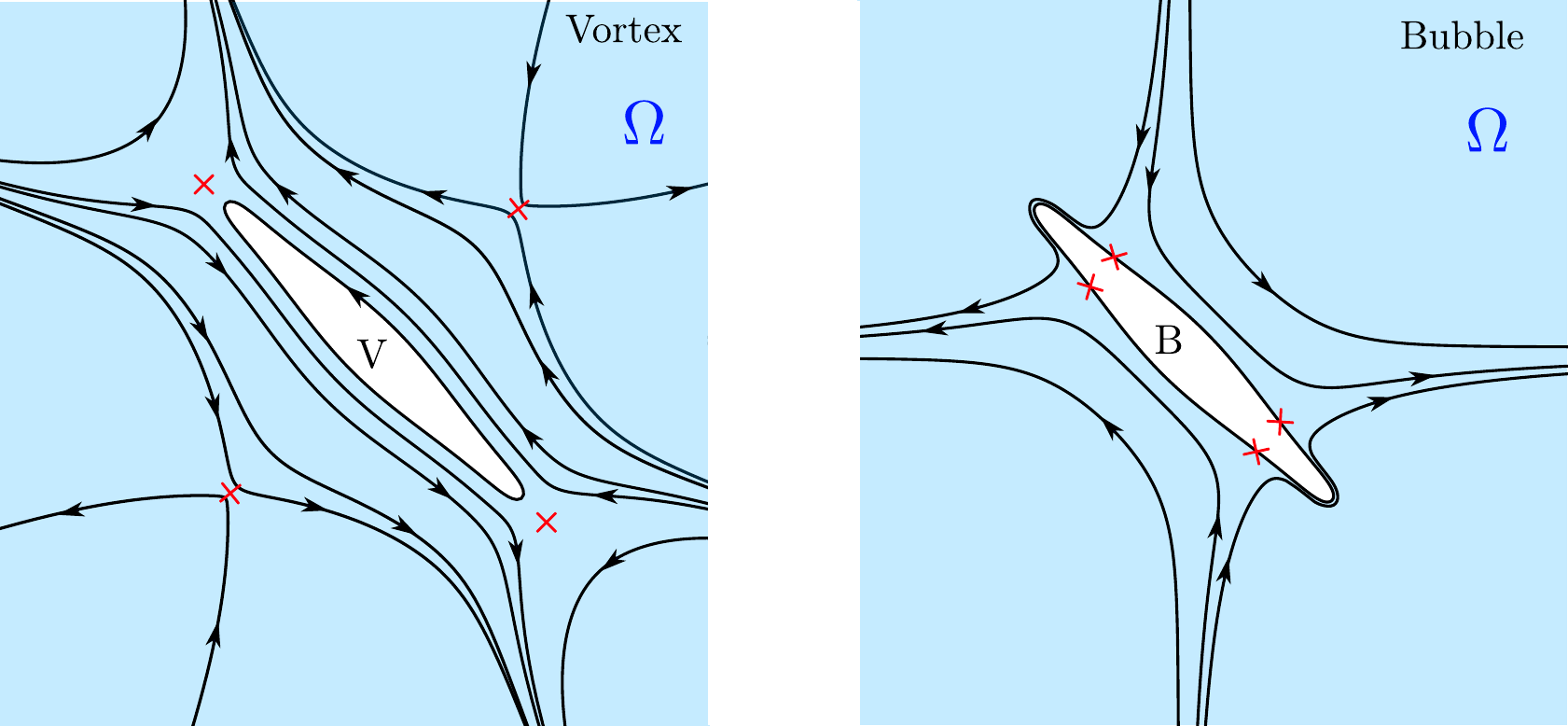}
				\caption{\small Illustration of the Vortex (left) and Bubble (right) problem. The regions labelled V and B show the hollow vortex and the bubble, respectively. The lines are streamlines whereon $\psi$ is constant, and the arrows indicate the fluid's flow direction. These figures are illustrative adaptations of the numerical solutions to Problems (i) and (ii), respectively, provided in \cite{CKT}. Particularly, the physical parameters are chosen such that the vortex and bubble shape are identical. The crosses indicate stagnation points.}
					\label{Fig.01}
	\end{figure}

\textbf{Aim and Scope of the Article.} The main objective of this article
is to prove that when $| \gamma |$ is sufficiently small, then both,
Problems (i) and (ii), admit a locally unique solution pair $(\psi, \Omega)$ with
$\partial \Omega$ real-analytic and the velocity-related functions $\xi_V$ and $dW/dz$ extend holomorphically past $\partial \Omega$. There are some other conditions that need to be satisfied
in order for such solution to exist, and these will be discussed in detail
throughout this article.\footnote{For instance, $\hat{\alpha}$ should not
approach $- \frac{1}{5}$.} We will give an overview of our results in Section
\ref{sec:write-up2.2}, once the appropriate
conformal mapping framework has been introduced. The article's structure and main-results are also summarised in the diagrammatic scheme of Figure \ref{Fig:02}.

As will transpire, both, Problems (i) and (ii), are closely related and, in
fact, mathematically analogous. This mathematical analogy will be used in
order to study symmetry properties of the problems' solutions, and determine
analyticity of $\partial \Omega$. We will also use it to show that a non-trivial conformal mapping identity is satisfied by the function $Z : \mathbb{D} \rightarrow \Omega \cup \{\infty\}$ mapping $\mathbb{D}$ onto the fluid-region  $\Omega$ if, and only if, the physical parameters $(\gamma, \alpha, \Gamma_0)$ (resp. $(\gamma, \hat{\alpha}, \beta)$) are constrained
to lie on a certain real-analytic surface $\mathcal{C}_f$, cf. the diagrammatic scheme in Figure
\ref{Fig:02}, and Section \ref{sec:05}.  Moreover, Problems (i) and (ii) yield identical Vortex/Bubble shapes $\partial \Omega$, iff this conformal mapping identity holds.

In addition to allowing one to apply results that are valid within the bubble problem to the vortex problem and vice-versa, the mathematical analogy between Problems (i) and (ii) \textit{also} plays an essential role in our existence theory for the bubble problem, as will transpire within the proof of Theorem \ref{thm:63}. More precisely, our analysis of Problem (ii) gives rise to certain technical difficulties which are overcome by, mathematically, viewing the bubble as a hollow vortex using the aforementioned conformal mapping identity. Remarkably, this identity is true in the vortex problem when the parameters $\gamma, \hat{\alpha},$ and $\beta$ are appropriately constrained even though, as we shall see in the ensuing, there appears no natural reason for such conformal mapping identity to hold for the vortex problem (see Remark \ref{rem:2.3}). The unified mathematical formulation is what makes it possible to identify and establish this otherwise unexpected conformal-mapping identity.

The physical properties of solutions $(\psi, \Omega)$ are explained in the
companion paper \cite{CKT}, which also provides
numerical computions supplementing and extending the theoretical results of
the present article. Moreover, \cite{CKT} derives exact solutions to
Problems (i) and (ii) when $\gamma = 0$, which we shall recall in Section
\ref{sec:write-up2.3}. \\

The remainder of this article is organised as follows. In Section  \ref{s:002}, we introduce the conformal mapping framework that is employed in order to prove the existence of solutions to the Vortex problem.  In Section \ref{sec:write-up3}, we will prove the existence of solutions to the Vortex problem (i) for specified physical parameters $\gamma, \hat{\alpha},$ and $\beta$. We also discuss symmetry and analyticity properties of solutions to both, the Vortex problem (i) and Bubble problem (ii). In Section \ref{sec:write-up04} we will prove the existence of solutions to the Bubble problem, and formulate the aforementioned constraint $(\gamma, \alpha, \Gamma_0) \in \mathcal{C}_f$ which must necessarily be satisfied in order for the Bubble problem (ii) to admit a solution. Finally, in Section \ref{sec:05}, we discuss the conformal mapping identity which is a criterion for when the Vortex and Bubble shapes, respectively, are identical \textit{even though the flows exterior to the Vortex and Bubble regions differ} (as illustrated in Figure \ref{Fig.01}).

\section{Conformal mapping framework and main results}\label{s:002}

\subsection{Two mathematically analogous physical
problems}\label{s:write-up2.1}

In \cite{CKT}, it was also shown that, akin to
Crowdy-Tanveer \cite{CrowdyTanveer2025} problems of a translating bubble and a vortex in
linear strain, both Problems (i) and (ii) can be cast into a unifying
mathematical framework. This is achieved as follows. Let us denote by
$\mathcal{Z} (s)$ the arclength parameterisation of the simple closed curve
$\partial \Omega$, with arc-length increasing in the positive (anti-clockwise)
direction. As detailed in \cite{CKT}, it can be shown that in both, Problems (i) and (ii), the following
non-linear differential equation needs to be satisfied by $\mathcal{Z}$:\footnote{Derivation of \eqref{e:Seq01} from \eqref{e:1.1}-\eqref{e:1.2} for the vortex problem allows for $\pm d\bar{\mathcal{Z}}/ds$ on the left-hand-side of \eqref{e:Seq01}. However, there is no loss of generality in taking only the positive sign for otherwise we may change sign of the physical parameters, which will also yield \eqref{e:Seq01}.}
\begin{equation}
  \frac{d}{d s} \overline{\mathcal{Z}} (s) = i \gamma \overline{\mathcal{Z}}
  (s) + \xi (\mathcal{Z} (s)) . \label{e:Seq01}
\end{equation}
Within Equation \eqref{e:Seq01}, $\xi (z)$ is an unknown function that is
analytic within the region $\Omega$ exterior to the vortex/bubble, is
linked to the flow's velocity field, and has prescribed asymptotic behaviour (cf.\! \eqref{e:write-up02} and \eqref{e:02}).

\begin{notation}
	 Everywhere below, we use the following notation. If $f(z)$ is a meromorphic function, then we set $\bar{f}(z) = \overline{f(\bar{z})}$ where $\bar{z}$ is the usual complex conjugate.
\end{notation}

 As will transpire, \eqref{e:Seq01} is enough to determine both $\mathcal{Z}(s)$ and $\xi(z)$ for specified far-field behaviour of $\xi$.  In \cite{CKT}, it is
explained that the following properties of $\xi (z)$ are known a-priori.
\begin{enumerate}[(a)]
  \item For the vortex problem, the asymptotic condition \eqref{e:write-up02}
  must hold for $\xi = \xi_V$, for specified $\hat{\alpha}$, $\beta$ within a
  certain range, as discussed in Section \ref{sec:write-up3.4}.
  \item For the bubble-problem, using the streamline condition $\psi=0$ in the form $\overline{W}_{\bar{z}} \frac{d\bar{z}}{ds} ds = W_z \frac{dz}{ds} ds$ on $\partial \Omega$, the Bernoulli condition \eqref{e:1.6} implies that 
  \begin{equation}
    \xi (z) = - i \int^z\! \left(\frac{dW}{dz'}\right)^2\! d z' 
    \label{e:write-up03}
  \end{equation}
  where $W (z)$ is the aforementioned complex velocity potential associated to
  the bubble problem, see \cite{CKT, CrowdyTanveer2025}. The function $W (z)$ can be described explicity by means
  of a conformal map $Z (\zeta)$, cf. Equation \eqref{e:write-up11}.
\end{enumerate}
To summarise, we will use the following terminology in the sequel.

\begin{definition}
	\label{def:01}We say that a simple closed curve $\mathcal{Z} (s)$,
	parametrised according to arc-length in the positive direction, solves the
	vortex problem, if the Equation \eqref{e:Seq01} holds for some $\xi (z)$
	which is analytic in the region $\Omega$ exterior to $\mathcal{Z} (s)$, except at $\infty$ where the asymptotic expansion \eqref{e:write-up02} holds for specified
	$\hat{\alpha}$, $\beta$. We say that $\mathcal{Z} (s)$ solves the bubble
	problem if, on the other hand, $\mathcal{Z} (s)$ solves  \eqref{e:Seq01} with corresponding  $\xi (z)$ compatible with \eqref{e:write-up03}.
	In both cases, vortex and bubble, there are a-priori regularity assumptions
	imposed on $\mathcal{Z} (s)$, which we formulate more conveniently in terms of a conformal map in the following Section
	\ref{sec:write-up2.2}.
\end{definition}

The Equation \eqref{e:Seq01} enables us to introduce the aforementioned
conformal mapping framework, which we will use to study both, Problems (i) and
(ii).

\begin{remark}\label{rem:2.2}
		If $\mathcal{Z}(s)$ solves the vortex problem, then we obtain a solution to Problem (i) formulated in the introduction as follows. The streamfunction $\psi$ is obtained from $\xi= \xi_V$ via \eqref{e:01}, and we define $\Omega$ as the region exterior to the closed curve $\mathcal{Z}(s)$. Similarly, if $\mathcal{Z}(s)$ solves the bubble problem, we obtain a solution to Problem (ii) formulated in the introduction. 
\end{remark}

\begin{remark}\label{rem:2.3}
		As we will see in the following Section \ref{sec:write-up2.2}, the conformal mapping identity mentioned in the introduction is forced upon us solely by the requirement that, for the bubble problem, $\xi(z)$, as determined by \eqref{e:Seq01}, is consistent with \eqref{e:write-up03}. Since there is no complex velocity potential for the vortex problem, and hence no corresponding compatibility condition, it is surprising that the same conformal mapping identity \textit{also} holds for the vortex problem, once the physical parameters are constrained appropriately.
\end{remark}

\subsection{Conformal mapping framework}\label{sec:write-up2.2}

Following \cite{CrowdyTanveer2025}, and similar approaches to other related problems  \cite{Tanveer1996}, we study Problems (i) and (ii)
\ by means of a conformal mapping
\begin{equation}
	Z : \mathbb{D} \rightarrow \Omega \cup \{\infty\}, \quad \zeta \mapsto Z (\zeta)
\end{equation}
that maps the interior of the unit disc $\mathbb{D} \subset \mathbb{C},$ lying
in the complex $\zeta$-plane, onto the region $\Omega$ exterior to the
vortex/bubble in the complex $z$-plane, with $Z (0) = \infty$. Note that the existence of region $\Omega$ has not yet been established. Thus, all of the following computations are only formal, and we will verify their validity a-posteriori. For brevity, we will simply write $\Omega$ instead of $\Omega \cup \{\infty\}$ in the sequel and it is understood that $Z$ is always normalised according to $Z(0)=\infty$.

Let us describe the Laurent-series expansion of $Z (\zeta)$. Without loss of
generality, after using a rotational degree of freedom granted by the
Riemann-mapping theorem, a translational degree of freedom in the physical
$z$-plane, and appropriate length-scaling, we may assume that
\begin{equation}
  Z (\zeta) = - \zeta^{- 1} + \sum_{j = 2}^{\infty} c_j \zeta^{j - 1}
  \label{e:Seq03}.
\end{equation}
Next, define the
function $y (\zeta) = \sqrt{Z_{\zeta}}$ which, due to \eqref{e:Seq03}, admits
the representation\footnote{By specifying the Series \eqref{e:Seq04}, we have implicitly chosen a branch of the square root.}
\begin{equation}
  y (\zeta) = \zeta^{- 1} + \sum_{j = 2}^{\infty} a_j \zeta^{j - 1}
  \label{e:Seq04} .
\end{equation}
Note that for $0 < | \zeta | < 1$, $y (\zeta) \neq 0$ is well-defined since
$Z$ is conformal. Moreover, we can now formulate the following a-priori
regularity assumption imposed on the boundary $\partial \Omega$ of the fluid
region. 

\begin{Regularity}[A-priori regularity assumption on $\partial \Omega$] \label{A-priori}
	The series  \eqref{e:Seq04} is absolutely convergent and non-vanishing on $\partial \mathbb{D}$. Note that this implies that \eqref{e:Seq03} is absolutely convergent on $\partial \mathbb{D}$ (see \eqref{e:seq05}).
\end{Regularity}

Next, introduce the
function $\bar{K}$ via
\begin{equation}
  \xi (Z (\zeta)) = i \bar{K} (\zeta) . \label{e:write-up07}
\end{equation}
Then, as detailed in \cite{CrowdyTanveer2025}, the interface
condition \eqref{e:Seq01} may be rewritten as
\begin{equation}
  \zeta y (\zeta) = \gamma \bar{y} (\zeta^{- 1}) Z (\zeta) + K  (\zeta^{- 1})
  \bar{y} (\zeta^{- 1}) \quad \text{on} \quad \partial \mathbb{D} .
  \label{e:seq07}
\end{equation}
Let us recall some results presented in \cite{CKT} that
concern the complex-analytic structure of the function $K (\zeta)$.

\begin{lemma}
  \label{lem:write-up01}Let $\mathcal{Z}$ solve the vortex problem. Then, the
  far-field condition \eqref{e:write-up02} on $\xi (z)$ yields the following
  Laurent-series expansion.
  \begin{equation}
    \bar{K} (\zeta) = \bar{k}_{- 2} \zeta^{- 3} + k_0 \zeta^{- 1} + \sum_{j =
    1}^{\infty} \bar{k}_j \zeta^{j - 1} \label{e:seq08},
  \end{equation}
  where the coefficients $\bar{k}_{- 2}, \bar{k}_0, \bar{k}_1,$and $\bar{k}_2$
  are linked to the physical parameters $\hat{\alpha}, \beta, \delta$ and
  $\omega$ as follows.
  \begin{equation}
    \bar{k}_{- 2} = - \hat{\alpha}, \quad \bar{k}_0 = 3 \hat{\alpha} c_2 -
    \beta + \gamma, \quad \bar{k}_1 = 3 c_3 \hat{\alpha} + \delta, \quad
    \bar{k}_2 = 3 \hat{\alpha} c_4 - 3 c_{2 }^2 \hat{\alpha} + (\beta -
    \gamma) c_2 - \omega, \label{e:write-up10}
  \end{equation}
  with $c_2, c_3,$and $c_4$ being Laurent-series coefficients of $Z (\zeta)$ in \eqref{e:Seq03}.
\end{lemma}

Before we formulate the analogous lemma concerning the structure of $\bar{K}
(\zeta)$ in the context of the bubble problem, let us note that, for the
bubble-problem, the streamline condition
\[ \psi = 0 \quad \text{on} \quad \partial \Omega \]
implies that the complex velocity potential $W$ is completely determined as a
function of $\zeta$:
\begin{equation}
  W (Z (\zeta)) = \frac{\alpha }{2} (\zeta^2 + \zeta^{- 2}) + i
  \Gamma_0 \log (\zeta), \label{e:write-up11}
\end{equation}
see \cite{CKT}. Therefore, \eqref{e:write-up03} may be re-written as 
	\begin{align}
		\bar{K}(\zeta)= -  \int^{\zeta} \left(\frac{d W}{d \zeta}\right)^2 \frac{d\zeta}{Z_\zeta}. \label{e:2.11}
	\end{align}
 Using the representation \eqref{e:2.11} and the explicit formula \eqref{e:write-up11} to expand $\bar{K}$ in a Laurent series, and using the equation \eqref{e:seq07} to compare terms of like powers of $\zeta$ (akin to the procedure detailed in Section \ref{sec:write-up3.1}),  it can be
 shown that the following Lemma holds (cf.\! \cite{CKT}).

\begin{lemma} \label{lem:write-up02} Let $\mathcal{Z}$ solve the
  bubble problem. Then the expansion \eqref{e:seq08} for $\bar{K} (\zeta)$
  still holds, where now the coefficients $\bar{k}_{- 2}, \bar{k}_0$ and
  $\bar{k}_2$ are linked to the physical parameters $\alpha$ and $\Gamma_0$ as
  follows.
  \begin{equation}
    \bar{k}_{- 2} = \frac{\alpha^2}{3}, \quad \bar{k}_0 = - 2 a_2 \alpha^2 -
    2 i \Gamma_0 \alpha, \quad \bar{k}_2 = 2 \alpha^2 +
   \Gamma_0^2 - 4 i \Gamma_0 \alpha a_2 - (3 a_2^2 -
    2 a_4) \alpha^2 \label{e:write-up12}
  \end{equation}
  with $a_2,$ and $a_4$ being series coefficients of $y (\zeta)$ in
  \eqref{e:Seq04}.
\end{lemma}

\begin{remark}
	For the bubble problem, there is an additional constraint $a_3=0$ required in order that \eqref{e:02} reduces to \eqref{e:seq08}. But this is moot since, as it will transpire, for any solution to \eqref{e:seq07} we necessarily have $a_{2j+1} \equiv 0$ for all $j$, see Lemma \ref{lem:seq12}.
\end{remark}

As will turn out, for specified $k_{- 2}$, $k_0$, $\gamma$, by using \eqref{e:seq07},  the series $\{ a_j
\}_{j = 2}^{\infty}$, $\{ c_j \}_{j = 2}^{\infty}$, and $\{ k_j \}_{j =
1}^{\infty}$ are determined uniquely, as for the vortex problem, without
regard to the relation \eqref{e:2.11} (and hence, without regards to the relation \eqref{e:write-up03}), as we will see in Section \ref{sec:write-up3}. A similar strategy can also be employed to study solutions to the bubble problem (see Section \ref{sec:write-up3}).  Particularly, if the bubble problem
were to have a solution, the $\{ k_j \}_{j = 1}^{\infty}$ coefficients determined by employing such strategy must, according to \eqref{e:2.11}, satisfy the following \textit{conformal mapping identity}
\begin{equation}
  - 3 {\bar{k}_{- 2}}  \zeta^{- 4} - k_0 \zeta^{- 2} + \sum_{j =
  2}^{\infty} (j - 1) \bar{k}_j \zeta^{j - 2}  = \frac{d \bar{K}}{d \zeta} = -
  \left( \alpha \left(\zeta -\zeta^{-3}\right) + \frac{i \Gamma_0}{\zeta} \right)^2 \frac{1}{Z_{\zeta}}. \label{e:16}
\end{equation}
 In particular, the expressions for
$\bar{k}_2$ obtained from \eqref{e:write-up12} must be consistent with those
derived later, from \eqref{e:seq09}. This will lead to the aforementioned constraint on
$(\gamma, \alpha, \Gamma_0)$ which needs to be satisfied by any solution to
the bubble problem. The constraint is discussed in detail within Sections \ref{sec:4.5}, \ref{sec:4.6}, and  \ref{sec:05}.

Even though Lemma \ref{lem:write-up02} gives no information on the coefficient
$\bar{k}_1$, we will see in Section \ref{sec:write-up3.1} that both,
$\bar{k}_1$ and $\bar{k}_2$ can be recovered from $\bar{k}_{- 2}$ and $\bar{k}_0$ in conjunction with the coefficients $\{ c_j \}_{j = 2}^{\infty}$ and $\{ a_j \}_{j =2}^{\infty}$ of the Laurent-series expansions of $Z (\zeta)$ and $y (\zeta)$ through \eqref{e:seq07}. 

Next, following \cite{CKT}, in order to formulate a common mathematical framework for both vortex and bubble, we introduce the parameters
$\lambda$ and $\sigma$ as follows.
\begin{align*}
	\text{Vortex Problem (i):}  ~ & (\lambda, \sigma) = (- \hat{\alpha}, -
	\bar{\beta} + \gamma) ; ~ \text{Bubble Problem (ii):} ~ (\lambda,
	\sigma) = \left( \frac{\alpha^2}{3} , 2 i \Gamma_0 \alpha \right) 
\end{align*}
The sequences $\{ c_j \}_{j = 2}^{\infty}$ and $\{ a_j \}_{j =
	2}^{\infty}$ of series coefficients are linked by the relation
\begin{equation}
	c_j = \frac{2 a_j}{j - 1} + \frac{1}{j - 1} \sum_{m = 2}^{j - 2} a_m a_{j -
		m}, \quad j \geq 2 \label{e:seq05},
\end{equation}
where the sum is
understood to vanish for $j = 2, 3$. Thus,
we know that $c_2 = 2 a_2$ and therefore
\begin{equation}
	(k_{- 2}, k_0) = \left( \lambda, - 6 \lambda\bar{a}_2 + \sigma \right)
	\label{e:seq015}
\end{equation}
for \textit{both} physical problems.  In summary, we have just shown that any
solution to either the vortex or the bubble problem solves, via the conformal
map $Z (\zeta)$, the following ``Problem (a)''.

\begin{definition}[Problem (a)]
  \label{def:Problem3 a}Let $\lambda > 0$, and $\sigma \in \mathbb{C}$ be
  constant. We say that a map
  \begin{equation}
    Z(\zeta) = - \zeta^{- 1} +
    \sum_{j = 2}^{\infty} c_j \zeta^{j - 1} \label{e:Seq003}
  \end{equation}
  solves ``Problem (a)'' if the a-priori regularity assumption \ref{A-priori} is satisfied by $y (\zeta) = \sqrt{Z_{\zeta} (\zeta)}$, and $Z (\zeta)$ solves
  \begin{equation}
    \zeta y (\zeta) = \gamma \bar{y} (\zeta^{- 1}) Z (\zeta) + K  (\zeta^{-
    1}) \bar{y} (\zeta^{- 1}), \quad \text{on} \quad | \zeta | = 1.
    \label{e:write-up16}
  \end{equation}
  Here,  $K$ is analytic and single valued within $\mathbb{D} \setminus \{0\}$
with Laurent-series of the form
  \begin{equation}
    K (\zeta) = k_{- 2} \zeta^{- 3} + k_0 \zeta^{- 1} + \sum_{j = 1}^{\infty}
    k_j \zeta^{j - 1}. \label{e:write-up17}
  \end{equation}
The coefficients $k_{- 2} $ and $k_0$ are given by \eqref{e:seq015}.
  Exact rational solutions to Problem (a) for $\gamma=0$ are reported in
  \cite{CKT}, and will be recalled in Section
  \ref{sec:write-up2.3}.
\end{definition}

In Section \ref{sec:write-up3.4}, we will discuss univalency of solutions to
Problem (a). The notion ``Problem (a)'' is mostly due to Figure
\ref{Fig:02}. Let us explain what we can, and cannot, gain by
studying it.
\begin{enumerate}
		\item Let us assume that $Z(\zeta)$ solves Problem
	(a), and that $\gamma, \lambda,$ and $\sigma$ are in a range where $Z (\zeta)$ is	univalent. In particular, corresponding $\omega$ is necessarily real, see Equation \eqref{e:CircReal}.  Then, the curve $\mathcal{Z} (s)$ describing the
	(counter-clockwise) arc-length parametrisation of $\partial \Omega = Z
	(\partial \mathbb{D})$ solves the vortex Problem for we may set $\xi (z) =
	i \bar{K} (Z^{- 1} (z))$. 
	
	\quad The sought-after pair $(\psi, \Omega)$ is then obtained according to the procedure detailed in Remark \ref{rem:2.2}.

  \item Solutions to Problem (a) generally \textit{do not} yield solutions
  to the Bubble problem, even when univalent. This is because Equation \eqref{e:write-up03}
  generally needs not hold for $\xi (z) = i \bar{K} ( Z^{- 1}
  (z))$. In fact, as we will see,  the
  bubble problem admits a solution if, and only if, the physical parameters
  $\gamma, \alpha,$ and $\Gamma_0$ are constrained to lie on an analytic surface
  $\mathcal{C}_f$, see Sections \ref{sec:4.5} and \ref{sec:4.6}. However, in Section
  \ref{sec:write-up3}, we will show that Problem (a) admits solutions without
  any such constraint imposed on $\gamma, \lambda$, and $\sigma$. Thus without the constraint $(\gamma, \alpha, \Gamma_0) \in \mathcal{C}_f$, solutions to Problem (a) do not correspond to solutions of the Bubble
  problem (ii).
  
  \item Studying Problem (a) will unveil symmetry and boundary analyticity
  properties of solutions to both, the vortex and the bubble problem, see
  Corollary \ref{cor:write-up18} and Lemma \ref{lem:seq20}. Particularly, we
  will see that $\partial \Omega$ describes a real-analytic curve.
  
  \item Given the constraint $(\gamma, \alpha, \Gamma_0) \in \mathcal{C}_f$
  described in Pt.\! 2, consistency of the two expressions \eqref{e:write-up03}
  and \eqref{e:write-up17} for $K (\zeta)$ will yield the conformal mapping identity \eqref{e:16}. This identity is discussed in detail in Section \ref{sec:05}, and transcends the bubble problem in the sense, that when parameters $\alpha,$ and $\Gamma_0$ are, by means of $\lambda$  and $\sigma$, replaced by equivalent vortex parameters $\hat{\alpha}$ and $\beta$, the same identity holds for the vortex solution.
\end{enumerate}
So, to summarise, Problem (a) allows us to understand the complex-analytic
structure of solutions to the vortex and bubble problem. It may be used to
prove existence of solutions to the vortex problem, though we will require an
existence proof of the bubble problem that does not rely on the ``Problem (a)'' formulation, since it is a-priori unclear that the $\bar{K}$ obtained through \eqref{e:write-up16} is consistent with \eqref{e:16}.

Indeed, in order to prove the existence of solutions to the bubble problem,
we will rely on the following ``Problem (b)'' formulation.

\begin{definition}[Problem (b)]
  \label{def:Problem3 b} Let $\alpha > 0$, and $\Gamma_0 \in 
  \mathbb{R}$ be constant. We say that a map
  \begin{equation}
    Z(\zeta)  =  - \zeta^{- 1} +
    \sum_{j = 2}^{\infty} c_j \zeta^{j - 1} \label{e:Seq0003}
  \end{equation}
  solves ``Problem (b)'' if the a-priori regularity assumption \ref{A-priori} is satisfied by $y(\zeta) = \sqrt{Z_{\zeta}}$, and $y(\zeta)$
  satisfies
  \begin{equation}
    | y (\zeta) |^2 \left( 1 + 2 \Re \left( \zeta \frac{y_{\zeta} (\zeta)}{y
    (\zeta)} \right) \right) + \left| \frac{d W}{d \zeta} \right|^2 = \gamma |
    y (\zeta) |^4, \quad \text{on} \quad | \zeta | = 1. \label{e:write-up75}
  \end{equation}
  Here, $W (\zeta)$ is given by \eqref{e:write-up11}.
\end{definition}

The fact that solutions to Problem (ii) yield solutions to Problem (b) can be proved by verbatim adaptation of Tanveer's arguments \cite{Tanveer1996} concerning the derivation of \eqref{e:write-up75} in the context of a steadily translating inviscid bubble, so we omit the computational details. Any univalent solution to Problem
(b) may be integrated to yield a solution to Problem (a) with $(\lambda,\sigma)$ identified to be $(\frac{\alpha^2}{3}, 2 i \Gamma_0 \alpha)$, but solutions to
Problem (a) only give rise to solutions of Problem (b) if the 
conformal mapping identity \eqref{e:16} is satisfied. If Problem (b) is
solved and function $Z(\zeta)$ is univalent, then the explicit formula for $W
(\zeta)$ in \eqref{e:write-up11} gives rise to the sought-after streamfunction
$\psi$ in Problem (ii). 

To conclude this subsection, let us summarise the present article's structure
and results in the diagrammatic scheme shown in Figure \ref{Fig:02}. As mentioned previously, the companion paper \cite{CKT}
reports exact solutions to Problems (a) and (b) when $\gamma = 0$. We will
present these solutions in the following Subsection \ref{sec:write-up2.3}.

\begin{figure}[h!]
	\centering
	\includegraphics[width=1\textwidth]{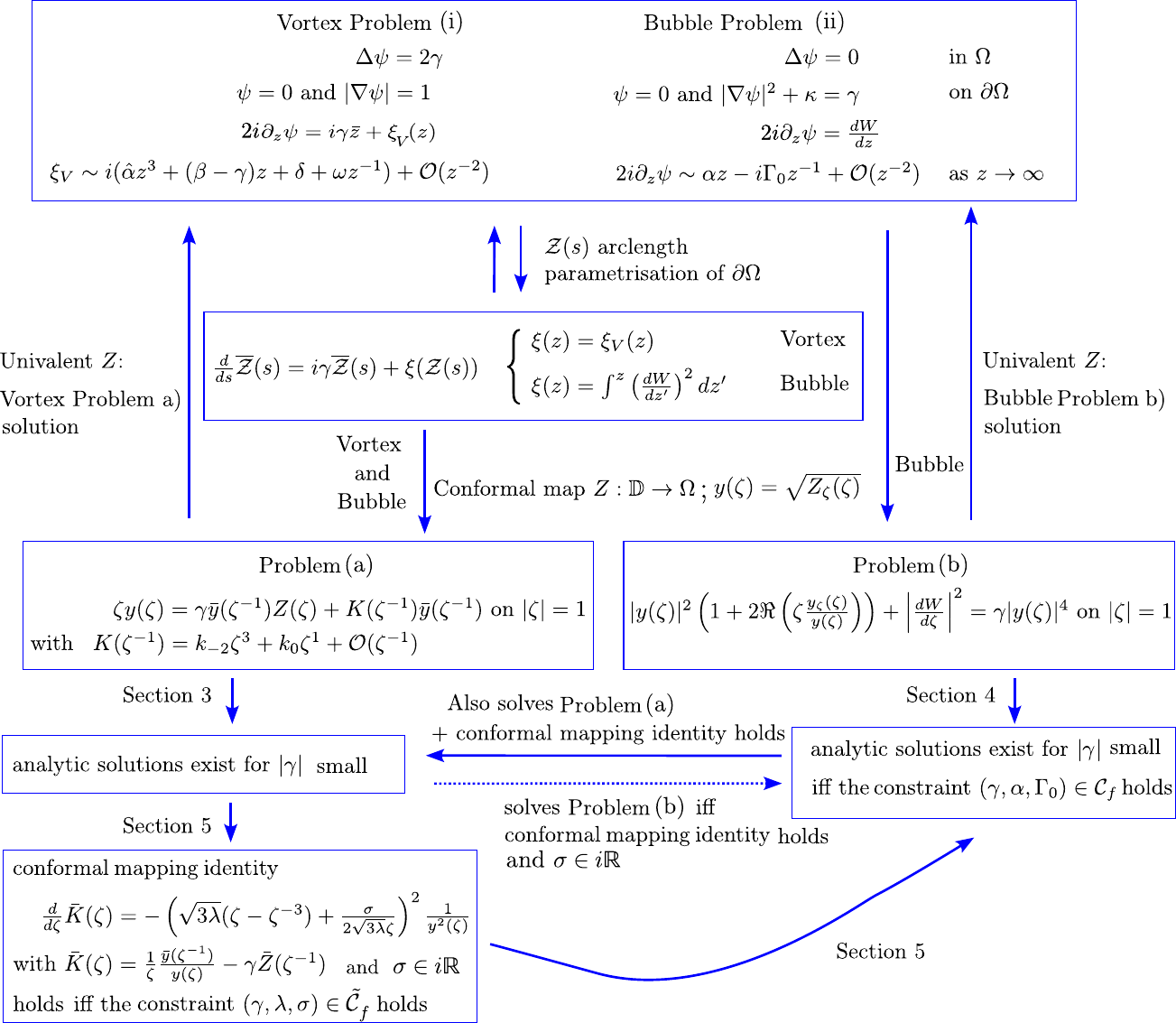}
		\caption{\small Diagrammatic overview of the article's structure and main results. Here, $(\gamma, \lambda, \sigma) \in \tilde{\mathcal{C}}_f$ indicates that, for $\alpha =\sqrt{3 \lambda}$ and $\Gamma_0 =\sigma/(2 i \sqrt{3 \lambda})$, we have $(\alpha, \Gamma_0) \in \mathcal{C}_f$ (cf.\! Corollary \ref{cor:5.5}).}
		\label{Fig:02}
\end{figure}

\subsection{Exact solutions for $\gamma = 0$}\label{sec:write-up2.3}

In the special case when $\gamma = 0$, exact solutions to Problem (a) are
available, and these solutions have been recently reported in
\cite{CKT}. Namely, we have
\begin{equation}
  y_0 (\zeta) = \frac{1}{\zeta} + \left( \frac{5 \lambda \bar{\sigma} -
  \sigma}{25 \lambda^2 - 1} \right) \zeta + \lambda \zeta^3 . \label{e:18}
\end{equation}
Thus, $a_0 = 1, a_2 = \left( \frac{5 \lambda \bar{\sigma} - \sigma}{25
\lambda^2 - 1} \right)$, $a_4 = \lambda$, and all the other coefficients
vanish. From the relation \eqref{e:seq05}, the coefficients of $Z (\zeta)$ can
be computed and we find

\begin{equation}
  Z_0 (\zeta) = - \frac{1}{\zeta} + \left( \frac{10 \lambda \bar{\sigma} - 2
  \sigma}{25 \lambda^2 - 1} \right) \zeta  + \frac{1}{3}  \left(2 \lambda +
  \left( \frac{5 \lambda \bar{\sigma} - \sigma}{25 \lambda^2 - 1} \right)^2 \right)
  \zeta^3  + \frac{2}{5} \left( \frac{5 \lambda \bar{\sigma} - \sigma}{25
  \lambda^2 - 1} \right) \lambda \zeta^5 + \frac{\lambda^2}{7} \zeta^7 .
  \label{e:19}
\end{equation}
The Equations \eqref{e:18} and \eqref{e:19} are valid for both, the vortex and
the bubble problem. The function $\bar{K} (\zeta)$ can be computed directly
from \eqref{e:18} via the relation
\begin{equation}
  \bar{K}_0 (\zeta) = \frac{\bar{y}_0 (\zeta^{- 1})}{\zeta y_0 (\zeta)}
\end{equation}
which gives
\begin{equation}
  \bar{K}_0 (\zeta) = \zeta^{- 3} \left( \frac{\lambda + \bar{a}_2 \zeta^2 +
  \zeta^4}{1 + a_2 \zeta^2 + \lambda \zeta^4} \right) .
\end{equation}
In the case of the bubble, validity of these solutions requires $\Gamma_0$ and
$\alpha$ to be constrained according to
\begin{equation}
 \Gamma_0^2 = \frac{(3 - 5 \alpha^2)^2 (3 - 7 \alpha^2)}{9 (25
  \alpha^2 + 3)} . \label{e:C-15}
\end{equation}
The relation \eqref{e:C-15} describes a curve in the real $(\alpha,\Gamma_0)$-plane, and is obtained from equating the two expressions for
$k_2$ given by \eqref{e:write-up12} and \eqref{e:seq09}, which must hold for
any solution to the bubble problem. The curve described by \eqref{e:C-15} is shown in Figure \ref{Fig:b}, left. 

Note that in the context of the vortex
problem, equating these two expressions for $k_2$ yields
\begin{equation}
	\omega = - 1 - \hat{\alpha}^2 - | a_2 |^2
\end{equation}
which, , in particular, implies that $\omega$ is real, as required, and it also illustrates that
no constraint akin to \eqref{e:C-15} is required in order for the vortex
problem to admit solutions. Now, when the constraint \eqref{e:C-15} is
satisfied, it can be checked directly that
\begin{equation}
  \frac{d}{d \zeta} \bar{K}_0  = - \left( \sqrt{3 \lambda} \left( \zeta -
  \frac{1}{\zeta^3} \right) + \frac{\sigma}{2 \sqrt{3 \lambda} \zeta}
  \right)^2
  \frac{1}{dZ_{0}/d\zeta} \label{e:29}
\end{equation}
which is  exactly the \textit{conformal mapping identity} given in \eqref{e:16} for the case $\gamma=0$ (see also the formulation of  Theorem
\ref{seq:thmConformal}). Moreover, if we replace $\alpha^2/3$ by $-\hat{\alpha}$ and $2 i \Gamma_0 \alpha$ by $-\bar{\beta}$ in \eqref{e:C-15}, then the conformal mapping identity \eqref{e:29} holds equally for the Vortex problem.  In addition to the conformal mapping identity
\eqref{e:29}, let us note the following properties of $Z_0$.
\begin{enumerate}
  \item $c_{2 j + 1} \equiv 0$ for all $j$.
  
  \item If $\sigma \in i \mathbb{R}$, then $c_{4 j}$ is real whereas $c_{4 j -
  2}$ is purely imaginary for all $j$.
  
  \item Solutions to the bubble problem exist only if the constraint
  \eqref{e:C-15} holds. 
\end{enumerate}
In the following Sections, we will show that Problems (a) and (b) admit solutions
even when $\gamma \neq 0$, and the above properties exemplified by the
explicit $\gamma = 0$ solutions will be shown to hold generally for $\gamma \neq0$ solutions
to Problem (a) (see Corollary \ref{cor:write-up18}, Lemma \ref{lem:seq20}, and Theorem \ref{theorem:bubble}). Particularly, the symmetry
properties stated in 2.\! and 3.\! above are valid for any solution to 
Problem (b), whereas 1.\! is valid for any solution to Problem (a), i.e., for any solution to either the bubble or the vortex problem. 
For $\gamma \neq 0$, instead of \eqref{e:C-15}, the constraint
discussed in Sections \ref{sec:4.5} and \ref{sec:05} needs to be satisfied by $(\gamma, \alpha, \Gamma_0)$ for the bubble solution.

\section{On the existence of solutions to Problem (a) }\label{sec:write-up3}

\subsection{Reformulation in sequence space}\label{sec:write-up3.1}

Let us assume that $Z(\zeta)$ is a solution to Problem (a). We will
prove the existence of such solutions by reformulating Problem (a) as 
fixed-point problem in the sequence space $\ell^1$. To achieve this, let us, akin to the discussion of Section \ref{sec:write-up2.2}, formally study the properties of $Z(\zeta)$.

To ensure more compactly written formulae, set
\begin{equation}
  a_0 = 1, ~ c_0 = - 1, ~ a_1 = 0, ~ c_1 = 0, ~ \text{ and } ~ k_{- 1} = 0.
\end{equation}
Then, the Equations \eqref{e:Seq03}, \eqref{e:Seq04}, and \eqref{e:write-up17}
for $Z (\zeta)$, $y (\zeta)$, and $K (\zeta)$ become rewritten as
\begin{equation}
  Z (\zeta) = \sum_{j = 0}^{\infty} c_j \zeta^{j - 1}, \quad y (\zeta) =
  \sum_{j = 0}^{\infty} a_j \zeta^{j - 1}, \quad \text{and} \quad K (\zeta) =
  \sum_{j = - 2}^{\infty} k_j \zeta^{j - 1} . \label{e:write-up19}
\end{equation}
Recall that the series coefficients $\{ c_j \}_{j = 2}^{\infty}$ and $\{ a_j
\}_{j = 2}^{\infty}$ are linked according to \eqref{e:seq05}. Now,
substituting the expansions \eqref{e:write-up19} into \eqref{e:write-up16} and
comparing the coefficients of $\zeta^j$ for $j \geq 0$ gives
\begin{equation}
  0 = \gamma \sum_{m = 2}^{\infty} c_{m + 1} \bar{a}_m + k_{- 2} \bar{a}_3 +
  k_1  \qquad \text{from} \quad j = 0, \label{e:seq10}
\end{equation}
and
\begin{equation}
  1 = - \gamma + \gamma \sum_{m = 2}^{\infty} c_m \bar{a}_m + \bar{a}_4 k_{-
  2} + \bar{a}_2 k_0 + k_2, \qquad \text{from} \quad j = 1. \label{e:seq09}
\end{equation}
The equation \eqref{e:seq10} can be used to determine $k_1$, once the series
$\{ c_j \}_{j = 3}^{\infty}$ and $\{ a_j \}_{j = 2}^{\infty}$ are known.
Equation \eqref{e:seq09}, on the other hand, can be used to determine $k_2$
once $\{ c_j \}_{j = 2}^{\infty}$ and $\{ a_j \}_{j = 2}^{\infty}$ are known.
Comparing the remaining coefficients of $\zeta^j$ for $j$ positive gives
\begin{equation}
  a_j = \gamma \sum_{m = 0}^{\infty} c_{m + j} \bar{a}_m + \delta_{4, j} k_{-
  2} + \delta_{2, j} (k_{- 2} \bar{a}_2 + k_0) \label{e:seq11}, \quad j \geq 2
\end{equation}
where $\delta_{m j} = 0$ if $m \neq j$ and $\delta_{j j} = 1$ is the Kronecker
delta.

Finally, we note that by matching negative powers of $\zeta^{- j}$ within
\eqref{e:seq07} we obtain a recurrence relation for determining $k_{j + 2}$
for $j \geq 1$. Namely, we find
\begin{equation}
  k_{j + 2} = - \gamma \left( \sum_{m = 2}^{\infty} c_m \bar{a}_{m + j} -
  \bar{a}_j \right) - \sum_{m = 2}^{j + 4} k_{j + 2 - m} \bar{a}_m, \quad
  \text{for} \quad j \geq 1. \label{e:seq12}
\end{equation}
This gives the following Lemma.

\begin{lemma}
  \label{lem:write-up06}Let $Z(\zeta)$ be a solution
  to Problem (a). Then, once \eqref{e:seq05} and \eqref{e:seq11} are solved for
  $\{ a_j \}_{j = 2}^{\infty}$ and $\{ c_j \}_{j = 2}^{\infty},$ the series
  coefficients $\{ k_j \}_{j = 1}^{\infty}$ of $K (\zeta)$ are fully
  determined by \eqref{e:seq10}, \eqref{e:seq09}, and \eqref{e:seq12}. 
\end{lemma}

Lemma \ref{lem:write-up06} suggests that solving \eqref{e:seq05} and
\eqref{e:seq11} for $\{ a_j \}_{j = 2}^{\infty}$ and $\{ c_j \}_{j =
2}^{\infty}$ should suffice in order to obtain solutions to Problem (a). As we
shall see, according to Proposition \ref{prop:write-up18}, this is
indeed the case. That is, in order to prove the existence of solutions to
Problem (a), we need to show that the Equation \eqref{e:seq11} admits a solution
$\{ a_j \}_{j = 2}^{\infty}$ whenever $\{ c_j \}_{j = 2}^{\infty}$ is given by
\eqref{e:seq05}. To achieve this, we will follow the methodology of \cite{CrowdyTanveer2025} and reformulate the problem of solving \eqref{e:seq11} subject
to \eqref{e:seq05} as a fixed-point problem within a weighted $\ell^1$ space. It is defined as follows.

\begin{definition}
  \label{def:seq01}Let $\boldsymbol{a} = (a_2, a_3, a_4, \ldots)$ and
  $\boldsymbol{c} = (c_2, c_3, c_4, \ldots)$. Let $\rho \in (0, 1]$ and let
  $\ell^1_{\rho}$ denote the sequence-space $\ell^1$ equipped with the
  weighted norm $\| \cdot \|_{\rho}$ according to
  \begin{equation}
    \| \boldsymbol{a} \|_{\rho} = \sum_{j = 2}^{\infty} \rho^{- j} | a_j | .
    \label{e:seq13}
  \end{equation}
  We denote by $\mathcal{B}_r (\boldsymbol{a}) \subset \ell^1_{\rho}$ the ball
  of radius $r$ centred at $\boldsymbol{a} \in \ell^1_{\rho}$.
\end{definition}

\begin{notation}
	Let $\mathcal{G}: \ell^1_{\rho} \to \ell^1_{\rho}$. For $\boldsymbol{a} \in \ell^1_{\rho} $, we will write $\mathcal{G}[\boldsymbol{a}] = \left(\mathcal{G}[\boldsymbol{a}]_1, \mathcal{G}[\boldsymbol{a}]_2,\mathcal{G}[\boldsymbol{a}]_3,...\right)$ i.e., the $j^{\text{th}}$ component of the sequence $\mathcal{G}[\boldsymbol{a}]$ will be denoted by $\mathcal{G}[\boldsymbol{a}]_j$.
\end{notation}

\begin{definition}
  \label{def:seq02} Let $\mathcal{M} : \ell^1_{\rho} \rightarrow
  \ell^1_{\rho}$, $\Lambda : \ell^1_{\rho} \rightarrow \ell^1_{\rho}$, and
  $\mathcal{T} : \ell^1_{\rho} \rightarrow \ell^1_{\rho}$ be the operators
  defined by
  \begin{equation}
    \mathcal{M} [\boldsymbol{a}]_j = \frac{1}{j - 1} \sum_{m = 2}^{j - 2} a_m
    a_{j - m}, \quad \Lambda [\boldsymbol{a}]_j = \frac{2}{j - 1} a_j \quad
    \text{and\quad$\mathcal{T} [\boldsymbol{a}] = \Lambda [\boldsymbol{a}] +
    \mathcal{M} [\boldsymbol{a}]$} \label{e:seq14} .
  \end{equation}
  Define, moreover, the operator $\mathcal{N} : \ell^1_{\rho} \times
  \ell^1_{\rho} \rightarrow \ell^1_{\rho}$ as
  \begin{equation}
    \mathcal{N} [\boldsymbol{c}, \boldsymbol{a}]_j = \sum_{m = 2}^{\infty} c_{m +
    j} \bar{a}_m \label{e:seq15} .
  \end{equation}
\end{definition}

The operators $\mathcal{M}, \Lambda, \mathcal{T}$, and $\mathcal{N}$ were
introduced in \cite{CrowdyTanveer2025}, wherein it was also proved that
they admit the mapping properties claimed within Definition \ref{def:seq02}.
More precisely, we have the following Lemma.

\begin{lemma}[\cite{CrowdyTanveer2025}, Section 3(i)]
  \label{lem:seq03}Let $\rho \in (0, 1]$ and $\boldsymbol{a }^{1, 2},
  \boldsymbol{c}^{1, 2} \in \ell^1_{\rho}$. The following inequalities are true.
  \begin{equation}
    \| \mathcal{N} [\boldsymbol{a}^1, \boldsymbol{c}^1] - \mathcal{N}
    [\boldsymbol{a}^2, \boldsymbol{c}^2] \|_{\rho} \leq \| \boldsymbol{a}^1
    \|_{\rho} \| \boldsymbol{c}^1 - \boldsymbol{c}^2 \|_{\rho} + \| \boldsymbol{c}^2
    \|_{\rho} \| \boldsymbol{a}^1 - \boldsymbol{a}^2 \|_{\rho} \label{e:seq16}
  \end{equation}
  \begin{equation}
    \| \mathcal{T} [\boldsymbol{a}^1] - \mathcal{T} [\boldsymbol{} \boldsymbol{a}^2]
    \|_{\rho} \leq \left( 2 + \frac{1}{3} (\| \boldsymbol{a}^1 \|_{\rho} + \|
    \boldsymbol{a}^2 \|_{\rho}) \right) \| \boldsymbol{a}^1 - \boldsymbol{a}^2
    \|_{\rho} . \label{e:seq18}
  \end{equation}
  Particularly, we have
  \begin{equation}
    \| \mathcal{N} [\boldsymbol{c}, \boldsymbol{a}] \|_{\rho} \leq \| \boldsymbol{}
    \boldsymbol{c} \|_{\rho} \| \boldsymbol{a} \|_{\rho}, \quad \text{and} \quad
    \| \mathcal{T} [\boldsymbol{a}] \|_{\rho} \leq 2 \| \boldsymbol{a} \|_{\rho} +
    \frac{1}{3} \| \boldsymbol{a} \|_{\rho}^2 .
  \end{equation}
\end{lemma}

Using the operators $\mathcal{T}$ and $\mathcal{N}$, as well as \
\eqref{e:seq015}, we can rewrite the Equations \eqref{e:seq05} and
\eqref{e:seq11} as follows.
\begin{equation}
  \boldsymbol{c} = \mathcal{T} [\boldsymbol{a}] \quad \text{and} \quad
  \boldsymbol{a} - \gamma \boldsymbol{c} = \gamma \mathcal{N} [\boldsymbol{c},
  \boldsymbol{a}] + \phi [\boldsymbol{a}] . \label{e:seq0024}
\end{equation}
where $\phi [\boldsymbol{a}]$ is the series \
\begin{equation}
  \phi [\boldsymbol{a}]_j = \lambda \delta_{j 4} + (- 5 \lambda \bar{a}_2 +
  \sigma) \delta_{j 2}, \quad j \geq 2. \label{e:seq26}
\end{equation}
\begin{remark} \label{rem:11} Comparing \eqref{e:seq0024} with the equations given in
\cite{CrowdyTanveer2025}, Section 3(h), we find that the 
  difference between ``Problem 2'' formulated within
	\cite{CrowdyTanveer2025} and  Problem (a) considered in the present
  article is the form of the forcing due to $\phi [\boldsymbol{a}]$. Indeed,
  within \cite{CrowdyTanveer2025}, instead of \eqref{e:seq0024}, one
  considers the equation $\boldsymbol{a} - \gamma \boldsymbol{c} = \gamma
  \mathcal{N} [\boldsymbol{c}, \boldsymbol{a}] + b \delta_{j 2}$ subject to
  $\boldsymbol{c} = \mathcal{T} [\boldsymbol{a}]$. Here $b$ is a parameter that,
  akin to $(\lambda, \sigma)$, is linked to the physical parameters of the
  problems considered in \cite{CrowdyTanveer2025}.
\end{remark}

The form of the ``forcing'' term $\phi [\boldsymbol{a}]$ makes it difficult to
use \eqref{e:seq0024} directly in order to prove the existence of solutions to
Equation \eqref{e:seq11}. This is because, in some sense, we will have no
good control over the term $- 5 \lambda\bar{a}_2$ without compromising
on the permitted range of $\lambda$ (cf. proof of Lemma \ref{lem:seq07}). So,
let us reformulate the Equation \eqref{e:seq0024}.

\begin{definition}
  Let $\boldsymbol{a} \in \ell^1_{\rho}$. We let $ \iota : \ell^1_{\rho}
  \rightarrow \ell^1_{\rho}$ be the map defined by
  \begin{equation}
    \iota [\boldsymbol{a}]_j = a_j + 5 \lambda \bar{a}_2 \delta_{j 2}, \quad j
    \geq 2. \label{e:seq27}
  \end{equation}
  Moreover, we let $\Phi \in \ell^1_{\rho}$ be the series $\Phi_j = \lambda
  \delta_{j 4} + \sigma \delta_{2 j}$.
\end{definition}

Using these definitions, we find that $\boldsymbol{a}$ solves \eqref{e:seq11}
with $\boldsymbol{c} = \mathcal{T} [\boldsymbol{a}]$ if, and only if,
\begin{equation}
  \iota [\boldsymbol{a}] = \gamma \mathcal{N} [\mathcal{T} [\boldsymbol{a}],
  \boldsymbol{a}] + \gamma \mathcal{T} [\boldsymbol{a}] + \Phi .
\end{equation}
However, we may also write $\boldsymbol{a}$ in terms of $\iota [\boldsymbol{a}]$
if $\lambda \neq \frac{1}{5}$. Namely, upon solving \eqref{e:seq27} for $a_j$,
we find
\begin{equation}
  a_j = \iota [\boldsymbol{a}]_j (1 - \delta_{j 2}) + \left( \frac{\iota
  [\boldsymbol{a}]_2}{1 - 25 \lambda^2} - \frac{5 \lambda}{1 - 25 \lambda^2}
  \overline{\iota[\boldsymbol{a}]}_2 \right) \delta_{j 2}, \quad j \geq 2.
\end{equation}
This gives the following Lemma.

\begin{lemma}
  Let $\lambda \neq \frac{1}{5}$. The map $\tau : \ell^1_{\rho} \rightarrow
  \ell^1_{\rho}$ defined via
  \begin{equation}
    \tau [\boldsymbol{b}]_j = b_j (1 - \delta_{j 2}) + \left( \frac{b_2}{1 - 25
    \lambda^2} - \frac{5 \lambda}{1 - 25 \lambda^2} \overline{b}_2 \right)
    \delta_{j 2}, \quad j \geq 2
  \end{equation}
  is inverse to $\iota$. 
\end{lemma}

Thus, for $\lambda \neq \frac{1}{5}$, we have just proved that $\boldsymbol{a}$
solves \eqref{e:seq11} with $\boldsymbol{c} = \mathcal{T} [\boldsymbol{a}]$ if,
and only if,
\begin{equation}
  \iota [\boldsymbol{a}] = \gamma \mathcal{N} [\mathcal{T} [ \boldsymbol{a}], \boldsymbol{a}] + \gamma \mathcal{T} [
  \boldsymbol{a}] + \Phi .
\end{equation}
This gives the following Lemma.

\begin{lemma}
  \label{prop:seq07}Let $\lambda \neq \frac{1}{5}$, and let $\mathcal{R} :
  \ell^1_{\rho} \rightarrow \ell^1_{\rho}$ be the operator defined by
  \begin{equation}
    \mathcal{R}  [\boldsymbol{b}] = \gamma \mathcal{N} [\mathcal{T} [\tau [\boldsymbol{b}]], \tau [\boldsymbol{b}]] + \gamma \mathcal{T} [\tau
    [\boldsymbol{b}]] + \Phi . \label{e:write-up38}
  \end{equation}
  Then, $\boldsymbol{a}$ solves \eqref{e:seq11} with $\boldsymbol{c} = \mathcal{T}
  [\boldsymbol{a}]$ if, and only if, $\boldsymbol{b} = \iota [\boldsymbol{a}]$ is a
  fixed point of the operator $\mathcal{R}$.
\end{lemma}

We will prove the existence of such fixed point in Section
\ref{sec:write-up3.2}. First, let us discuss why this existence is sufficient
in order to show that Problem (a) has a solution. Since $\iota [\boldsymbol{a}] = \mathcal{R} [\iota
[\boldsymbol{a}]]$,
\begin{equation}
  \boldsymbol{a} = \gamma \mathcal{N} [\mathcal{T} [\boldsymbol{a}], \boldsymbol{a}]
  + \gamma \mathcal{T} [\boldsymbol{a}] + \phi [\boldsymbol{a}]
\end{equation}
so $\boldsymbol{a}$ is a fixed point of the operator
\begin{equation}
  \widehat{\mathcal{R}} [\boldsymbol{a}] = \gamma \mathcal{N} [\mathcal{T}
  [\boldsymbol{a}], \boldsymbol{a}] + \gamma \mathcal{T} [\boldsymbol{a}] + \phi
  [\boldsymbol{a}],
\end{equation}
with $\phi [\boldsymbol{a}] $ given by \eqref{e:seq26}. In analogy with
\cite{CrowdyTanveer2025}, Lemma 3.14, we have the following result.

\begin{lemma}
  \label{lem:seq12}Let $\tilde{\boldsymbol{a}} = (a_2, - a_3, a_4, - a_5,
  a_6, \ldots) .$ If $\widehat{\mathcal{R}} [\boldsymbol{a}] = \boldsymbol{a}$,
  then we also have $\widehat{\mathcal{R}} [\tilde{\boldsymbol{a}}] =
  \tilde{\boldsymbol{a}}$.
\end{lemma}

\begin{proof}
  Exactly as in \cite{CrowdyTanveer2025}, Lemma 3.14, we obtain
  \begin{equation}
    \gamma \mathcal{N} [\mathcal{T} [\tilde{\boldsymbol{a}} \boldsymbol{}],
    \tilde{\boldsymbol{a}}]_{2 j}  + \gamma \mathcal{T}
    [\tilde{\boldsymbol{a}}]_{2 j} = \gamma \mathcal{N} [\mathcal{T}
    [\boldsymbol{a} \boldsymbol{}], \boldsymbol{a}]_{2 j}  + \gamma \mathcal{T}
    [\boldsymbol{a}]_{2 j}
  \end{equation}
  and
  \begin{equation}
    \gamma \mathcal{N} [\mathcal{T} [\tilde{\boldsymbol{a}} \boldsymbol{}],
   \tilde{\boldsymbol{a}}]_{2 j +1}  + \gamma \mathcal{T}
   [\tilde{\boldsymbol{a}}]_{2 j +1} = - \gamma \mathcal{N} [\mathcal{T}
   [\boldsymbol{a} \boldsymbol{}], \boldsymbol{a}]_{2 j +1}  - \gamma \mathcal{T}
   [\boldsymbol{a}]_{2 j+1}
  \end{equation}
  Thus, it only remains to investigate the term $\phi [\boldsymbol{a}]$.
  However,  $\phi[\boldsymbol{a}] = \phi[\tilde{\boldsymbol{a}}]$. Thus,
  $\widehat{\mathcal{R}} [\tilde{\boldsymbol{a}}]_{2 j} =
  \widehat{\mathcal{R}} [\boldsymbol{a}]_{2 j}$, $\widehat{\mathcal{R}}
  [\tilde{\boldsymbol{a}}]_{2 j + 1} = - \widehat{\mathcal{R}}
  [\boldsymbol{a}]_{2 j + 1}$, and consequently $\widehat{\mathcal{R}}
  [\tilde{\boldsymbol{a}}] = \tilde{\boldsymbol{a}}$.
\end{proof}

\begin{corollary}
  \label{cor:write-up18}If $\iota [\boldsymbol{a}]$ is a locally unique fixed point of 
  $\cal{R}$, then $a_{2 j + 1} = 0$ and
  $c_{2 j + 1} = 0$ for all $j \geq 2$.
\end{corollary}

We can now prove the following proposition, which ensures that fixed points of
$\mathcal{R}$ indeed give rise to solutions of Problem (a).

\begin{proposition}
  \label{prop:write-up18}Let $\lambda \in \mathbb{R}_+ \setminus \left\{
  \frac{1}{5} \right\}$, $\sigma \in \mathbb{C}$ and $\gamma \in \mathbb{R}$
  be fixed. Let $0 < \rho \leq 1$, and assume that $\iota [\boldsymbol{a}] =
  \mathcal{R} [\iota [\boldsymbol{a}]] \in \ell^1_{\rho}$ is a locally unique fixed point
  of $\mathcal{R}$. If, moreover, $\lambda$ and $\sigma$ are in a range where the function $Z (\zeta)$ determined by the
  series $\boldsymbol{c} = \mathcal{T} [\boldsymbol{a}]$ according to
  \eqref{e:write-up19} is univalent on $\partial \mathbb{D}$ then $Z(\zeta)$ solves Problem (a). Such range of univalency is discussed in Section \ref{sec:write-up3.4}.
\end{proposition}

\begin{proof}
  We need to show that
  for $y (\zeta)$ and $Z (\zeta)$ given by
  \begin{equation}
    y (\zeta) = \sum_{j = 0}^{\infty} a_j \zeta^{j - 1}, \quad Z (\zeta) =
    \sum_{j = 0}^{\infty} c_j \zeta^{j - 1} \quad \text{with} \quad a_0 = 1,
    c_0 = - 1, a_1 = c_1 = 0,
  \end{equation}
  the ratio $[y (\zeta) - \gamma \bar{y} (\zeta^{- 1}) Z (\zeta)] / \bar{y}
  (\zeta^{- 1})$ is well-defined and admits a Laurent series expansion according to
  \begin{equation}
    \frac{\zeta y (\zeta) - \gamma \bar{y} (\zeta^{- 1}) Z (\zeta)}{\bar{y}
    (\zeta^{- 1})} = k_{- 2} \zeta^3 + \sum_{j = 0}^{\infty} k_j \zeta^{1 - j}
    \label{e:seq48}
  \end{equation}
  with $k_{- 2}$ and $k_0$ determined by \eqref{e:seq015}. Again, we follow
  \cite{CrowdyTanveer2025}. Introduce the projection operator
  $\boldsymbol{P}$ acting on Fourier-series according to
  \begin{equation}
    \boldsymbol{P} \sum_{m = - \infty}^{\infty} p_m e^{2 \pi i m} 
    = \sum_{m = 1}^{\infty} p_m e^{2 \pi i m},
  \end{equation}
  and let $\boldsymbol{I}$ denote the identity operator. Using  the a-priori regularity assumption \ref{A-priori} on absolute convergence, computations show that
  \begin{equation}
    \boldsymbol{P}  \left(\gamma \bar{y} (\zeta^{- 1}) Z (\zeta)\right) = \gamma \sum_{j =
    1}^{\infty} \zeta^j \sum_{m = 0}^{\infty} \bar{a}_m c_{j + m}, \quad |
    \zeta | = 1.
  \end{equation}
  By Corollary \ref{cor:write-up18} we know that $\sum_{m = 0}^{\infty}
  \bar{a}_m c_{1 + m} = 0$ and therefore, using the fact that the pair 
  $(\boldsymbol{a}, \boldsymbol{c})$  solves \eqref{e:seq11}, we find
  \begin{equation}
    \zeta y (\zeta) = \boldsymbol{P} \left(\gamma \bar{y} (\zeta^{- 1}) Z (\zeta)\right) +
    k_{- 2} \zeta^4 + (k_{- 2} \bar{a}_2 + k_0) \zeta^2 + 1, \quad \text{on}
    \quad | \zeta | = 1, \label{e:seq52}
  \end{equation}
  so
  \begin{equation}
    \zeta y (\zeta) - \gamma \bar{y} (\zeta^{- 1}) Z (\zeta) = k_{- 2} \zeta^4
    + (k_{- 2} \bar{a}_2 + k_0) \zeta^2 + 1 - (\boldsymbol{I} - \boldsymbol{P} ) \left(\gamma
    \bar{y} (\zeta^{- 1}) Z (\zeta)\right) . \label{e:seq53}
  \end{equation}
  Since $Z$ is univalent on $\partial \mathbb{D}$ with $Z_{\zeta} \neq 0$ in $\overline{\mathbb{D}}$, we know that the $\bar{y} (\zeta^{- 1}) \neq 0$ on $\partial \mathbb{D}$. So, on
  $| \zeta | = 1$, we may expand the left-hand-side of \eqref{e:seq48} in a
  Fourier-series and find that
  \begin{equation}
    \zeta y (\zeta) - \gamma \bar{y} (\zeta^{- 1}) Z (\zeta) = \bar{y}
    (\zeta^{- 1}) \sum_{j = - \infty}^{\infty} \kappa_j \zeta^{1 - j} \quad
    \text{on} \quad | \zeta | = 1. \label{e:seq54}
  \end{equation}
  We can now compare Fourier-coefficients of the right-hand-sides of
  \eqref{e:seq53} and \eqref{e:seq54} giving
  \begin{equation}
    \sum_{j = - \infty}^{\infty} \kappa_j \zeta^{1 - j} = k_{- 2} \zeta^3 +
    k_0 \zeta^1 + \sum_{j = 1}^{\infty} \kappa_j \zeta^{1 - j}. \quad
    \text{on} \quad | \zeta | = 1. \label{e:3.32}
  \end{equation}
This proves the Proposition, if we associate to $K(\zeta^{-1})$ the Laurent-series \eqref{e:3.32}. Indeed, note that convergence, on $\partial \mathbb{D}$, of the Laurent-series on the left-hand-side of \eqref{e:3.32} implies convergence of the right-hand-side series on $\partial \mathbb{D}$. Hence, $K(\zeta^{-1})$ is analytic and single-valued for $\zeta \in \mathbb{C} \setminus \mathbb{D}$.
\end{proof}

Before we proceed with proving the existence of fixed points to $\mathcal{R}$ in $\ell^1_{\rho}$ for $\rho <1$,
let us note that in addition to Lemma \ref{lem:seq12}, we have an additional
symmetry property when $\sigma$ is purely imaginary, as is the case for the
bubble problem.

\begin{lemma}
  \label{lem:seq20}Let $\sigma \in i \mathbb{R}$ and let $\boldsymbol{a}$ be a
  fixed-point of $\widehat{\mathcal{R}}$. Let us define $\tilde{\boldsymbol{a}}=(\tilde{a}_2,\tilde{a}_3,...)$ as
  the series such that
  \begin{equation}
    \tilde{a}_{4 j} =\bar{a}_{4 j} \quad \text{and}
    \quad \tilde{a}_{4 j - 2} = -\bar{a}_{4 j - 2}, \quad \text{for all} ~ j.
  \end{equation}
  Then, $\widehat{\mathcal{R}} [\tilde{\boldsymbol{a}}] =
  \tilde{\boldsymbol{a}}$. Particularly, if $\iota [\boldsymbol{a}]$ is a locally unique fixed point of $\mathcal{R}$, then the coefficients $a_{4 i}$ and
  $c_{4 i}$ are fully real, whereas the coefficients $a_{4 j - 2}$ and $c_{4 j
  - 2}$ are purely imaginary.
\end{lemma}

The proof of this Lemma is analogous to the proof of Lemma \ref{lem:seq12},
and hence omitted.

\subsection{Existence of fixed points of $\mathcal{R}$}\label{sec:write-up3.2}

Proposition \ref{prop:write-up18} dictates that in order to prove the existence
of solutions to Problem (a), we need to prove the existence of fixed points
$\iota [\boldsymbol{a}]$ of $\mathcal{R}$.

Let
\begin{equation}
  \boldsymbol{a}^0 = (- 5 \lambda \bar{a}_2 + \sigma, 0, \lambda, 0, 0, \ldots)
  = \left( \frac{5 \lambda \bar{\sigma} - \sigma}{25 \lambda^2 - 1}, 0,
  \lambda, 0, 0, \ldots \right) \label{e:write-up50}
\end{equation}
be the series coefficients of the exact $\gamma = 0$ solutions presented in
Section \ref{sec:write-up2.3}. Indeed, it is readily checked that for $\gamma
= 0$, we have not only $\iota (\boldsymbol{a}^0) = \mathcal{R} [\iota
(\boldsymbol{a}^0)]$, but in fact $\iota (\boldsymbol{a}^0) = \mathcal{R}
[\boldsymbol{0}]$ where $\boldsymbol{0} \in \ell^1_{\rho}$ denotes the trivial
series.

\begin{lemma}\label{lem:3.13}
  \label{lem:seq07}Let $\lambda \in \mathbb{R}_+ \setminus \left\{ \frac{1}{5}
  \right\}$, $\sigma \in \mathbb{C}$. There exists an $r>0$ and $\gamma^{\star} >0$,  such that there exists a unique fixed-point  $\boldsymbol{b} = \mathcal{R} [\boldsymbol{b}] \in \mathcal{B}_r
  (\boldsymbol{0}) \subset \ell^1$ for all $\gamma \in (- \gamma^{\star}, \gamma^{\star})$. Furthermore, for $\rho$ sufficiently close to $1$, $\boldsymbol{b}\in \ell^1_{\rho}$.
\end{lemma}

\begin{proof}
  For brevity, let us set
  \[ \max \left\{ 1, \frac{1 + | 5 \lambda |}{| 1 - 25 \lambda^2 |} \right\}
     = \lambda_1 \]
  and note that $\| \tau (\boldsymbol{b}^1) - \tau (\boldsymbol{b}^2) \|_{\rho}
  \leq \lambda_1 \| \boldsymbol{b}^1 - \boldsymbol{b}^2 \|_{\rho}$. Then, using
  the estimates on $\mathcal{T}$ and $\mathcal{N}$ provided in Lemma
  \ref{lem:seq03}, we note that
  \[ \| \mathcal{T} [\tau (\boldsymbol{b}^1)] - \mathcal{T} [\boldsymbol{} \tau
     (\boldsymbol{b}^2)] \|_{\rho} \leq \lambda_1 \left( 2 + \frac{1}{3}
     \lambda_1 (\| \boldsymbol{b}^1 \|_{\rho} + \| \boldsymbol{b}^2 \|_{\rho})
     \right) \| \boldsymbol{b}^1 - \boldsymbol{b}^2 \|_{\rho} \]
  and
  \begin{align*}
    \| \mathcal{N} [\mathcal{T} [\tau
    (\boldsymbol{b}^1)],\tau (\boldsymbol{b}^1)] - \mathcal{N} [\mathcal{T} [\tau
    (\boldsymbol{b}^2)], \tau (\boldsymbol{b}^2)] \|_{\rho} \leq & \lambda_1 \| \boldsymbol{b}^1 \|_{\rho}
    \left( 2 + \frac{\lambda_1}{3} (\| \boldsymbol{b}^1 \|_{\rho} + \|
    \boldsymbol{b}^2 \|_{\rho}) \right) \| \boldsymbol{b}^1 - \boldsymbol{b}^2
    \|_{\rho}\\
    & + \lambda_1^2 \left( 2 \| \boldsymbol{b}^2 \|_{\rho} + \frac{\lambda_1
    }{3} \| \boldsymbol{b}^2 \|_{\rho}^2 \right) \| \boldsymbol{b}^1 -
    \boldsymbol{b}^2 \|_{\rho} .
  \end{align*}
From the definition of $\mathcal{R}$, we find
  {\fontsize{10pt}{11pt} \selectfont
  \begin{align*}
    \| \mathcal{R} [\boldsymbol{b}^1] - \mathcal{R} [\boldsymbol{b}^2] \|_{\rho}
    \leq & | \gamma | \lambda_1  (\| \boldsymbol{b}^1 \|_{\rho} + 1) \left( 2 +
    \frac{1}{3} \lambda_1  (\| \boldsymbol{b}^1 \|_{\rho} + \| \boldsymbol{b}^2
    \|_{\rho}) + \lambda_1^2 \left( 2 \| \boldsymbol{b}^2 \|_{\rho} +
    \frac{1}{3} {\lambda_1 }  \| \boldsymbol{b}^2 \|_{\rho}^2 \right) \right) \|
    \boldsymbol{b}^1 - \boldsymbol{b}^2 \|_{\rho} .
  \end{align*}}
  Let
  \begin{equation}
    r = r (\rho, \sigma, \lambda) = 2 \| \mathcal{R} [\boldsymbol{0}] \|_{\rho} =
    \rho^{- 2} | \sigma | + \rho^{- 4} | \lambda | \label{e:seq29},
  \end{equation}
such that, for $\boldsymbol{b} \in \mathcal{B}_r (\boldsymbol{0})$, we have
  \begin{equation}
    \| \mathcal{R} [\boldsymbol{b}] \|_{\rho} \leq \| \mathcal{R} [\boldsymbol{b}]
    - \mathcal{R} [\boldsymbol{0}] \|_{\rho} + \frac{r}{2} \leq | \gamma |
    \left( \lambda_1  (r + 1) \left( 2 + \frac{2}{3} \lambda_1  r \right) +
    \lambda_1^2 \left( 2 r + \frac{\lambda_1 }{3} r^2 \right) \right) 2 r +
    \frac{r}{2} . \label{e:3.36}
  \end{equation}
  Thus, for sufficiently small $\gamma^{\star} > 0$, we find that $\|
  \mathcal{R} [\boldsymbol{b}] \|_{1} < r$ and that $\mathcal{R} :
  \mathcal{B}_r (\boldsymbol{0}) \rightarrow \mathcal{B}_r (\boldsymbol{0})$ is a
  contraction for all $\gamma \in (- \gamma^{\star}, \gamma^{\star})$. From continuity of $r$ with respect to $\rho$, there exists $\rho^{\star}$ sufficiently close to $1$ such that the same contraction argument remains valid in the ball $\mathcal{B}_{r}(\boldsymbol{0}) \subset \ell^1_{\rho}$ for all $\rho \in (\rho^{\star},1]$. Since $\ell^1_{\rho} \subset \ell^1$,
  the Lemma follows by Banach's fixed point theorem.
\end{proof}

Now Proposition \ref{prop:write-up18} in combination with Lemmas
\ref{lem:seq12} and \ref{lem:seq20} gives the following result.

\begin{theorem}  \label{thm:seq09} 
	Let $\lambda \in \mathbb{R}_+ \setminus \left\{ \frac{1}{5}
  \right\}$, and $\sigma \in \mathbb{C}$ be fixed. Let $\gamma^{\star} > 0$ be
  suitably small, and let $\gamma \in (- \gamma^{\star}, \gamma^{\star})$
  also be fixed. Assume that the conditions of Proposition \ref{prop:write-up18} are satisfied. Then, $Z (\zeta)$ is the locally unique solution to
  Problem (a).\footnote{We have not excluded the possibility for multiple
  isolated solutions for fixed $\gamma, \lambda,$ and $\sigma$.} For such
  solution, $c_{2 j + 1} \equiv 0$. Moreover, if $\sigma \in i \mathbb{R}$,
  then $c_{4 j} \in \mathbb{R}$ whereas $c_{4 j - 2} \in i \mathbb{R}$.
  Moreover, for the function $K (\zeta)$ corresponding to the solution of
  Problem (a), $k_1 = 0$ according to Equation \eqref{e:seq10}. Moreover, $Z(\zeta)$ extends holomorphically past $\partial \mathbb{D}$.
\end{theorem}

The following two corollaries are immediate. Corollary \ref{cor:write-up20}
will only be required later, within the proof of Theorem \ref{thm:63}, whereas Corollary
\ref{prop:write-up26} addresses the existence of solutions to the Vortex
Problem a).

\begin{corollary}
  \label{cor:write-up20}The locally unique solution $Z (\zeta)$ to Problem
  (a) depends real-analytically on the parameters $(\lambda, \sigma)$ and
  $\gamma$.
\end{corollary}

\begin{proof}
  The operator $\mathcal{R}$ depends real-analytically on $(\lambda, \sigma)$,
  $\gamma$, and $\boldsymbol{b}$. So the fixed-point $\boldsymbol{b} = \mathcal{R}
  [\boldsymbol{b}]$ depends analytically on $(\lambda, \sigma)$, and $\gamma$.
  From this, it is readily checked that the Laurent-series coefficients
  $\boldsymbol{a} = \tau [\boldsymbol{b}]$ of $y (\zeta)$ depend real-analytically
  on $(\lambda, \sigma)$, and $\gamma$. Thus, the same statement is true for
  $\boldsymbol{c} = \mathcal{T} [\boldsymbol{a}]$, the Laurent-series coefficients
  of $Z (\zeta)$.
\end{proof}

\begin{corollary}  \label{prop:write-up26}
	%	Let $\beta - \gamma \in i \mathbb{R}$
	%\VK{Update this if we find that $\omega \in \mathbb{R}$ always},
	%	and 
	Let $| \gamma | $ be sufficiently small. Let the parameters $\hat{\alpha}, \beta,$ and $\gamma$ be in a range where
	the map $Z(\zeta)$ determined by the series coefficients $\boldsymbol{c} = \mathcal{T}
	[\boldsymbol{a}]$ with $\boldsymbol{a} = \widehat{\mathcal{R}} [\boldsymbol{a}]
	\in \mathcal{B}_{\delta} (\boldsymbol{0})$ is univalent (as discussed in the following Section \ref{sec:write-up3.4}). Then, the Vortex
	Problem a) admits a locally unique solution. For this solution,
	$\delta = 0$ and $\omega \in \mathbb{R}$.
\end{corollary}

\begin{proof}
	As we discussed at the end of Section \ref{sec:write-up2.2}, univalent $Z
	(\zeta)$ determined by $\boldsymbol{c} = \mathcal{T} [\boldsymbol{a}]$ with
	$\boldsymbol{a} = \widehat{\mathcal{R}} [\boldsymbol{a}] \in
	\mathcal{B}_{\delta} (\boldsymbol{0})$ yields a solution to the Vortex
	Problem (i) so long as $\omega \in \mathbb{R}$. Moreover, following  \cite{CrowdyNelsonKrishnamurthy2021} and using the fact that $\mathcal{Z}(s)$ is an arc-length parametrisation of $\partial \Omega$,  it can be seen from  \eqref{e:Seq01} that $\omega$ is always real for 
	\begin{align}
	-\omega = \frac{1}{2\pi } \int_{\partial \Omega} \xi(z) dz = \frac{1}{2\pi} \int_{\partial \Omega} \left[\frac{d}{ds} \overline{\mathcal{Z}}(s) - i \gamma \overline{\mathcal{Z}}(s)\right] \frac{d \mathcal{Z}}{ds} ds = \frac{1}{2\pi} |\partial \Omega| + \frac{\gamma}{ \pi} |\Omega|  \label{e:CircReal}
	\end{align}		
	which expresses $\omega$ in terms of the vortex area $|\Omega|$ and perimeter $|\partial \Omega|$.\footnote{We are grateful to Ali Farsad for pointing out this argument.}
	Thus, the Vortex Problem
	(i) has a solution. Moreover, recall that, according to \eqref{e:write-up10},
	$\bar{k}_1 = 3 c_3 \hat{\alpha} + \delta$ whereas Corollary \ref{thm:seq09}
	ensures that $k_1 = 0$ and $c_3 = 0$, giving $\delta = 0$ as was claimed in
	the Corollary.
\end{proof}

\subsection{Univalency of solutions to Problem (a)}\label{sec:write-up3.4} 

Let $\| . \|_{\infty}$ denote the sup-norm on $| \zeta | = 1$, and let $Z
(\zeta)$ denote the locally unique solution to Problem (a). The following
``stability result'' ensures that $Z (\zeta)$ is univalent whenever it is
close to another univalent function $\tilde{Z}(\zeta)$ \textit{that need not
necessarily solve Problem (a)}.

\begin{proposition}
  \label{prop:23}Let $\tilde{Z}$ be analytic in $\mathbb{D} \setminus \{0\}$ with a simple pole at the origin such that $\res{\zeta=0}\tilde{Z}=-1$. Let $\| \tilde{Z} \|_{\infty} < \infty$ on $\partial \mathbb{D}$, and assume that
  \begin{equation}
    \tilde{m}_Z :=\inf_{\zeta_1, \zeta_2 \in \partial \mathbb{D},
    \zeta_1 \neq \zeta_2} \left| \frac{\tilde{Z} (\zeta_1) - \tilde{Z}
    (\zeta_2)}{\zeta_1 - \zeta_2} \right| > 0 \label{e:66}
  \end{equation}
  and $\| Z_{\zeta} - \tilde{Z}_{\zeta} \|_{\infty} <  \tilde{m}_Z$.\footnote{We are grateful to Ali Farsad for pointing out that we do not require the bound $\frac{\pi}{2} \tilde{m}_Z$ which was given in \cite{CrowdyTanveer2025}.} Then $Z: \mathbb{D} \to \Omega$ is univalent and $Z_{\zeta} \neq 0$ on $\partial \mathbb{D}$.
\end{proposition}

\begin{proof}
	We follow the proof of \cite{CrowdyTanveer2025}, Proposition 6, verbatim, so we
	will be brief. It is sufficient to prove univalency on $\partial
	\mathbb{D}$. For $\zeta_1 = e^{i \nu_1} \neq e^{i \nu_2} = \zeta_2$,
	one finds that
	\begin{equation}
		\left| \frac{Z (\zeta_1) - Z (\zeta_2)}{\zeta_1 - \zeta_2} \right| \geq
		\tilde{m}_Z - \| Z_{\zeta} - \tilde{Z}_{\zeta} \|_{\infty} =: \hat{\delta}>0.
		\label{e:write-up65}
	\end{equation}
	Since, $\hat{\delta}$ does not depend on $\zeta_1$ and $\zeta_2$, we find that  $Z$ is univalent and, moreover, upon letting $\zeta_1 \to \zeta_2$, that $Z_{\zeta} \neq 0$ on $\partial \mathbb{D}$.
\end{proof}

\begin{remark} \label{rem.25} 
	Note that the function $\tilde{Z} (\zeta)$ in the formulation
	of Proposition \ref{prop:23} is \textit{not} required to solve Problem 3
	a).  For instance, we may take $\tilde{Z} = Z^0 $ to be the $\gamma = 0$
	exact solution given in Section \ref{sec:write-up2.3}, or we may take
	$\tilde{Z} (\zeta)$ to be any approximative solution computed numerically,
	like those presented in the companion manuscript. 
	Indeed, for a given triplet
	$(\gamma, \lambda, \sigma)$ of parameters, and any explicitly known
	approximative solution $\tilde{Z} (\zeta)$, the condition \eqref{e:66} may
	be checked explicitly, and one may use the preceding proposition to
	verify univalence through a computer assisted proof. How, exactly, such computer assisted proof of univalence can be achieved is discussed in Appendix \ref{App:Univ}.
\end{remark}

Let us now discuss the case when $\tilde{Z} (\zeta) = Z^0 (\zeta)$ is the
exact $\gamma = 0$ solution in detail.

\begin{lemma}
  Let
  \begin{equation}
    m^0_Z = \inf_{\zeta_1, \zeta_2 \in \partial \mathbb{D}, \zeta_1 \neq
    \zeta_2} \left| \frac{Z^0 (\zeta_1) - Z^0 (\zeta_2)}{\zeta_1 - \zeta_2}
    \right| .
  \end{equation}
  Then
  \begin{equation}
    m^0_Z \geq 1 - 2 \left| a_2^0 \right| - \left| 2 \lambda +\left(a_2^0\right)^2 \right| - 2
    \lambda \left|a_2^0
    \right| - \lambda^2, ~ \text{ with } ~ a_2^0 = \frac{5 \lambda \bar{\sigma} - \sigma}{25 \lambda^2 - 1}.   \label{e:write-up62}
  \end{equation}
\end{lemma}

\begin{proof}
  Let $\zeta_1 = e^{i (\theta + \varphi)}$ and $\zeta_2 = e^{i (\theta -
  \varphi)}$ with $\varphi \notin \pi \mathbb{Z}$, $u = e^{i \theta} $and $x =
  \cos (\varphi)$. Then, for $k \geq 0$, we have
  \begin{equation}
    \frac{\zeta_1^{2 k + 1} - \zeta_2^{2 k + 1}}{\zeta_1 - \zeta_2} = u^{2 k}
    U_{2 k} (x) \quad \text{and} \quad \frac{\zeta_1^{- 1} - \zeta_2^{-
    1}}{\zeta_1 - \zeta_2} = - u^{- 2} \label{e:write-up63}
  \end{equation}
  where $U_{2 k}$ denotes the Chebyshev polynomial of the second kind. 
  Using \eqref{e:write-up63}, we find
  \begin{equation}
    \frac{Z^0 (\zeta_1) - Z^0 (\zeta_2)}{\zeta_1 - \zeta_2} = - u^{- 2} + 2
    a_2^0 + \frac{2 \lambda + (a_2^0)^2}{3} u^2 U_2 (x) + \frac{2}{5} \lambda
    a_2^0 u^4 U_4 (x) + \frac{\lambda^2}{7} u^6 U_6 (x) . \label{e:write-up64}
  \end{equation}
  Using the fact that $| U_2 (x) | \leq 3$, $| U_4 (x) | \leq 5$ and $| U_6
  (x) | \leq 7$, the inequality \eqref{e:write-up62} now follows from
  \eqref{e:write-up64} and the reverse triangle inequality.
\end{proof}

Particularly, the inequality \eqref{e:write-up62} gives a direct way to check
whether the condition $m_Z^0 > 0$ holds. Of course, one may also use the
equation \eqref{e:write-up64} directly to check, numerically, whether $m_Z^0
\neq 0$. Since the sup-norm is weaker than the $\ell^1_{\rho}$ norm on
functions analytic on $\partial \mathbb{D}$, we find that $\| Z_{\zeta} -
Z_{\zeta}^0 \|_{\infty} < m^0_Z$ whenever $| \sigma |, | \lambda |$ are
sufficiently small (cf. Section \ref{sec:write-up3.2}). Moreover, whenever $|
\sigma |, | \lambda |$ are sufficiently small, it is readily seen from
\eqref{e:write-up62} that $m_Z^0 > 0$, so we find the following result.

\begin{corollary}
  \label{cor:write-up26}Let $| \lambda |, | \sigma |$ be
  small enough for $m_Z^0>0$, and let $Z (\zeta) = - \zeta^{- 1} + \sum_{j = 1}^{\infty} c_{2 j}
  \zeta^{2 j - 1}$ be the locally unique solution to Problem (a) for $|
  \gamma |$ small. Then, $Z (\zeta)$ is univalent, so the Vortex Problem (i)
  admits a locally unique solution.
\end{corollary}

\begin{remark}
	Note that the conditions stated in Corollary \ref{cor:write-up26} are rather restrictive. Indeed, according to Remark \ref{rem.25},  the parameter range can be extended significantly through a computer assisted proof of univalence, by following the strategy detailed in Appendix \ref{App:Univ}. Thus, one may prove univalency for parameter ranges that are much closer to what is observed numerically the companion paper \cite{CKT}.
\end{remark}

Since, according to the discussion of Section \ref{sec:write-up2.2}, solutions
to Problem (b) also solve Problem (a), the results formulated within this
Section \ref{sec:write-up3.4} hold for any solution to Problem (b).

\section{On the existence of solutions to Problem (b)}\label{sec:write-up04}
Within this section, we will discuss the existence of locally unique solutions
to Problem (b). Note that the results presented within this section apply primarily to the bubble problem,
whereas the results formulated in Section \ref{sec:write-up3}, concerned
 \textit{both}, the bubble \textit{and} the vortex problem.

\subsection{The boundary condition \eqref{e:write-up75}
rewritten.}\label{sec:write-up4.1}

Let us introduce the notation
\begin{equation}
  Y (\zeta) = \zeta y (\zeta), \label{e:write-up74}
\end{equation}
where $y (\zeta) = \sqrt{Z_{\zeta} (\zeta)}$ for the conformal map $Z :
\mathbb{D} \rightarrow \Omega$ introduced in Section \ref{sec:write-up2.2}.
Let us rewrite the boundary condition \eqref{e:write-up75} that must be
satisfied by any solution to Problem (b).
\begin{equation}
  - 1 + 2 \Re \left( \frac{\zeta Y'}{Y} \right) + \frac{| W_{\zeta} |^2}{| Y
  |^2} - \gamma | Y |^2 = 0, \quad \text{on} \quad | \zeta | = 1.
  \label{e:C-01}
\end{equation}
Now, define
\begin{equation}
  F (\zeta) = \log (Y (\zeta)) \label{e:C-02}
\end{equation}
where the normalisation $Y (0) = 1$ guarantees that
\begin{equation}
  F (0) = 0.
\end{equation}
\begin{remark}
  \label{rem:write-up29}
  Within Section \ref{sec:4.6}, we will prove that there exists $F (\zeta)$, holomorphic for $| \zeta | < 1$ and analytic on $| \zeta | = 1$ for which corresponding $Y$ solves \eqref{e:C-01} (see Theorem \ref{thm:63}).  Thus all subsequent calculations will be justified a-posteriori, and should be viewed as formal computations
  that lead to the framework which we use in order to prove the existence of
  solutions to Problem (b).
\end{remark}

Bearing the previous remark in mind, let us, formally, re-write the boundary
condition \eqref{e:C-01} by using \eqref{e:C-02}. This gives
\begin{equation}
  \boldsymbol{M}  [F] :=- 1 + 2 \text{Re} (\zeta F') + |
  W_{\zeta} |^2 \exp (- 2 \text{Re} (F)) - \gamma \exp (2 \text{Re} (F)) = 0.
  \label{e:C-04}
\end{equation}
Let us re-write the equation \eqref{e:C-04}. To that end, let
\begin{equation}
  Y_0 (\zeta) = \zeta y_0 (\zeta), \quad \text{and} \quad F_0 (\zeta) = \log
  (Y_0 (\zeta)),
\end{equation}
with $y_0$ the exact $\gamma = 0$ solution given in Equation \eqref{e:18}.
In order to prove that the bubble problem admits a solution with $\gamma \neq
0$, we will look for solutions to \eqref{e:C-01} of the form
\begin{equation}
  F (\eta) = F_0 (\eta) + f (\eta), \quad \eta = \zeta^2 . \label{e:C-08}
\end{equation}
with $f \in \mathcal{S}$, where $\mathcal{S}$ is defined as follows.

\begin{definition}
  \label{def:C-01}Let $A (\mathbb{D})$ denote the ``disc-algebra'' of
  functions holomorphic in $\eta = \zeta^2$ that extend continuously onto
  $\partial \mathbb{D}$, equipped with the sup-norm
  \begin{equation}
    \| f \| = \sup_{\eta \in \partial \mathbb{D}} |f (\eta) |,
  \end{equation}
  and let
  \[ B_{\delta} = \{ f \in \mathcal{S} : \| f \| < \delta \} \]
  denote the $\delta$-ball. We consider the following two subspaces of $A
  (\mathbb{D})$.
  \begin{enumerate}
    \item $\mathcal{S} = \{ f \in A (\mathbb{D}) : f (0) = 0 \}.$ This is a
    Banach-subspace;
    \item $\mathcal{S}^{\omega} = \left\{ f \in \mathcal{S} : f ~
    \text{extends analytically onto} ~ \partial \mathbb{D} \right\}$.
  \end{enumerate}
\end{definition}

Theorem \ref{thm:seq09} ensures that for any solution to the bubble problem, which necessarily solves Problem (a), $Y (\zeta)$ will be a function of
$\eta$ only. So, the condition that $f$ depends only on $\eta$ can be imposed
without loss of generality. Our goal is to prove that for $\gamma \neq 0$,
there always exists a function $f (\eta) \in \mathcal{S}^{\omega}$ such that
$F (\zeta)$ solves the Equation \eqref{e:C-04}. To achieve this goal, let us
study which equation needs to be satisfied by such $f (\eta) \in
\mathcal{S}^{\omega}$. Computations show that, on $\partial \mathbb{D}$,
\begin{equation}
  \boldsymbol{M} [F_0] = \frac{B_0}{D (\eta)} - \gamma D (\eta) =: r (\eta) \label{e:C-10}
\end{equation}
with
\begin{equation}
  B_0 = \left( \frac{(25 \alpha^2 + 3) (\alpha^2 + 3)}{(5 \alpha^2 - 3)^2}
  \right) \left(\Gamma_0^2 - \frac{(3 - 5 \alpha^2)^2 (3 - 7
  \alpha^2)}{9 (25 \alpha^2 + 3)} \right) \label{e:C-11}
\end{equation}
and
\begin{equation}
  D (\eta) = \eta^{- 2} (A_4 (\eta^4 + 1) + i A_2 (\eta^3 - \eta) + A_0
  \eta^2) \label{e:C-12}
\end{equation}
where the real coefficients $A_0, A_2, A_4 \in \mathbb{R}$ are given by
\begin{equation}
  A_0 = \frac{\alpha^4}{9} + \frac{36 \Gamma_0^2 \alpha^2}{ (5 \alpha^2 -
  3)^2} + 1, \quad A_2 = \frac{2 \alpha (\alpha^2 - 3) \Gamma_0}{5 \alpha^2 -
  3}, \quad A_4 = \frac{\alpha^2}{3} . \label{e:C-13}
\end{equation}
\begin{remark}
  For $\gamma = 0$, $F_0 (\eta)$ must satisfy the boundary condition $r (\eta)
  = 0$. This requires $B_0 \equiv 0,$ which confirms that the constraint
  \eqref{e:C-15} must be satisfied by the parameters $\alpha$ and $\Gamma_0$
\end{remark}

\begin{remark}
  Note that $D (\zeta) = A_4 (\zeta^4 + \zeta^{- 4}) + i A_2 (\zeta^2 -
  \zeta^{- 2}) + A_0$, which is why we choose the notations $A_0$, $A_2$, and
  $A_4$.
\end{remark}

After some more computations, we obtain the following Lemma.

\begin{lemma}
  \label{lem:write-up33}Let $F = F_0 + f$, and let
  \begin{equation}
    U (\eta) = - \frac{| W_{\zeta} |^2 (\eta)}{2 D (\eta)}
    \label{e:write-up90}
  \end{equation}
  with $W_{\zeta}=\alpha(\zeta - \zeta^{-3}) + i\Gamma_0 /\zeta$ and $D$ given by \eqref{e:C-12}.
  Then, we have $\boldsymbol{M} [F] = 0$ if, and only if, the following
  equation is true for $f$.
  \begin{equation}
    \Re (\eta f' + U f) = - \frac{1}{4} r + \frac{1}{2} \gamma D \Re (f) +
    \frac{1}{4} \boldsymbol{N} [f] \label{e:C-20}, \quad \text{on} \quad
    \partial \mathbb{D}
  \end{equation}
  where
  \begin{equation}
    \boldsymbol{N} [f] = \boldsymbol{N}_{\!\!1} [f] + \gamma \boldsymbol{N}_{\!\!2} [f]
  \end{equation}
  with
  \begin{equation}
    \boldsymbol{N}_{\!\!1} [f] = 2 U (\exp (- 2 \Re f) - 1 + 2 \Re (f)), \quad 
    \boldsymbol{N}_{\!\!2} [f] = D (\exp (2 \Re f) - 1 - 2 \Re (f)),
  \end{equation}
\end{lemma}

\subsection{The integral operator $\mathcal{Q}$}\label{s:C-02}

The discussions of Section \ref{sec:write-up2.2} and \ref{sec:write-up4.1}
dictate that in order to prove the existence of solutions to the bubble
problem, we need to show that the Equation \eqref{e:C-20} admits a solution $f
\in \mathcal{S}^{\omega}$. In order to prove this, we first need understand
the complex-analytic structure of the functions $r (\eta)$ and $U (\eta)$ for
this will enable us to reformulate the Equation \eqref{e:C-20} as a
first order integro differential equation which may, formally, be
integrated. Once this integration is achieved, the problem of proving the
existence of solutions to Problem (b) will be reformulated as a fixed-point
problem of the form $f = \mathcal{Q} [f]$ where $\mathcal{Q}$ is some
appropriate integral-operator.

The following Section \ref{sec:write-up4.2.1} will be dedicated to studying
the singularity structure of $r$ and $U$. The reformulation of \eqref{e:C-20}
as a fixed-point problem will be the topic of Section \ref{sec:write-up4.2.2},
wherein we will also introduce the integral operator $\mathcal{Q}$.

\subsubsection{Singularity structure of $r (\eta)$ and $U
(\eta)$.}\label{sec:write-up4.2.1}

\begin{lemma}
  \label{lem:write-up35}The function $r (\eta)$ has a second order pole at
  $\eta = 0$. Moreover, the following four points are first order poles of $r
  (\eta)$.
  \begin{equation}
    \eta_1 = \frac{i \tau_1 - \sqrt{4 - \tau_1^2}}{2}, \quad \eta_2 = \frac{i
    \tau_2 + \sqrt{4 - \tau_2^2}}{2}, \quad \eta_3 = \frac{i \tau_1 + \sqrt{4
    - \tau_1^2}}{2}, \quad \eta_4 = \frac{i \tau_2 - \sqrt{4 - \tau_2^2}}{2},
  \end{equation}
  where
  \begin{equation}
    \tau_{1, 2} = - \frac{3 (\alpha^2 - 3)}{\alpha (5 \alpha^2 - 3)}
    \Gamma_0 \mp \frac{3 (\alpha^2 + 3)}{\alpha
    (5 \alpha^2 - 3)} \sqrt{ \Gamma_0^2 + \frac{(5
    \alpha^2 - 3)^2}{27}} . \label{e:write-up92}
  \end{equation}
  The function $r (\eta)$ is holomorphic in $\mathbb{C} \setminus \{ 0,
  \eta_1, \eta_2, \eta_3, \eta_4 \}$. The choice of branch of the square root
  will be explained in Remark \ref{rem:write-up40}.
\end{lemma}
The lemma's proof is computational since $r(\eta)$ is an explicitly given rational function, so we omit it for brevity. Moreover, we will require the following rational function $b:(0,\infty) \times \mathbb{R} \to \mathbb{R}$ defined by 
	\begin{align}
		b(\alpha,\Gamma_0) = \Gamma_0^2 - \frac{(3-5 \alpha^2)^2 (3-7 \alpha^2)}{9 (25 \alpha^2 + 3)}. \label{e:write-up93}
	\end{align}

\begin{remark}
  The points $\tau_1$ and $\tau_2$ may also be described in terms of the
  coefficients $A_0, A_2,$ and $A_4$ via $\tau_{1, 2} = \frac{1}{2 A_4} \left(
  - A_2 \pm \sqrt{A_2^2 + 4 A_4 (2 A_4 + A_0)} \right)$. However, the
  representation \eqref{e:write-up92} will be more advantageous for the
  purpose of studying the equation \eqref{e:C-20}, as we shall see in the
  sequel.
\end{remark}

\begin{remark}
  \label{rem:write-up37}Note that $B_0$, which was introduced in
  \eqref{e:C-11} and directly feeds into the definition of $r (\eta)$,
  satisfies $B_0 = \frac{(25 \alpha^2 + 3) (\alpha^2 + 3)}{(5 \alpha^2 - 3)^2}
  b (\alpha, \Gamma_0)$. Particularly, the constraint \eqref{e:C-15}, which must
  be satisfied by the parameters $\alpha$ and $\Gamma_0$ in order for the bubble
  problem to admit a solution for $\gamma = 0$, is equivalent to imposing $b
  (\alpha, \Gamma_0) = 0$.
\end{remark}

\begin{remark}
	Since $b(\alpha,\Gamma_0)$, defined in \eqref{e:write-up93}, is a rational function we shall write $b(\alpha,\Gamma_0)=0$ and $(\alpha,\Gamma_0) \in b^{-1}(0)$ interchangeably. The curve described by $b(\alpha,\Gamma_0) =0$ is shown in Figure \ref{Fig:b}, left.
\end{remark}

Now, note that if $b (\alpha, \Gamma_0) = 0$, we are guaranteed that $\tau_{1,
2} \in \mathbb{R}$. In fact, recalling that $\alpha > 0$, we find that $\tau_1
> 2$ and $\tau_2 < - 2$. Upon treating $b$ as independent parameter, solving \eqref{e:write-up93} for $\Gamma_0$ shows that $\Gamma_0$ depends continuously on $b$ and $\alpha$, whenever the inequality
	\begin{align}
		 b+ \frac{(3-5 \alpha^2)^2(3-7\alpha^2)}{9 (25 \alpha^2 + 3)} \geq 0 \label{e:4.20}
	\end{align}
holds.	The inequality \eqref{e:4.20} necessarily must hold to ensure that $\Gamma_0$ be real, as required for the bubble problem.
 Thus, according to \eqref{e:write-up92}, we have the
following Lemma.

\begin{lemma} \label{lem:C-04}
	Let $0< \alpha^2 < \frac{3}{5}$ and $\Gamma_0 \in \mathbb{R}$ such that $|b (\alpha, \Gamma_0) |$ is sufficiently small (but not necessarily vanishing). Then, the two roots $\eta_1$ and $\eta_2$
  lie inside the unit circle, and they are purely imaginary. The roots
  $\eta_3$ and $\eta_4$, on the other hand, lie outside the unit circle, and
  they are also purely imaginary. Moreover, these roots are bounded away from
  $\partial \mathbb{D}$.
\end{lemma}

\begin{remark}  \label{rem:write-up40} 
	Let us comment on  Lemma \ref{lem:C-04}.
	\begin{enumerate}
		\item For the Lemma to be true, we
		must have $\tau_1 > 2$ and $\tau_2 < - 2$. Particularly, the square root
		appearing within Lemma \ref{lem:write-up35} is simply the usual square root
		function on $\mathbb{R }_+$ with branch chosen according to $\sqrt{1} = 1$;
		\item The condition that $\alpha^2 < 3/5$ is due to the fact that the point $\alpha^2 = 3/5$ is a polar singularity of $\tau_{1,2}$ and $B_0$.
		\item Since, for  $b \neq 0$, the inequality \eqref{e:4.20} is required to ensure $\Gamma_0 \in \mathbb{R}$, the formulation of Lemma \ref{lem:C-04} implicitly assumes the validity of \eqref{e:4.20}.  A generic region corresponding to ``$|b(\alpha, \Gamma_0)|$ small'', without regards to the inequality \eqref{e:4.20}, is shown in Figure \ref{Fig:b}, right.
		\item For $b >0$, the inequality \eqref{e:4.20} permits $\alpha^2 > 3/7$ whereas, for $b<0$, the inequality \eqref{e:4.20} necessarily enforces $\alpha^2 < 3/7$.
		\item Later, within Section \ref{sec:4.5}, we will further impose the condition that $0 \leq | \Gamma_0 | < 1$.
	\end{enumerate}
\end{remark}

	\begin{figure}[h!]
	\centering 
	\includegraphics[width=\textwidth]{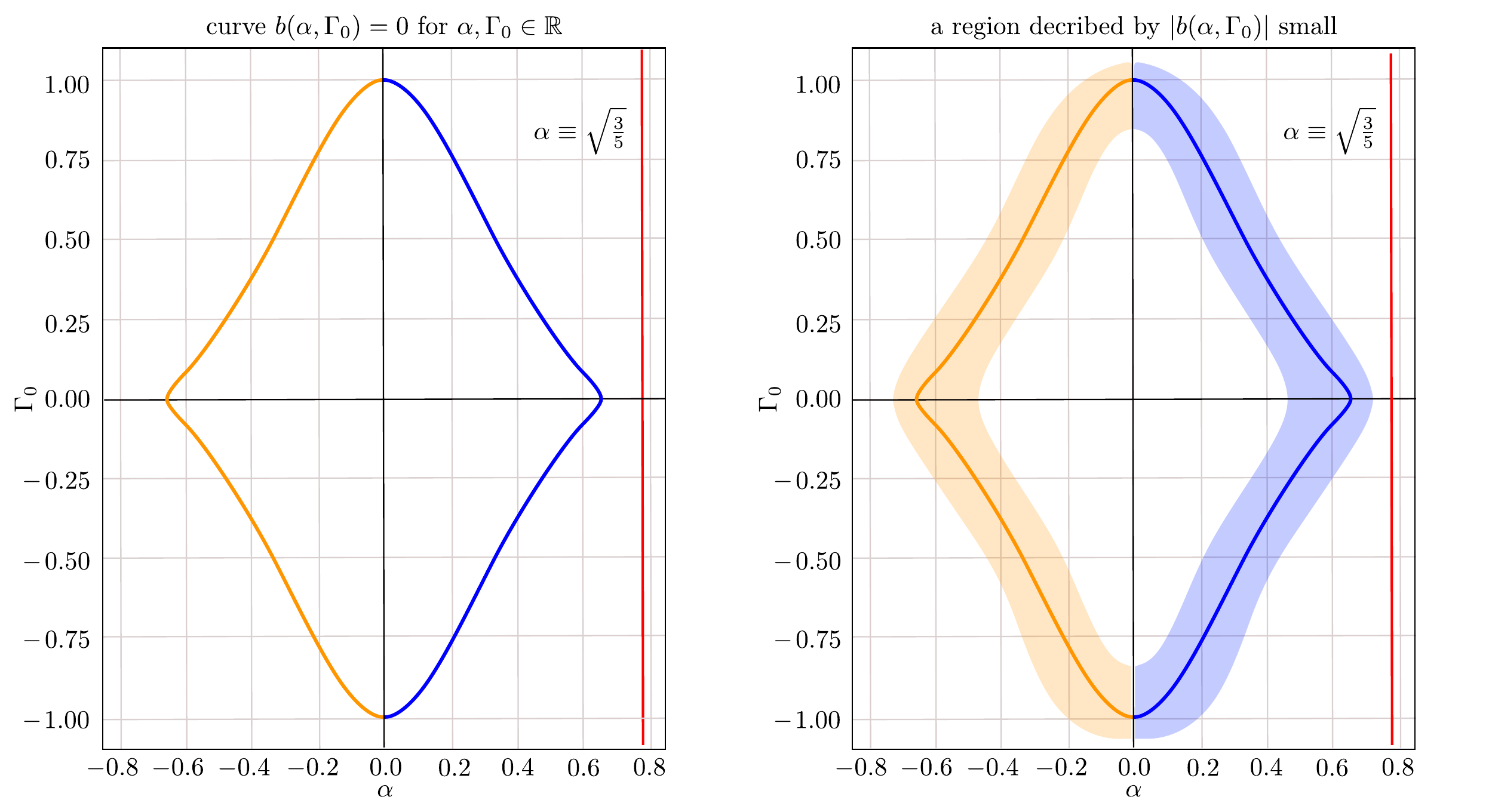}
	\caption{\small Left: Curve $b(\alpha,\Gamma_0)=0$. Right: Illustration of a region for which $|b(\alpha,\Gamma_0)|$ is ``small'', without regards to the inequality \eqref{e:4.20}, but not necessarily vanishing.  Recall that we only consider $\alpha>0$ throughout this article. On the left and right, we also see the line $\alpha \equiv \sqrt{3/5}$ whereon $\tau_1, \tau_2$, and $B_0$ exhibit a polar singularity.}
	\label{Fig:b}
\end{figure} 

Finally, let us describe the singularity structure of $U (\eta)$ within
$\mathbb{D}$, which one obtains by inspection of \eqref{e:write-up90}.

\begin{lemma} \label{lem:write-up40} Let $|b (\alpha, \Gamma_0) |$ be sufficiently small and let the inequality \eqref{e:4.20} hold. Then, the
  following statements are true.
  \begin{enumerate}
    \item Within $\mathbb{D}$, the function $U (\eta)$ has exactly two simple
    poles, which are located at $\eta = \eta_1$ and $\eta = \eta_2$
    
    \item $U$ is real-valued on $\partial \mathbb{D}$.
  \end{enumerate}
  Moreover, letting $r_j = \text{res}_{\eta = \eta_j} U (\eta)$ denote the
  residue of $U$ at $\eta_j$, we have $r_j \in i \mathbb{R}$.
\end{lemma}

Lemma \ref{lem:C-04} will be crucial in the sequel since our analysis will rely on the fact that $\eta_1$ and $\eta_2$ not only lie within the unit disc, but are bounded away from $\partial \mathbb{D}$. That is, we will henceforth assume the following.\\

\begin{Regularity}[A-priori assumption on $\alpha$ and $\Gamma_0$.]\label{Assumption4.13}
	For the remainder of Section \ref{sec:write-up04}, we will \textit{always} assume that $\alpha$ and $\Gamma_0$ satisfy the conditions required for Lemma \ref{lem:C-04} to be true. Indeed, whenever we require that $|b(\alpha,\Gamma_0)|$ is ``sufficiently small'', we require, in particular, that the conditions of Lemma \ref{lem:C-04} are satisfied by $\alpha$ and $\Gamma_0$,  so  the inequality \eqref{e:4.20} holds.
\end{Regularity}

\subsubsection{Reformulation of the Equation \eqref{e:C-20}
}\label{sec:write-up4.2.2}

Let us use the results of the previous Section \ref{sec:write-up4.2.1} in
order to reformulate the Equation \eqref{e:C-20} as
an integro differential equation which may be integrated. To achieve this, let
\begin{equation}
  W_j (\eta) = - \left( \frac{r_j}{\eta - \eta_j} + \frac{r_j \eta}{1 + \eta_j
  \eta} \right), \quad j = 1, 2. \label{e:C-23}
\end{equation}
where $r_j$ denotes the residue of $U$ about $\eta = \eta_j$ and $\eta_j$
denotes either of $U$'s simple poles within $\mathbb{D}$. Note that, due to
Lemmas \ref{lem:C-04} and \ref{lem:write-up40}, the functions $W_j$ have the
following properties.
\begin{enumerate}
  \item $W_j$ is purely imaginary on $\partial \mathbb{D}$
  
  \item $\text{res}_{\eta = \eta_j} W_j = -r_j$
\end{enumerate}
Particularly, the function
\begin{equation}
  g  = \eta f' + U f + W_1 f (\eta_1) + W_2 f (\eta_2) \label{e:C-24}
\end{equation}
lies in $A(\mathbb{D})$ i.e., the purpose of introducing the functions $W_j$ is
to remove the poles of $U f$ whilst leaving the Equation \eqref{e:C-20}
invariant, in the sense that \eqref{e:C-20} becomes rewritten as
\begin{equation}
  \Re (\eta f' + U f + W_1 f (\eta_1) + W_2 f (\eta_2)) = - \frac{1}{4} r +
  \frac{1}{2} \gamma D \Re (f) + \frac{1}{4} \boldsymbol{N} [f] . \label{e:C-25}
\end{equation}
Next, introduce the Schwarz/Poisson-Operator\footnote{The following expression requires $g$ to be H\"older continuous on $\partial \mathbb{D}$. However, using Fourier theory, the formula for $\mathcal{P}[g]$ extends to $g \in L^2(\partial \mathbb{D})$, see \cite{Duren1970TheoryHpSpaces} Chapter 4.}
\begin{equation}
  \mathcal{P} [g] (\eta) = \frac{1}{2 \pi i} \int_{\partial \mathbb{D}}
  \frac{\chi + \eta}{\chi - \eta} g (\chi) \frac{d \chi}{\chi} .
  \label{e:C-26}
\end{equation}
$\mathcal{P}$ has the following properties that will be important for us (see,
for instance, \cite{Duren1970TheoryHpSpaces}, Chapter 4). For
referencing in the sequel, we shall formulate these results as a lemma.

\begin{lemma}\label{lem:write-up42}
  \hfill 
  \begin{enumerate}
    \item We consider $\mathcal{P}$ acting on real-valued functions $g :
    \partial \mathbb{D} \rightarrow \mathbb{R}$ that are continuous. For such
    $g$, $\mathcal{P} [g]$ is analytic in $\mathbb{D}$.
    \item For $g \in L^2 (\partial \mathbb{D}),$ we have $\mathcal{P} [g]
    (\eta) \rightarrow g (\eta) + i \mathcal{H} [g] (\eta)$ in $L^2 (\partial
    \mathbb{D})$ as $\eta \rightarrow \partial \mathbb{D}$ from within
    $\mathbb{D}$, where $\mathcal{H} [g] (\eta)$ denotes the $2\pi$-periodic
    Hilbert-transform.    
    \item For $g \in A(\mathbb{D})$, we have $\mathcal{P} [\Re (g)] = g - i \Im
    (g (0))$.
  \end{enumerate}
\end{lemma}

Since $f (0) = 0$ and since, by Lemma \ref{lem:write-up40}, $U$ has no pole at
the origin, we have 
	\begin{align*}
		(\eta f' + U f) (0) = 0.
	\end{align*} 
So, applying the operator
$\mathcal{P}$ to \eqref{e:C-25}, we find, by Lemma \ref{lem:write-up42}, 3.,
that
\begin{equation}
  \eta f' + U f + W_1 f (\eta_1) + W_2 f (\eta_2) = \mathcal{P} [G] (\eta) +
  i \Im \left( \frac{r_1}{\eta_1} f (\eta_1) +
  \frac{r_2}{\eta_2} f (\eta_2) \right), \quad | \eta | < 1 \label{e:C-27},
\end{equation}
with
\begin{equation}
  G (\eta) = - \frac{1}{4} r (\eta) + \frac{1}{2} \gamma D (\eta) \Re (f
  (\eta)) + \frac{1}{4} \boldsymbol{N} [f] (\eta) . \label{e:C-28}
\end{equation}
Equation \eqref{e:C-27} is the aforementioned reformulation of \eqref{e:C-20}
which may be integrated. Indeed, using the fact that $\{ \eta_3, \eta_4 \} =
\left\{ - \frac{1}{\eta_1}, - \frac{1}{\eta_2} \right\}$ (which is obtained
from the symmetry $U (\eta) = U (- 1 / \eta)$), we find that
{\fontsize{10.2pt}{11pt}\selectfont
\begin{equation}
   f (\eta) = \frac{1}{\mu (\eta)} \int_0^{\eta} \frac{\mu
  (\chi)}{\chi} \left( - W_1 (\chi) f (\eta_1) - W_2
  (\chi) f (\eta_2) + i \Im \left( \frac{r_1}{\eta_1}
  f (\eta_1) + \frac{r_2}{\eta_2} f (\eta_2) \right) + \mathcal{P} [G]
  (\chi) \right) d \chi, \label{e:C-29}
\end{equation}}
with integrating factor $\mu (\eta)$ given by
\begin{equation}
  \mu (\eta) = \eta^{\frac{3}{2}} (\eta - \eta_1)^{\nu_1} (\eta -
  \eta_2)^{\nu_2} (1 + \eta \eta_1)^{- \nu_1} (1 + \eta \eta_2)^{- \nu_2},
  \quad \text{for} \quad \nu_j = \frac{r_j}{\eta_j}, \quad j = 1, 2.
  \label{e:C-30}
\end{equation}
\begin{remark}
  \label{rem:C-07}The equation \eqref{e:C-27}, evaluated at $\eta =
  0$, dicates that we necessarily have
  \begin{equation}
\Re (\nu_1 f (\eta_1) + \nu_2 f (\eta_2)) - \mathcal{P}
    [G] (0) = 0, \label{e:C-31}
  \end{equation}
  since $W_j (0) = \nu_j$. As we will show in Corollary \ref{cor:C-15}, the
  condition \eqref{e:C-31} guarantees that the right-hand side of
  \eqref{e:C-29} vanishes at the origin. Thus, we obtain no additional
  constant of the form $\frac{C}{\mu (\eta)}$ upon integrating \eqref{e:C-27}.
 % \textcolor{red}{Saleh, I agree with your last comment on this remark. Namely, that it may be better to postpone this remark until we give the Corollary  \ref{cor:C-15}. But I would like to give a reason as to why we have no  additional constant in \eqref{e:C-29}. How should we go about this?}
\end{remark}

Some comments are in place regarding the path of integration that is taken in
\eqref{e:C-29}, but we will postpone these comments until Remark
\ref{rem:C-10}. First, let us describe the domain of $\mu
(\eta) $. The function $\mu (\eta)$ has exactly three branch points within the
unit disc, and they are located at
\[ \eta = 0, \quad \eta = \eta_1 \quad \text{and} \quad \eta = \eta_2 . \]
So, for definiteness, let us introduce some branch cuts to render $\mu (\eta)$
single-valued. This is done as follows. Choose the cuts which give the domain
$\mathbb{D}_c$ shown in Figure \ref{Fig:03}, (1). That is, we
connect $\eta_1$ to $- 1 / \eta_1$ with a cut along the positive imaginary
axis, we connect $\eta_2$ to $- 1 / \eta_2$ with a cut along the negative
imaginary axis, and we connect $0$ to $\infty$ along the positive real-axis. \
Moreover, define
\begin{equation}
  \quad \mathbb{T}_c = \overline{\mathbb{D}}_c \cap \partial \mathbb{D}
  \label{e:C-32}
\end{equation}
Within \eqref{e:C-32}, the shores of the branch cuts are considered as part
of the boundary $\partial \mathbb{D}_c$ of $\mathbb{D}_c$, so $\mathbb{T}_c$
is the cut boundary of the unit disc, as displayed in Figure
\ref{Fig:03}, (2). The branch of $\mu (\eta)$ in these domains is
determined by $\arg (- 1) = \pi$. We can now factorise $\mu(\eta), ~ \eta \in \mathbb{D}_c,$ as follows.
\begin{equation}
  \mu (\eta) = \eta^{\frac{3}{2}} \left( \frac{\eta - \eta_1}{1 + \eta \eta_1}
  \right)^{\nu_1} \left( \frac{\eta - \eta_2}{1 + \eta \eta_2} \right)^{\nu_2}\hspace{-.35cm}. \label{e:C-33}
\end{equation}

\begin{figure}[h!]
	\includegraphics[width=\textwidth]{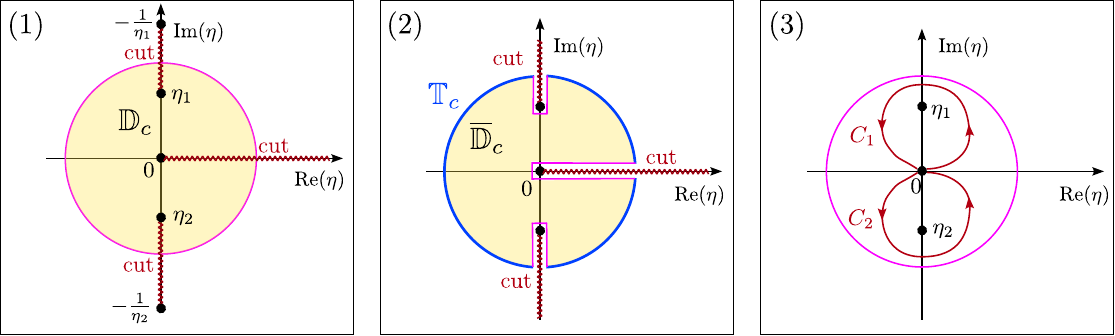}
		\caption{\small (1): Domain $\mathbb{D}_c.$ (2):   $\mathbb{D}_c$'s closure $\overline{\mathbb{D}}_c$ and the ``cut''-boundary $\mathbb{T}_c$. (1) and (2): branch cuts used for these sets' definitions. (3): curves $C_1$ and $C_2$ used within the definition of the integral operators $\mathcal{J}_1$ and $\mathcal{J}_2$ in \eqref{e:C-35}.}
			\label{Fig:03}
\end{figure}

\begin{notation}
  \label{def:C-08}We introduce the operator $\mathcal{Q}$, acting on
  $\mathcal{S}$, as
  \[ \mathcal{Q} [f] (\eta) = \frac{1}{\mu (\eta)} \int_0^{\eta} \frac{\mu
     (\chi)}{\chi} (- W_1 (\chi) f (\eta_1) - W_2 (\chi) f (\eta_2) +
    i \Im (\nu_1 f (\eta_1) + \nu_2 f (\eta_2)) +
     \mathcal{P} [G] (\chi)) d \chi . \]
  Recall that according to \eqref{e:C-28}, the function $G$ depends on
  $\gamma$ and $f$ even though this dependence is not explicitly highlighted.
  Recall also that all functions appearing in the definition of $\mathcal{Q}$
  depend on the parameters $\alpha$ and $\Gamma_0$. 
\end{notation}

\begin{remark}
  \label{rem:write-up46}At this point, neither are the mapping properties of
  $\mathcal{Q}$ known, nor has the path of integration been specified. Thus,
  one should treat the operator $\mathcal{Q}$ as a short-hand notation which
  allows us to re-write the integrated boundary condition \eqref{e:C-29} as $f
  = \mathcal{Q} [f]$.
\end{remark}

In order to study the mapping properties of $\mathcal{Q}$, we introduce
following intergral operators.

\begin{definition} \label{def:09}
   Let $p$ be meromorphic in $\mathbb{D} \setminus \{ 0, \eta_1,
  \eta_2 \}$ such that $p (z)$ grows at most like $z^{- 1}$ as $z \rightarrow
  0$. We define the integral operator $\mathcal{I}$ as
  \begin{equation}
    \mathcal{I} [p] (\eta) = \frac{1}{\mu (\eta) } \int_0^{\eta} \frac{\mu
    (\chi)}{\chi} p (\chi) d \chi  \label{e:C-34}
  \end{equation}
  and, moreover, the \textit{two} operators
  \begin{equation}
    \mathcal{J}_1 [p] = \int_{C_1} \frac{\mu (\chi)}{\chi} p (\chi) d \chi
    \quad \text{and} \quad \mathcal{J}_2 [p] = \int_{C_2} \frac{\mu
    (\chi)}{\chi} p (\chi) d \chi. \label{e:C-35}
  \end{equation}
  Here, $C_j$ is a loop in $\mathbb{D}$ that starts at $\eta = 0$, and
  encircles only $\eta_j$ exactly once in the positive direction, as displayed
  in Figure \ref{Fig:03}, (3). The path of integration in
  \eqref{e:C-35} is such that it is bounded away from the singularities
  $\eta_1$ and $\eta_2$ of $\mu$, and such that it fully lies within
  $\mathbb{D}_c$. Note that, due to the fact that $0$ is a branch point of $\mu(\chi)$, $C_1$ and $C_2$ are not closed loops.
\end{definition}

\begin{remark}
  \label{rem:C-10}Note that we have $\mathcal{Q} [f] = \mathcal{I} [p]$ for
  \[ p (\eta) = - W_1 (\eta) f (\eta_1) - W_2 (\eta) f (\eta_2) +
    i \Im (\alpha_1 f (\eta_1) + \alpha_2 f (\eta_2)) +
     \mathcal{P} [G] (\eta), \]
  and the path of integration in Equation \eqref{e:C-29} is then chosen
  according to Definition \ref{def:09}. That is, it lies fully within
  $\mathbb{D}_c$. The precise choice of path does not matter by Cauchy's
  theorem. 
\end{remark}

We note that Lemmas \ref{lem:C-04} and \ref{lem:write-up40} yield the
following result concerning the exponents $\nu_1$ and $\nu_2$. Since these
exponents feed directly into the definition of $\mu (\eta)$, and thus into the
definition of $\mathcal{Q}$, the following Corollary will play an important
role in proving the existence of solutions to the bubble problem.

\begin{corollary} \label{lem:C-11}
	Let $\alpha, \Gamma_0$ be real such that $ | b (\alpha,
  \Gamma_0) | $ is sufficiently small. Then, we have $\nu_j \in \mathbb{R}$ for $j = 1, 2$. If $b
  (\alpha, \Gamma_0) = 0$, then, $\nu_j = - 1$ for $j = 1, 2$.
\end{corollary}

Finally,  recall the Definition \ref{def:C-01} of $\mathcal{S}$ and
$\mathcal{S}^{\omega}$. Since $\mathcal{S}^{\omega}$ is not a Banach
space,\footnote{$\mathcal{S}^{\omega}$ is not complete, since the space of
polynomials is dense in $A (\mathbb{D})$.} we choose to work
within the space $\mathcal{S}$ instead. We will show in the
sequel that the equation $f = \mathcal{Q} [f]$ admits a solution $f \in
\mathcal{S}$ iff the parameters $(\gamma, \alpha, \Gamma_0)$ are constrained
according to \eqref{e:C-31}. The following lemma guarantees that whenever $|
\gamma | $ is sufficiently small, we have $f \in \mathcal{S}^{\omega}$. 

\begin{lemma} \label{lem:write-up43} 
  Let $| \gamma |$ be sufficiently
  small, and let $f = \mathcal{Q}[f] \in \mathcal{S}$. Then, the function $$y (\zeta) = \frac{1}{\zeta} \exp (F_0
  (\zeta^2) + f (\zeta^2))$$ extends analytically onto $\partial \mathbb{D}$,
  and, particularly, $f \in \mathcal{S}^{\omega}$.
\end{lemma}

\begin{proof}
  Differentiating \eqref{e:C-29} we find that \eqref{e:C-27} holds.
  Therefore, as we let $\eta \rightarrow \partial \mathbb{D}$ from within
  $\mathbb{D}$, we have, according to Lemma \ref{lem:write-up42}, 2.,
  \begin{equation}
    \eta f' (\eta) = G (\eta) + i \mathcal{H} [G] (\eta) - [U (\eta) f (\eta) + W_1
    f (\eta_1) + W_2 f (\eta_2)] + i \Im (\nu_1 f (\eta_1) +
    \nu_2 f (\eta_2)) \label{e:write-up128}
  \end{equation}
  where the Equality \eqref{e:write-up128} is taken in $L^2 (\partial
  \mathbb{D})$. Therefore, %since $y (\zeta) = \frac{1}{\zeta} \exp (F_0 + f)$
 % is continuous and non-vanishing on $\partial \mathbb{D}$, 
 we find that
  $y_{\zeta} \in L^2 (\partial \mathbb{D}) \cong \ell^2$.
  So, by Cauchy-Schwarz
  \begin{equation}
    \left( \sum_{n = 2}^{\infty} | a_n | \right)^2 \leq \frac{\pi^2}{6} 
    \sum_{n = 2}^{\infty} | n |^2 | a_n |^2 < \infty
  \end{equation}
  i.e. $\{a_j\}_{j=2}^{\infty} \in \ell^1$, so the a-priori regularity assumption \ref{A-priori} is satisfied. Moreover, it is readily checked that the
  boundary condition \eqref{e:write-up75} holds for $y$. Upon integrating
  \eqref{e:write-up75}, we therefore find that the function $Z (\zeta)$,
  normalised according to \eqref{e:Seq03}, with $\sqrt{Z_{\zeta}} = y$ solves
 the equation \eqref{e:write-up16}. Thus, according to Lemma \ref{lem:3.13}, whenever $| \gamma | $ is sufficiently small, $\{a_j\}_{j=2}^{\infty} \in \ell^1_{\rho}$ for some $\rho \in (0,1)$, and so $y$ extends analytically
  onto $\partial \mathbb{D}$.
\end{proof}

\subsection{Analyticity properties of $\mathcal{Q}$}\label{s:C-03}

According to the previous Sections \ref{sec:write-up4.1} and \ref{s:C-02}, our
goal is to show that there exists a fixed point $f = \mathcal{Q} [f] \in
\mathcal{S}$ of $\mathcal{Q}$. We would like to achieve this by applying
Banach's fixed point theorem. This requires us to first investigate the
analyticity properties of the operator $\mathcal{Q}$. Indeed, as we will see
in this section, we are not guaranteed that $\mathcal{Q} [f] \in \mathcal{S}$,
unless we impose some additional constraints.

\begin{notation}
  Throughout the remainder of this article, we will use indexing modulo 2 as
  follows. If $j = 1$, then $j - 1$=0=2, and if $j = 2$, then $j - 1$=1. 
\end{notation}

\begin{notation}
  Whenever we speak of a function `analytic in $\mathbb{D} \setminus \{ \eta_j
  \}$, it is presumed that $\eta_j$ is either a pole or a branch point. If
  $\eta_j$ is a branch point, analyticity in $\mathbb{D} \setminus \{ \eta_j
  \}$ implies multivaluedness, but we shall always work within simply
  connected domains $U \subset \mathbb{D} \setminus \{ \eta_j \}$, so we may
  always choose a single-valued branch in $U$. 
\end{notation}

Let us study the analyticity properties of $\mathcal{Q}$. To this end, we will
use the following Lemma, which is proved in Appendix \ref{App.A}.

\begin{lemma}
  \label{lem:C-14}Let $\nu_j < 0$. Then, for $p$ as in Definition
  \ref{def:09}, $\mathcal{I} [p]$ is analytic within $\mathbb{D}_c$.
  Moreover, the following statements about the operator $\mathcal{I}$ are
  true.
  \begin{enumerate}
    \item If $(\eta - \eta_j) p$ is analytic within
    $\mathbb{D} \setminus \{ \eta_{j - 1} \}$,
   then $\mathcal{I} [p] (\eta)$ is analytic at $\eta =
    \eta_j$ iff
    \begin{equation}
      \mathcal{J}_j [p] = 0 . \label{e:C-36}
    \end{equation}
    \item If $p$ is analytic at $\eta = 0$, then $\mathcal{I} [p]$ is analytic
    at $\eta = 0$. Moreover, $p (0) = 0$ iff $\mathcal{I} [p] (0) = 0$.
    \item If, additionally, $p$ is continuous on $\partial \mathbb{D}$  then
    $\mathcal{I} [p]$ is continuous on $\mathbb{T}_c$.
  \end{enumerate}
\end{lemma}
Note that whenever $(\eta - \eta_1) (\eta - \eta_2) p \in \mathcal{S} $ we
have the inequalities
\begin{equation}
  | \mathcal{I} [p ] (\eta) | \leq c_{\mathcal{I}} \sup_{\eta \in
  \mathbb{\partial \mathbb{D}}} | p (\eta)  | , ~ \text{ for } ~ \eta \in
  \mathbb{T}_c, ~ \text{ and } ~  | \mathcal{J}_j [p ] | \leq c_{\mathcal{} j}^{\mathcal{J}} \sup_{\eta \in
  	\partial \mathbb{D}} | p (\eta) |. \label{e:C-38}
\end{equation}
The constants $c_{\mathcal{I}}$ and $c_{\mathcal{} j}^{\mathcal{J}}$ in
\eqref{e:C-38} are independent of the locations of $\eta_1$
and $\eta_2$, since the paths of integration are chosen in $\mathbb{D}$ to be bounded away
from these singularities. In particular, the constants $c_{\mathcal{I}}$ and
$c_{j}^{\mathcal{J}}$ are independent of the parameters $\alpha$
and $\Gamma_0$.

Let us now use Lemma \ref{lem:C-14} to study the analyticity properties of
$\mathcal{Q}$. First, the following Corollary gives criteria for when
$\mathcal{Q} [f]$ is analytic at $\eta = \eta_1$ and $\eta = \eta_2$.

\begin{corollary}\label{cor:09} Let $| b (\alpha, \Gamma_0) |$ be sufficiently small, and recall the a-priori assumption \ref{Assumption4.13} on $\alpha$ and $\Gamma_0$, which we impose throughout Section \ref{sec:write-up04}.
  Then, $\mathcal{Q} [f] (\eta)$ is analytic at $\eta = \eta_1$ if, and only
  if,
  \begin{equation}
    \mathcal{J}_1 [- W_1 (\eta) f (\eta_1) - W_2 (\eta) f (\eta_2) +
   i \Im (  \nu_1 f
    (\eta_1) + \nu_2 f (\eta_2)) + \mathcal{P} [G] (\eta)] = 0. \label{e:C-39}
  \end{equation}
  $\mathcal{Q} [f]$ is analytic at $\eta = \eta_2$ if, and only if,
  \begin{equation}
    \mathcal{J}_2 [- W_1 (\eta) f (\eta_1) - W_2 (\eta) f (\eta_2) +
 i \Im (\nu_1 f
    (\eta_1) + \nu_2 f (\eta_2)) + \mathcal{P} [G] (\eta)] = 0. \label{e:C-40}
  \end{equation}
\end{corollary}

\begin{proof}
  Let
  \begin{equation}
    p (\eta) = - W_1 (\eta) f (\eta_1) - W_2 (\eta) f (\eta_2) +
 i \Im (  \nu_1 f
    (\eta_1) + \nu_2 f (\eta_2)) + \mathcal{P} [G] (\eta), \label{e:C-41}
  \end{equation}
  so that $\mathcal{Q} [f] = \mathcal{I} [p]$. The condition that $| b(\alpha, \Gamma_0) |$ is ``sufficiently small'' guarantees, by Lemma \ref{lem:C-11}, that $\nu_j <0$. Moreover, $(\eta - \eta_1) p (\eta)$ is analytic at $\eta = \eta_1$.
  So, Lemma \ref{lem:C-14} establishes analyticity of $\mathcal{Q} [f]$ at $\eta
  = \eta_1$ if, and only if, the Equation \eqref{e:C-39} is satisfied.
  Similarly, we find that $\mathcal{Q} [f]$ is analytic at $\eta = \eta_2$ if,
  and only if, the Equation \eqref{e:C-40} is satisfied.
\end{proof}

Next, we find the following criterion to guarantee analyticity of $\mathcal{Q}
[f]$ at $\eta = 0$.

\begin{corollary}
  \label{cor:C-15}If
  \begin{equation}
    \Re (\nu_1 f (\eta_1) + \nu_2 f (\eta_2)) - \mathcal{P}
    [G] (0) = 0 \label{e:C-42},
  \end{equation}
Then $\mathcal{Q} [f] (\eta)$ is analytic at
  $\eta = 0$ and $\mathcal{Q} [f] (0) = 0$.
\end{corollary}

\begin{proof}
  Again, let $p (\eta)$ be given by \eqref{e:C-41}, so that $\mathcal{Q} [f] =
  \mathcal{I} [p]$. The equation \eqref{e:C-42} guarantees that we have $p (0)
  = 0$. Therefore, Lemma \ref{lem:C-14}, guarantees that $\mathcal{Q} [f]$ is
  analytic at $\eta = 0$, and that it vanishes there.
\end{proof}

So, overall, we have found three constraints, \eqref{e:C-39}, \eqref{e:C-40},
and \eqref{e:C-42} which need to be satisfied in order for $\mathcal{Q} [f]$ to lie in $\mathcal{S}$. Also, these conditions are not independent as will transpire
later. It seems difficult to implement these constraints on the physical
parameters $\gamma, \alpha$, and $\Gamma_0$ directly. So, instead, we will now
pass to a different integral operator, $\mathscr{Q}$, where the
constraints \eqref{e:C-39} and \eqref{e:C-40} are automatically satisfied, but
whose fixed points, subject to \eqref{e:C-42}, are also fixed points of
$\mathcal{Q}$. So, we will have ``freed up'' two of the three constraints.

\subsection{Modified $\mathcal{Q}$-Operator}\label{sec:modQ}

Let us define an operator $\mathscr{Q} : \mathcal{S} \rightarrow
\mathcal{S}$ which will allow us to find non-trivial fixed points of
$\mathcal{Q}$. To that end, introduce some notation. We set
\begin{equation}
  \mathscr{W}_j (\eta) = (- W_j (\eta) + \nu_j), \quad \mathscr{P}
  [G] (\eta) = \mathcal{P} [G] (\eta) - \mathcal{P} [G] (0),
\end{equation}
and define the functionals $F_1 (f)$ and $F_2 (f)$ through
\begin{equation}
  \left(\begin{array}{c}
    F_1 (f)\\
    F_2 (f)
  \end{array}\right) = \frac{1}{J} \left(\begin{array}{cc}
    - \mathcal{J}_2 [\mathscr{W}_2 (\eta)] & \mathcal{J}_1 [\mathscr{W}_2
    (\eta)]\\
    \mathcal{J}_2 [\mathscr{W}_1 (\eta)] & - \mathcal{J}_1 [\mathscr{W}_1 (\eta)]
  \end{array}\right) \left(\begin{array}{c}
    \mathcal{J}_1 [\mathscr{P} [G] (\eta)] \mathcal{}\\
    \mathcal{J}_2 [\mathscr{P} [G] (\eta)]
  \end{array}\right) \label{e:C-44}
\end{equation}
with
\begin{equation}
  J = \mathcal{J}_1 [\mathscr{W}_1 (\eta)] \mathcal{J}_2 [\mathscr{W}_2 (\eta)] -
  \mathcal{J}_1 [\mathscr{W}_2 (\eta)] \mathcal{J}_2 [\mathscr{W}_1 (\eta)] .
  \label{e:write-up117}
\end{equation}
Recall that $G$ depends on $f$, which justifies our notation.

Now, we need to discuss whether $J \neq 0$ such that $F_1(f)$ and $F_2(f)$ are well-defined. To that end, let us assume that $\alpha > 0$ and $\Gamma_0 \in
\mathbb{R}$ are such that $b (\alpha, \Gamma_0) = 0$. First, note that from the
explicit formula \eqref{e:write-up93} for $b (\alpha, \Gamma_0)$, we find
$\alpha^2 \leq \frac{3}{7}$. Next, recall that by Corollary \ref{lem:C-11},
we have $\nu_1 = \nu_2 = - 1$ for $b (\alpha, \Gamma_0) = 0$. Thus, for $b
(\alpha, \Gamma) = 0$, the integrals $\mathcal{J}_n [\mathscr{W}_m (\eta)]$ may
be computed explicitly via residue calculus. From these computations, we find
that
\begin{equation}
  J = 0 \quad \text{iff} \quad L (t) : = 1375 t^3 + 2019 t^2 - 1107 t + 81 =
  0, \quad \text{for} \quad t = \alpha^2 \label{e:write-up118}
\end{equation}
so, since $\alpha^2 \leq \frac{3}{7}$, we obtain the following result.

\begin{lemma} \label{lem:write-up56} Let $t_{1, 2} \in \left( 0, \frac{3}{7} \right)$
  denote the two zeros of $L (t)$ within $\left( 0, \frac{3}{7} \right]$. Let
  $(\alpha,\Gamma_0) \in b^{-1}(0)$. Then, $F_1 (f)$ and $F_2 (f)$ are well-defined whenever
  \begin{equation}
    \alpha \neq \sqrt{t_j}, ~ \text{ and } ~ \Gamma_0
   \neq  \pm \sqrt{\frac{(3 - 5 t_j)^2 (3 - 7 t_j)}{9 (25 t_j + 3)}},
~ \text{ for } ~ j = 1, 2. \label{e:write-up119}
  \end{equation}
\end{lemma}

Let us denote
\begin{equation}
  \mathfrak{a}_j = \sqrt{t_j} ~ \text{ and } ~ \mathfrak{b}_j = \sqrt{\frac{(3 - 5
  t_j)^2 (3 - 7 t_j)}{9 (25 t_j + 3)}}, ~  j=1,2. \label{e:119}
\end{equation}
\begin{corollary} \label{cor:write-up57}
  Let $\alpha > 0$ and $\Gamma_0 \in \mathbb{R}$ be such
  that $| b (\alpha, \Gamma_0) |$ is sufficiently small. Then, $F_1 (f)$ and $F_2 (f)$ are
  well-defined so long as the pair $(\alpha, \Gamma_0)$ is bounded away from 
  $(\mathfrak{a}_j, \pm\mathfrak{b}_j)$ for $j = 1, 2$.
\end{corollary}

\begin{remark}
 Later, within the proof of Theorem \ref{thm:63}, we will require the numerical values of $t_1$ and $t_2$. These values can, of course, be computed explicitly since $L$ defines a cubic polynomial in $t$. However, we will only require following approximate numerical values. 
 	\begin{align}
 		t_1 \approx 0.0882171317364483 ~ \text{ and } ~ t_2 \approx 0.3502085039944982. \label{e:4.47}
 	\end{align}
\end{remark}

We can now define the following ``modified $\mathcal{Q} -$operator'', which we will call
$\mathscr{Q}$.

\begin{lemma}
  \label{lem:C-17} Let us assume that the Jacobian $J$ introduced in
  \eqref{e:write-up117} is non-vanishing. Let
  \begin{equation}
    \mathscr{Q} [f] = \mathcal{I} \left[\mathscr{W}_1 (\eta) F_1 (f) +
    \mathscr{W}_2 (\eta) F_2 (f) + \mathscr{P} [G] (\eta)\right].
  \end{equation}
  The operator $\mathscr{Q}$ maps $\mathcal{S}$ into  $\mathcal{S}$.
\end{lemma}

\begin{proof}
  We will use Lemma \ref{lem:C-14}. First, the conditions \eqref{e:C-36} in
  Lemma \ref{lem:C-14} are given by
  \begin{equation}
    \mathcal{J}_j \left[\mathscr{W}_1 (\eta) F_1 (f) + \mathscr{W}_2 (\eta) F_2 (f) +
    \mathscr{P} [G] (\eta)\right] = 0, \quad j = 1, 2. \label{e:C-47}
  \end{equation}
  But $F_1$ and $F_2$ do not depend on $\eta$ so, using linearity of
  $\mathcal{J}_1$ and $\mathcal{J}_2$, \eqref{e:C-47} is solved exactly when
  $F_1 (f) $and $F_2 (f)$ are as in \eqref{e:C-44}. Thus, $\mathscr{Q}$ is analytic at
  $\eta_1$ and $\eta_2$. Now, $\mathscr{W}_j (0) = 0$ for $j = 1, 2$ and
  $\mathscr{P} [G] (0) = 0$, so Lemma \ref{lem:C-14} also
  guarantees that $\mathscr{Q}$ is analytic at $\eta = 0$.
  Overall, we have shown that for $f \in \mathcal{S}$,
  $\mathscr{Q} [f]$ is analytic within $\mathbb{D}$
  
 It remains to show that
  $\mathscr{Q} [f]$ extends continuously onto $\partial \mathbb{D}.$ By Lemma \ref{lem:C-14}, we know that the term $\mathcal{I} [\mathscr{W}_1
  (\eta) F_1 (f) + \mathscr{W}_2 (\eta) F_2 (f)]$ is continuous $\mathbb{T}_c$,
  so it only remains to discuss continuity of $\mathcal{I}
  [\mathscr{P} [G] (\eta)]$ on $\partial \mathbb{D}$. Continuity
  of $\mathcal{I} [\mathscr{P} [G] (\eta)]$ on $\mathbb{T}_c$
  follows by the discussion given in Appendix \ref{App.A.2} Next, note that we
  have already established analyticity of $\mathscr{Q} [f]$ within
  $\mathbb{D}$. Thus, $\mathscr{Q} [f]$ is necessarily continuous
  across the cuts on $\mathbb{T}_c$, i.e., $\mathscr{Q} [f]$
  extends continuously onto $\partial \mathbb{D}$, so $\mathscr{Q}
  [f] \in \mathcal{S}$.
\end{proof}

The importance of $\mathscr{Q}$ stems from the following
proposition.

\begin{proposition}
  \label{prop:C-18} For $f \in \mathcal{S}$, $| b (\alpha, \Gamma_0) |$ sufficiently small,
  and $J \neq 0$, we have $f = \mathcal{Q} [f]$ if, and only if, both of the
  following two statements are true.
  \begin{enumerate}
    \item $f = \mathscr{Q} [f]$ is a fixed point of
    $\mathscr{Q}$.
    
    \item The constraint
    \begin{equation}
      \Re (\nu_1 f (\eta_1) + \nu_2 f (\eta_2)) - \mathcal{P}
      [G] (0) = 0 \label{e:C-48}
    \end{equation}
    is satisfied.
  \end{enumerate}
\end{proposition}

\begin{proof}
  First, let $f = \mathcal{Q} [f]$ be a fixed point of $\mathcal{Q}$. Then,
  the constraint $\eqref{e:C-48}$ must hold by \eqref{e:C-27} (cf. also Remark
  \ref{rem:C-07}). Thus, using \eqref{e:C-48}, we may re-write $\mathcal{Q}
  [f]$ as
  \begin{equation}
    \mathcal{Q} [f] = \mathcal{I} \left[\mathscr{W}_1 (\eta) f (\eta_1) + \mathscr{W}_2
    (\eta) f (\eta_2) + \mathscr{P} [G] (\eta)\right] .
  \end{equation}
  Since $\mathcal{Q} [f] = f \in \mathcal{S}$ is analytic at $\eta_1$ and also
  at $\eta_2$ we know by Corollary \ref{cor:09} that the conditions
  \eqref{e:C-39} and \eqref{e:C-40} must hold. Then, exactly as in the proof
  of the previous Lemma \ref{lem:C-17}, we find that the vector $(f (\eta_1),
  f (\eta_2))^T$ is given by the right-hand-side of \eqref{e:C-44}, so $f
  (\eta_1) = F_1 (f)$ and $f (\eta_2) = F_2 (f)$. Thus,
  \begin{equation}
    \mathcal{Q} [f] = \mathcal{I} \left[\mathscr{W}_1 (\eta) F_1 (f) + \mathscr{W}_2
    (\eta) F_2 (f) + \mathscr{P} [G] (\eta)\right] =
    \mathscr{Q} [f] .
  \end{equation}
  It remains to show that the equality $f = \mathscr{Q} [f]$ in
  conjunction with the constraint \eqref{e:C-48} implies that $f = \mathcal{Q}
  [f]$.
  
  Let $f = \mathscr{Q} [f]$, and let $p = \mathscr{W}_1 (\eta) F_1
  (f) + \mathscr{W}_2 (\eta) F_2 (f) + \mathscr{P} [G] (\eta)$. From
  Equation \eqref{e:C-73}, we find that
  \begin{equation}
    \mathcal{I} [p] (\eta_j) = - \frac{1}{r_j} \lim_{\eta \rightarrow \eta_j}
    (\eta - \eta_j) p (\eta)
  \end{equation}
  and so $F_j (f) = \mathscr{Q} [f] (\eta_j) = f (\eta_j)$ for $j
  = 1, 2$. Therefore, for a fixed point of $\mathscr{Q}$, we may
  use the equalities $F_j (f) = f (\eta_j)$ and find that
  \begin{equation}
    \mathscr{Q} [f] = \mathcal{I} \left[\mathscr{W}_1 (\eta) f (\eta_1) +
    \mathscr{W}_2 (\eta) f (\eta_2) + \mathscr{P} [G] (\eta)\right].
    \label{e:C-52}
  \end{equation}
  Now, using the constraint \eqref{e:C-48}, we see that the right-hand-side of
  \eqref{e:C-52} is indeed equal to $\mathcal{Q} [f]$ thus proving the Proposition.
\end{proof}

\subsection{The constraint \eqref{e:C-48}}\label{sec:4.5}

In order to prove the existence of solutions to the Problem (b), it remains
to show that $\mathscr{Q}$ has a fixed point in some sufficiently
small ball $B_{\delta}$ centred at the origin. Before we can show
this, however, we need to study the constraint \eqref{e:C-48} in more detail.
Particularly, we will require the fact that \eqref{e:C-48} implicitly defines
a real-analytic surface $\mathcal{C}_f$ whereon the parameters $(\gamma,
\alpha, \Gamma_0)$ must lie. Let us now prove this statement.

First, a direct computation using residue calculus shows that
\begin{equation}
  \mathcal{P} [r] (\eta) = B_0 R_1 (\eta) - \gamma R_2 (\eta)
\end{equation}
with
\begin{equation}
  R_1 (\eta) = \frac{2 \eta^2}{D (\eta)} + \sum_{j = 1}^2 \frac{\eta_j +
  \eta}{\eta_j - \eta} \frac{\eta_j}{D_{\eta}(\eta_j)}, \quad \text{and} \quad R_2
  (\eta) = A_0 + 2 i A_2 \eta + 2 A_4 \eta^2 .
\end{equation}
Here, $B_0$, $D (\eta)$, and $A_{0, 2, 4}$ are given by
\eqref{e:C-11}--\eqref{e:C-13}.

Therefore
\begin{equation}
  \mathcal{P} [G] (\eta) = - \frac{1}{4} (B_0 R_1 (\eta) - \gamma R_2 (\eta))
  + \frac{1}{2} \gamma \mathcal{P} [D \Re (f)] (\eta) + \frac{1}{4}
  \mathcal{P} [\boldsymbol{N} [f]] (\eta) . \label{e:C-82}
\end{equation}
Let
\begin{align}
  \frac{1}{4} R_2 (0) + \frac{1}{2} \mathcal{P} [D \Re (f)] (0) + \frac{1}{4}
  \mathcal{P} [\boldsymbol{N}_{\!\!2} [f]] (0) & = D_{\text{BC}} (\alpha, \Gamma_0 ;
  f), \label{e:4.59} \\
  \Re (\nu_1 f (\eta_1) + \nu_2 f (\eta_2)) + \frac{1}{4} (B_0 R_1 (0) -
  \mathcal{P} [\boldsymbol{N}_{\!\!1} [f]] (0)) & = N_{\text{BC}} (\alpha, \Gamma_0 ; f). \label{e:4.60}
\end{align}
Then, the condition $\Re(\alpha_1 f (\eta_1) + \alpha_2 f
(\eta_2)) - \mathcal{P} [G] (0) = 0$ is equivalent to
\begin{equation}
  \gamma D_{\text{BC}} (\alpha, \Gamma_0 ; f) - N_{\text{BC}} (\alpha, \Gamma_0 ;
  f) = 0. \label{e:0074}
\end{equation}

It can be checked directly that for $(\alpha, \Gamma_0) \in b^{- 1} (0)$ and $f
= 0$, we have $D_{\text{BC}} (\alpha, \Gamma_0 ; 0) \geq \kappa > 0$, where
$\kappa$ is independent of $(\alpha, \Gamma_0, f) .$ So, for $\| f \| < \delta$
sufficiently small, and for $| b (\alpha, \Gamma_0) | $ also sufficiently
small  such that the inequality \eqref{e:4.20} is satisfied, $D_{\text{BC}} (\alpha, \Gamma_0 ; f)$ is bounded away from $0$. This yields the following result.

\begin{lemma}
  \label{lem:024} For $\| f \|$ and $| b (\alpha, \Gamma_0) |$
  both sufficiently small, with $b(\alpha,\Gamma_0)$ such that the inequality \eqref{e:4.20} holds, the constraint \eqref{e:C-48} is satisfied whenever
  \begin{equation}
    \gamma = \frac{N_{\text{BC}} (\alpha, \Gamma_0 ; f)}{D_{\text{BC}} (\alpha,
    \Gamma_0 ; f)} . \label{e:00056}
  \end{equation}
  Recall that validity of \eqref{e:4.20} is needed to ensure that $\Gamma_0 \in \mathbb{R}$.
\end{lemma}

\begin{remark}
  Lemma \ref{lem:024} allows us fix the parameters $\alpha$ and $\Gamma_0$, but
  in doing so, we lose the freedom to enforce $\gamma \neq 0$. 
\end{remark}

Let us now show that we may, alternatively, choose to fix $\gamma \neq 0$, so
long as $| \gamma |$ is sufficiently small. Consider the equation
\eqref{e:0074}. Evaluated at $f = 0$, and with $b (\alpha, \Gamma_0) = 0$,
$\Gamma_0 \neq 0$, $0 < \alpha^2 < \frac{3}{7}$, we find
\begin{equation}
  \partial_{\Gamma_0} N_{\text{BC}} (\alpha, \Gamma_0 ; 0) = \frac{(25 \alpha^2 +
  3) (\alpha^2 + 3)}{(5 \alpha^2 - 3)^2} \frac{\Gamma_0}{4 \pi} \left(
  \frac{\eta_1}{D_{\eta}(\eta_1)} + \frac{\eta_2}{D_{\eta}(\eta_2)} \right) \neq 0.
  \label{e:write-up133}
\end{equation}

Let $C (\gamma, \alpha, \Gamma_0, f) :=\gamma D_{\text{BC}} (\alpha,
\Gamma_0 ; f) - N_{\text{BC}} (\alpha, \Gamma_0 ; f)$. Then, we have just shown
that
\begin{equation}
  \partial_{\Gamma_0} C (0, \alpha, \Gamma_0, 0) \neq 0, \quad \text{for} \quad
  (\alpha, \Gamma_0) \in b^{- 1} (0) .
\end{equation}
So, in particular, we may parametrise the set $\mathcal{C} \subset \mathbb{R}^3 \times \mathcal{S}$, given by
\begin{align}
	\mathcal{C} :=\{ (\gamma, \alpha, \Gamma_0 ; f) : \gamma D_{\text{BC}}
	(\alpha, \Gamma_0 ; f) - N_{\text{BC}} (\alpha, \Gamma_0 ; f) = 0, ~ | b (\alpha,
	\Gamma_0)| ~ \text{sufficiently small} \}, \label{e:4.61}
\end{align}
as a graph over $(\gamma, \alpha ; f)$, for $\gamma$ and $f$ sufficiently
small with $0 < \alpha^2 < \frac{3}{7}$ (particularly, the inequality \eqref{e:4.20} holds automatically for $b$ sufficiently small ensuring that $\Gamma_0 \in \mathbb{R}$). This gives the following Lemma.

\begin{lemma}
  \label{lem:0026}For $\| f \| $ and $| \gamma | $ both
  sufficiently small, and whenever $0 < \alpha^2 < \frac{3}{7}$, then the
  constraint \eqref{e:C-48} is satisfied when $\Gamma_0 = \Gamma_0 (\gamma,
  \alpha; f)$ is implicitly determined by \eqref{e:0074}.
\end{lemma}

Finally, let us show that we may also choose to fix $\gamma$ and $\Gamma_0$,
which will then determine the strain $\alpha$. We proceed exactly as before.
That is, computations show that for $(\alpha, \Gamma_0) \in b^{- 1} (0)$, we
have
\begin{equation}
  \partial_{\alpha} N_{\text{BC}} (\alpha, \Gamma_0 ; 0) = \frac{\alpha
  (\alpha^2 + 3) (875 \alpha^4 - 30 \alpha^2 - 189)}{9 (5 \alpha^2 - 3) (25
  \alpha^2 + 3)} R_1 (0) \label{e:write-up135}
\end{equation}
which is non-vanishing for $0 < \alpha^2 \leq \frac{3}{7}$. This gives
the following result.

\begin{lemma} \label{lem:0027}
  For $\| f \|$ and $| \gamma | $ both
  sufficiently small, and for $0 \leq | \Gamma_0 | < 1$, the constraint
  \eqref{e:C-48} is satisfied whenever $\alpha = \alpha (\gamma, \Gamma_0; f)$
  is implicitly determined by \eqref{e:0074}. 
\end{lemma}

Finally, since $N_{\text{BC}}$ and $D_{\text{BC}}$  in equations \eqref{e:4.59} and \eqref{e:4.60}, respectively, depend real-analytically on
$(\gamma, \alpha, \Gamma_0; f),$ we obtain the following result.

\begin{proposition}
  For any given $f \in \mathcal{S}^{\omega}$, the surface
  $\mathcal{C}_f \subset \mathbb{R}^3$ defined by
  \begin{equation}
    \mathcal{C}_f :=\{ (\gamma, \alpha, \Gamma_0): \gamma D_{\text{BC}} (\alpha, \Gamma_0 ; f)
    - N_{\text{BC}} (\alpha, \Gamma_0 ; f) = 0, ~ | b (\alpha, \Gamma_0) | \text{ sufficiently small, } 0 \leq | \Gamma_0 | < 1 \}
    \label{e:write-up140}
  \end{equation}
  is real analytic. For $|\gamma|$ sufficiently small and $0 \leq | \Gamma_0 | < 1$, this surface may be parametrised as a graph over
  $(\gamma, \Gamma_0)$, or, whenever $\Gamma_0 \neq 0$, as a graph over $(\gamma,
  \alpha)$. Since we consider only real parameters $\gamma, \alpha$, and $\Gamma_0$,  whenever $(\gamma, \alpha, \Gamma_0) \in \mathcal{C}_f \subset \mathbb{R}^3$, the inequality \eqref{e:4.20} necessarily holds.
\end{proposition}

\begin{definition} \label{rem:60}
	Recall that $(\lambda, \sigma) = \left( \frac{1}{3}
	\alpha^2, 2 i \Gamma_0 \alpha \right)$. We will write $(\gamma, \lambda, \sigma) \in \tilde{\mathcal{C}}_f$ whenever
		\begin{align}
			\left(\gamma, \sqrt{3 \lambda}, \frac{\sigma}{2i \sqrt{3 \lambda}}\right) \in \mathcal{C}_f, 
		\end{align}
 for $\alpha := \sqrt{3 \lambda}$ and $\Gamma_0 :=  \frac{\sigma}{2i \sqrt{3 \lambda}}$.
\end{definition}

We will discuss the constraint \eqref{e:C-48} in more detail in Section
\ref{sec:05}, after we have proved the existence of solutions to Problem (b).

\subsection{Existence of solutions to the bubble problem for $\gamma \neq
0$.}\label{sec:4.6}

Let us show that for $| \gamma | $ sufficiently small, the operator
$\mathscr{Q} : \mathcal{S} \rightarrow \mathcal{S}$ is a
contraction on some small ball $B_{\delta} \subset \mathcal{S}$.

\begin{lemma}
  \label{lem:C-20}Let $g$ be continuous on $\partial \mathbb{D}$. We have $|
  \mathcal{J}_j [\mathcal{P} [g]] | \leq C_j^{\mathcal{J}} \| g \| $ for a
  constant $C_{j}^{\mathcal{J}}$ that is independent of $\alpha$, $\Gamma_0$,
  and $\gamma$. Here, the norm is the sup-norm taken on $\partial \mathbb{D}$.
\end{lemma}

\begin{proof}
  We follow the proof of Lemma 31 in the supplementary material of \cite{CrowdyTanveer2025}. By
  Fubini's theorem,
  \begin{equation}
    \mathcal{J}_j [\mathcal{P} [g]] = \int_{| \chi | = 1} \left( \int_{C_j}
    \left( \frac{\chi + \eta}{\chi - \eta} \right) \frac{\mu (\eta)}{\eta} d
    \eta \right) \frac{g}{\chi} d \chi . \label{e:C-53}
  \end{equation}
  Now, choose the path $C_j$ of integration such that it is bounded away from
  $\eta = \eta_j$ and $\partial \mathbb{D}$. This guarantees that the function
  \begin{equation}
    g_j (\chi) = \int_{C_j} \left( \frac{\chi + \eta}{\chi - \eta} \right)
    \frac{\mu (\eta)}{\eta} d \eta
  \end{equation}
  is bounded and analytic in a neighbourhood of $\partial \mathbb{D}$. The
  statement of the Lemma now follows by inspection of \eqref{e:C-53}.
\end{proof}

Moreover, by \eqref{e:076}, for $g$ be continuous on $\partial \mathbb{D}$,
we have
\begin{equation}
  | \mathcal{I} [\mathcal{P} [g]] (\eta) | \leq C_j^{\mathcal{I}} \|
  \mathcal{} g \|, \quad \text{for } \quad \eta \in \mathbb{T}_c .
\end{equation}
For brevity, we will henceforth use the notation $a \lesssim b$ to indicate
that $a \leq C b$ for a constant $C$ independent of $(\gamma, \alpha,
\Gamma_0)$. Note that by  \eqref{e:C-38}, we have
\begin{equation}
  | \mathcal{I} [p ] (\eta) | \lesssim \| p  \| , \quad \text{for } \eta \in
  \mathbb{T}_c \text{, \quad and\quad} | \mathcal{J}_j [p ] | \lesssim \| p 
  \| .
\end{equation}
Now, let us assume that the parameters $\alpha$ and $\Gamma_0$ are bounded away
from $\mathfrak{a}_j$ and $\pm \mathfrak{b}_j$ (defined according to \eqref{e:119}), such that
the Jacobian $J$ introduced in \eqref{e:write-up117} is non-vanishing. This
ensures that the operator $\mathscr{Q}$ is well-defined.

By \eqref{e:C-44}, and using the formula \eqref{e:C-82} for $\mathcal{P} [G]$,
we have
\begin{equation}
  \| F_1 (f^1) - F_1 (f^2) \| \lesssim \delta \| D \Re (f^1) - D \Re (f^2) \|
  + \delta \| \boldsymbol{N}_{\!\!2} [f^1] - \boldsymbol{N}_{\!\!2} [f^2] \| + \|
  \boldsymbol{N}_{\!\!1} [f^1] - \boldsymbol{N}_{\!\!1} [f^2] \| 
\end{equation}
but $\| \boldsymbol{N}_j [f^1] - \boldsymbol{N}_j [f^2] \| \lesssim \| f^1 + f^2
\| \| f^1 - f^2 \|  \leq 2 \delta \| f^1 - f^2 \|$, for $\| f^j \| < \delta$
sufficiently small, so
\begin{equation}
  \| F_1 (f^1) - F_1 (f^2) \| \lesssim \delta \| f^1 - f^2 \|, \quad
  \text{for} \quad f^1, f^2 \in B_{\delta} .
\end{equation}
Therefore
\begin{equation}
  \| \mathscr{Q} [f^1] - \mathscr{Q} [f^1] \| \lesssim
  \sum_{j = 1}^2 \| F_j (f^1) - F_j (f^2) \| + \delta \| f^1 - f^2 \| \lesssim
  \delta \| f^1 - f^2 \|, \quad \text{for} \quad f^1, f^2 \in
  B_{\delta} . \label{e:C-59}
\end{equation}

Next, note %that for $(\gamma,\alpha,\Gamma_0)\in \mathcal{C}_0$, 
\begin{equation}
  \| \mathscr{Q} [0] \| \lesssim | B_0 | + | \gamma |,
\end{equation}
and recall that for $b (\alpha, \Gamma_0) = 0$, we have $B_0 = 0$. So, we let
$(\alpha, \Gamma)$ be such that $| b (\alpha, \Gamma_0) |$ is sufficiently small, choose $\gamma
\neq 0$ to be sufficiently small, and let $\delta = 2 \| \mathscr{Q}
[0] \|$, so that $\| \mathscr{Q} [f] \| \leq \delta$ and
$\mathscr{Q}$ maps $B_{\delta}$ onto
$B_{\delta}$. Since we may choose both, $| B_0 |$ and $| \gamma |$
sufficiently small, we are guaranteed that $\delta$ is small enough for all
estimates formulated in \eqref{e:C-53}--\eqref{e:C-59} to be true. So,
$\mathscr{Q} : B_{\delta} \rightarrow
B_{\delta}$ is a contraction. This enables us to prove the following
Theorem concerning the existence of solutions to Problem (b).

\begin{theorem} \label{thm:63}
  Recall the Definition \eqref{e:4.61} of
  $\mathcal{C}$, and let $0 \leq | \Gamma_0 | < 1$. The Equation \eqref{e:write-up75} has a unique solution \
  $y = \frac{1}{\zeta} \exp (F_0 + f)$ with $f \in B_{\delta} \cap
  \mathcal{S}^{\omega}$ if, and only if, $(\gamma, \alpha, \Gamma_0; f) \in
  \mathcal{C}$.
\end{theorem}

\begin{proof}
  We have already shown that for $y (\zeta)$ to solve \eqref{e:write-up75},
  the parameters $(\gamma, \alpha, \Gamma_0) $ necessarily need to be constrained
  according to \eqref{e:C-48}. Thus, $(\gamma, \alpha, \Gamma_0; f) \in
  \mathcal{C}$ is a necessary condition for Equation \eqref{e:write-up75} to
  admit a solution. Let us now show that $(\gamma, \alpha, \Gamma_0; f) \in
  \mathcal{C}$ is also sufficient. To that end, recall the definition of the
  Jacobian $J$ in \eqref{e:write-up117} and recall that its zeros depend on
  $(\gamma, \alpha, \Gamma_0)$ and $f$. We need to distinguish two cases.
  First, we consider the case  $(\gamma, \alpha, \Gamma_0; f) \in \mathcal{C}$
  such that $J \neq 0$ and show that solution $y (\zeta)$ exists. Second, we
  will discuss how to analytically continue this solution $y (\zeta)$ to the
  case when $J = 0$.
  
  \textbf{Case 1} ($J \neq 0$). Recall the definition of $\mathcal{C}_f$ in \eqref{e:write-up140}.  By the previous discussion,
  $\mathscr{Q}$ admits a unique fixed-point $f \in
  B_{\delta}$ when $| \gamma |$ is sufficiently small. Thus,
  for $(\gamma, \alpha, \Gamma_0) \in \mathcal{C}_f$, Proposition
  \ref{prop:C-18} ensures that $f = \mathcal{Q} [f]$ is a fixed-point of
  $\mathcal{Q}$. By Lemma \ref{lem:write-up43}, we find $f \in
  \mathcal{S}^{\omega}$ and it may be checked directly that $y =
  \frac{1}{\zeta} \exp (F_0 + f)$ solves Equation \eqref{e:write-up75}.
 
  \textbf{Case 2} ($J = 0$). This case will be dealt with via analytical
  continuation from Case 1. Consider function $y (\zeta)$ solving Equation
  \eqref{e:write-up75} in Case 1. Particularly, upon integrating Equation
  \eqref{e:write-up75}, we find that the corresponding function $Z (\zeta)$ solves the equation \eqref{e:write-up16}  with 
  	\begin{align}
  		\bar{K}(\zeta) = \frac{\bar{y}(1/\zeta)}{\zeta y(\zeta)} - \gamma \bar{Z}(1/\zeta)    \label{e:seq014}
  	\end{align}
  analytic on $\partial \mathbb{D}$.
	On the other hand, since the Equation \eqref{e:write-up16} is obtained from
  integration of \eqref{e:write-up75}, we find that necessarily,
  \begin{equation}
    \frac{d}{d \zeta} \bar{K} (\zeta) = - \left( \frac{d W}{d \zeta} \right)^2
    \frac{1}{Z_{\zeta}} = - \left( \alpha \left( \zeta -
    \frac{1}{\zeta^3} \right) + \frac{i \Gamma_0}{ \zeta} \right)^2
    \frac{1}{Z_{\zeta} (\zeta)}. \label{e:145}
  \end{equation}
  Combining \eqref{e:seq014} and \eqref{e:145}, we thus have
  \begin{equation}
   \frac{d}{d \zeta} \left( \frac{\bar{y}(1/\zeta)}{\zeta y(\zeta)} - \gamma \bar{Z}\left(\frac{1}{\zeta}\right)\right)= - \left( \alpha \left( \zeta -
    \frac{1}{\zeta^3} \right) + \frac{i \Gamma_0}{ \zeta} \right)^2
    \frac{1}{Z_{\zeta} (\zeta)} \label{e:write-up053} .
  \end{equation}
	Now, since both, the left and the right-hand-side in \eqref{e:write-up053} extend holomorphically past $\partial \mathbb{D}$, the equation \eqref{e:write-up053} is valid within a neighbourhood of $\partial \mathbb{D}$ and extends, in fact, to $\mathbb{D} \setminus \{0\}$ as we shall discuss within the proof of Theorem \ref{seq:thmConformal}.  
	In order to emphasize the dependence on the parameters $(\gamma, \alpha, \Gamma_0;f)$, denote
  \begin{equation*}
	A (\gamma, \alpha, \Gamma_0; f) =   \frac{d}{d \zeta} \left( \frac{\bar{y}(1/\zeta)}{\zeta y(\zeta)} - \gamma \bar{Z}\left(\frac{1}{\zeta}\right)\right) ~\text{ and }~    B (\gamma, \alpha, \Gamma_0; f) = - \left( \alpha \left( \zeta -
	\frac{1}{\zeta^3} \right) + \frac{i \Gamma_0}{ \zeta} \right)^2
	\frac{1}{Z_{\zeta} (\zeta)}.
	\end{equation*}
	
We let $j=1$ or $j=2$, and consider the point $(\alpha,\Gamma_0) = (\mathfrak{a}_j, \pm \mathfrak{b}_j) $ at which the operator $\mathscr{Q}$ is ill-defined, as explained at the beginning of Section \ref{sec:modQ}. Using the numerical values provided within \eqref{e:4.47}, computations show that the inequality \eqref{e:write-up62} holds at the points $(\lambda_j,  \pm \sigma_j) =  \left( \frac{1}{3} \mathfrak{a}_j^2, \pm 2 i \mathfrak{b}_j \mathfrak{a}_j
\right)$. Thus, according to  Proposition \ref{prop:23}, for sufficiently small $\gamma$ and for $(\alpha, \Gamma_0)$ sufficiently close to $(\mathfrak{a}_j, \pm \mathfrak{b}_j)$, the function $Z(\zeta)$ is univalent and consequently $Z(\zeta)$ is the locally unique solution to Problem (a). Finally, whenever $|\gamma|$ is sufficiently small and the parameters $(\lambda, \sigma)$ are such that Problem (a) admits a solution let us set
  \begin{equation}
    q = \log (\zeta y (\zeta)) - F_0 \in \mathcal{S}^{\omega}, ~ y =
    \sqrt{Z_{\zeta}} ~ \text{ with $Z(\zeta)$ locally unique solution to Problem  (a).}
    \label{e:write-up152}
  \end{equation}
  When $J \neq 0$, local uniqueness of solutions to Problem (a) gives $q = f
  = \mathscr{Q} [f]$ but \textit{the Definition
  \eqref{e:write-up152} makes sense even when}  $(\lambda,\sigma) = (\lambda_j,\sigma_j)$, i.e., when $(\alpha,\Gamma_0) = (\mathfrak{a}_j, \pm \mathfrak{b}_j)$.  
  
We observe that, for  $(\alpha, \Gamma_0)$ sufficiently close to $(\mathfrak{a}_j, \pm \mathfrak{b}_j)$, the following statements are true.
  \begin{enumerate}
    \item The function $q$ depends real-analytically on $(\alpha, \Gamma_0,
    \gamma)$ (cf. Corollary \ref{cor:write-up20} and the relation $(\lambda,
    \sigma) = \left( \frac{1}{3} \alpha^2, 2 i \Gamma_0 \alpha
    \right)$ linking $(\lambda, \sigma)$ to $(\alpha, \Gamma_0)$).
    
    \item $\mathscr{Q}$ depends real-analytically on $(\gamma,
    \alpha, \Gamma_0; q)$. 
    
    \quad Since the constraint \eqref{e:C-48} is also
    real-analytic, $y$, $Z$, $A$, and $B$ are all real-analytic in
    $(\gamma, \alpha, \Gamma_0)$. So, for $\zeta$ in a neighbourhood of $\partial \mathbb{D}$, $A = B$ is a real-analytic identity in the variables $(\gamma,\alpha, \Gamma_0)$, and it is a-priori valid
	on open, non-empty subsets of the real-analytic surface
    $\mathcal{C}_q$. From these open subsets, we have, for the moment, excluded the
    points $(\mathfrak{a}_j, \pm \mathfrak{b}_j)$ whereon $J $ vanishes.
    
    \item According to Theorem \ref{thm:seq09}, for $|\gamma|$ sufficiently small, a solution to Problem (a)
    exists whenever $(\gamma, \alpha, \Gamma_0)$, corresponding to $(\lambda,
    \sigma, \gamma)$, lie within $\tilde{\mathcal{C}}_q$, \textit{irrespective of whether $(\alpha, \Gamma_0) =
    (\mathfrak{a}_j, \pm \mathfrak{b}_j)$}.
  \end{enumerate}
 From these observations, we find, by real-analytic continuation, that the conformal mapping identity $A
  = B$ is valid for the locally unique solution $Z (\zeta)$ to Problem (a)
  whenever $(\gamma, \alpha, \Gamma_0) \in \mathcal{C}_q$ and $(\alpha, \Gamma_0)$ are close \textit{or equal to} $(\mathfrak{a}_j, \pm \mathfrak{b}_j)$. Finally, using the fact that $Z (\zeta)$ solves
  Problem (a), differentiation of Equation \eqref{e:write-up16} with respect to $\zeta$ \textit{in conjunction with the conformal mapping identity $A=B$} shows that
  $y(\zeta)$ solves Equation \eqref{e:write-up75} whenever $(\alpha,\Gamma_0)$ are close \textit{or equal to} $(\mathfrak{a}_j, \pm \mathfrak{b}_j)$.
\end{proof}

As a corollary, we obtain the following result.

\begin{theorem}
  \label{theorem:bubble}
  Let the parameters $(\gamma, \alpha, \Gamma_0)$ be in a range where the solution $Z(\zeta)$ to Problem (a) is univalent  (as discussed in Section \ref{sec:write-up3.4}), and let $0 \leq | \Gamma_0 | < 1$. Let $q$ be defined according to \eqref{e:write-up152}, and assume that the constraint $(\gamma, \alpha, \Gamma_0) \in \mathcal{C}_q$ is satisfied. Then, the Bubble
  Problem b) admits a locally unique solution. For this solution, according to Theorem \ref{thm:seq09}, the
  coefficients $\{ c_j \}_{j = 2}^{\infty}$ of the conformal map
  \begin{equation}
    Z (\zeta) = - \zeta^{- 1} + \sum_{j = 2}^{\infty} c_j \zeta^{j - 1} \label{e:4.77}
  \end{equation}
  satisfy the following symmetry properties. $c_{2 j + 1} \equiv 0$ for all
  $j$, $c_{4 j} \in \mathbb{R}$ and $c_{4 j - 2} \in i \mathbb{R}$. Moreover,
  the conformal mapping identity \eqref{e:write-up053} holds for $Z (\zeta)$,
  the boundary curve $\partial \Omega$ is real-analytic, and $Z (\zeta)$
  extends holomorphically past $\partial \mathbb{D}$.
\end{theorem}

The previous proof introduced concepts that are of interest in their own
right, like the conformal mapping identity \eqref{e:write-up053}. As it turns
out, the identity \eqref{e:write-up053} transcends the bubble problem and links the two mathematical problems (a) and (b).
The conformal mapping identity and its importance is discussed in more detail
in the following section.

\section{Conformal mapping identity and scalar constraint }\label{sec:05}

Due to its independent interest, let us formulate the conformal mapping
identity \eqref{e:write-up053} as a theorem. To do so, recall the notation
$(\gamma, \lambda, \sigma) \in \tilde{\mathcal{C}}_q$ introduced in Definition
\ref{rem:60}. Particularly, in writing $(\gamma, \lambda, \sigma) \in
\tilde{\mathcal{C}}_q,$ it is implied that $\sigma \in i \mathbb{R}$.

\begin{theorem} \label{seq:thmConformal} 
  Let $(\gamma, \lambda, \sigma)$ be in a range where Problem (a) has a locally unique solution  $Z(\zeta)$.  Assume, moreover, that $(\gamma, \lambda, \sigma)
  \in \tilde{\mathcal{C}}_q$, with $q$ defined according to \eqref{e:write-up152}. 
 The following conformal mapping identity is valid for the function $Z(\zeta)$.
  \begin{align}
 \frac{d}{d \zeta} \left( \frac{\bar{y}(1/\zeta)}{\zeta y(\zeta)} - \gamma \bar{Z}(1/\zeta) \right) = - \left( \sqrt{3 \lambda} \left( \zeta -
 \frac{1}{\zeta^3} \right) + \frac{\sigma}{2 \sqrt{3 \lambda} \zeta}
 \right)^2 \frac{1}{Z_{\zeta} (\zeta)}, ~ \text{ for } ~ \zeta \in \mathbb{D}\setminus \{0\}. \label{e:write-up53} 
  \end{align}
  The identity  \eqref{e:write-up53} extends, moreover, to an open neighbourhood of $\partial \mathbb{D} \subset \mathbb{C}$.
\end{theorem}

\begin{proof}
	The identity in a neighbourhood of $\partial \mathbb{D}$ is the equation \eqref{e:write-up053} re-written in terms of $\sigma$ and $\lambda$, so it remains only to discuss the case when $0<|\zeta|<1$. From \eqref{e:write-up17}, we find that the function $\bar{K} (\zeta) - (\bar{k}_{- 2}
	\zeta^{- 3} + \bar{k}_0 \zeta^{- 1})$ is holomorphic within $\mathbb{D}$ and consequently the expression $\frac{\bar{y}(1/\zeta)}{\zeta y(\zeta)} - \gamma \bar{Z}(1/\zeta)$ extends analytically onto $\mathbb{D}\setminus \{0\}$.
\end{proof}

Importantly, the validity of \eqref{e:write-up53} solely depends on the fact
that $(\gamma, \lambda, \sigma) \in \tilde{\mathcal{C}}_q$, and on the Problem (a)
formulation. The identity \eqref{e:write-up53} is thus also satisfied by
solutions to the Vortex Problem (i) whenever $(\gamma, \lambda, \sigma) \in
\tilde{\mathcal{C}}_q$ for $(\lambda, \sigma) = (- \hat{\alpha}, - \bar{\beta} +
\gamma)$. Moreover, as we have seen in the proof of Theorem \ref{thm:63},
solutions to Problem (a) that satisfy the identity \eqref{e:write-up53} yield
solutions to Problem (b). More precisely, the following lemma is true.

\begin{lemma}
  \label{seq:lemConformal}Let $\sigma \in i \mathbb{R}$, $\lambda >0$ and $\gamma$ be in a range were Problem (a) has a locally unique solution $Z : \mathbb{D} \rightarrow \Omega$. Then, $Z (\zeta)$ is also the
  locally unique solution to Problem (b) if, and only if, the conformal
  mapping identity \eqref{e:write-up53} holds. I.e., if, and only if,
  $(\gamma, \lambda, \sigma) \in \tilde{\mathcal{C}}_q$, with $q$ defined according to
  \eqref{e:write-up152}.
\end{lemma}

Particularly, for $(\gamma, \lambda, \sigma) \in \tilde{\mathcal{C}}_q$, with $q$
defined according to \eqref{e:write-up152}, such that solution $Z (\zeta)$ to
Problem (a) is univalent, the Vortex Problem (i) and Bubble Problem (ii) admit
the same boundaries $\partial \Omega$ and become mathematically equivalent in
the sense, that Problems (a) and (b) become equivalent. Numerical evidence
supplementing these theoretical findings is provided in \cite{CKT}.
Note, however, that despite this mathematical analogy the observed streamlines
on which $\psi \equiv \text{const} \neq 0$ generally differ for Vortex Problem
(i) and Bubble Problem (ii), even if \eqref{e:write-up53} holds, since the two streamfunctions are different (cf. discussion of Section
\ref{sec:write-up2.2} and Figure \ref{Fig.01}).

Finally, let us note that, according to the discussion of Section
\ref{sec:write-up2.2}, for any solution to Problem (b), the two a-priori
different expressions for $k_2$ obtained from \eqref{e:write-up12} and
\eqref{e:seq09} agree. This results in the constraint \eqref{e:152}, which
appears weaker compared to the constraint $(\gamma, \lambda, \sigma) \in
\tilde{\mathcal{C}}_q$ since it only relies on the $k_2$-coefficient whereas,
according to \eqref{e:16}, the identity \eqref{e:write-up53} takes into
account all of the $k_j$ coefficients. However, as we will now prove, these
two seemingly different constraints are, in fact, locally equivalent.

\begin{lemma}
  \label{lem:67}Let $\boldsymbol{a}$ and $\boldsymbol{c} = \mathcal{T}
  [\boldsymbol{a}]$ be the series corresponding to the locally unique solution
  $Z (\zeta)$ to Problem (a). Let $\tilde{\alpha}$ and $\tilde{\Gamma}_0$  be defined according to $\tilde{\alpha} := \sqrt{3 \lambda}$ and $\tilde{\Gamma}_0 := \frac{\sigma}{2 i \sqrt{3 \lambda}}$.\footnote{We use the $\tilde{\phantom{\Gamma}}$ notation to emphasize the fact that this Lemma is valid for solutions to Problem (a). It automatically holds for solutions to Problem (b).} Assume that $ | \tilde{\Gamma}_0 | < 1$,  $\tilde{\Gamma}_0 \in \mathbb{R}$, and that $(\gamma, \tilde{\alpha}, \tilde{\Gamma}_0)$ be constrained according to
  \begin{equation}
    2 \tilde{\alpha}^2 +\tilde{\Gamma}_0^2 - 4 i \tilde{\Gamma}_0- (3 a_2^2
    - 2 a_4) \tilde{\alpha}^2 - 1 - \gamma + \gamma \sum_{m = 2}^{\infty} \bar{c}_m
    a_m + \frac{1}{3} a_4 \tilde{\alpha}^2 - 2 \bar{a}_2 \tilde{\alpha}^2 +2  i \tilde{\Gamma}_0 \tilde{\alpha} = 0. \label{e:152}
  \end{equation}
  Then, for $| \gamma |$ small and $(\lambda, \sigma)$ such that $| b (\tilde{\alpha}, \tilde{\Gamma}_0) |$ is also sufficiently small, we have $(\gamma, \lambda, \sigma) \in \tilde{\mathcal{C}}_q$ with
  $q$ defined according to \eqref{e:write-up152}.
\end{lemma}

\begin{remark}
  Since \eqref{e:152} is obtained by imposing a consistency relation on $k_2$,
  one may think of it as \textit{scalar constraint}. The conformal mapping
  identity \eqref{e:write-up53}, on the other hand, is obtained by imposing
  consistency relations amongst all the $k_j$, so we shall think of it as    \textit{functional constraint.}
\end{remark}

\begin{proof}[Proof of Lemma \ref{lem:67}]
  The proof is similar to the discussion of Section \ref{sec:4.5}, so we will
  be somewhat brief. Let us denote the left-hand-side of \eqref{e:152} by $C_2
  (\gamma, \tilde{\alpha}, \tilde{\Gamma}_0; \boldsymbol{a}),$ so the constraint \eqref{e:152}
  becomes rewritten as $C_2 (\gamma, \tilde{\alpha}, \tilde{\Gamma}_0; \boldsymbol{a}) = 0$. Let
  $\boldsymbol{a}^0$ denote the series corresponding to the $\gamma = 0$ exact
  solution (cf. Equation \eqref{e:write-up50}). Then, a routine computation
  shows that
  \begin{equation}
    \partial_{\tilde{\alpha}} C_2 (0,\tilde{\alpha},\tilde{\Gamma}_0; \boldsymbol{a}^0) \neq 0
  \end{equation}
  whenever $b (\tilde{\alpha}, \tilde{\Gamma}_0) = 0$. Since for $b (\tilde{\alpha}, \tilde{\Gamma}_0) = 0$ we
  have $C_2 (0, \tilde{\alpha}, \tilde{\Gamma}_0; \boldsymbol{a}^0) = 0$, the implicit function
  theorem grants the existence of a solution $\tilde{\alpha}  =\tilde{\alpha} (\gamma,
  \tilde{\Gamma}_0; \boldsymbol{a})$ to \eqref{e:152}. Let
  \begin{equation}
    \mathcal{C}_{\boldsymbol{a}} = \left\{ (\gamma,\tilde{\alpha}, \tilde{\Gamma}_0): C_2 (\gamma, \tilde{\alpha},\tilde{\Gamma}_0;
    \boldsymbol{a} ) = 0, ~ | b (\tilde{\alpha},\tilde{\Gamma}_0) | ~ \text{ sufficiently small}, ~ | \tilde{\Gamma}_0 | < 1 \right\} .
  \end{equation}
  Now, for $|\gamma|$ sufficiently small, we have $\mathcal{C}_q \subset \mathcal{C}_{\boldsymbol{a}}$ since any solution
  to Problem (b) necessarily satisfies the constraint \eqref{e:152}. So, for
  $| \gamma |$ small and $| \tilde{\Gamma}_0 | < 1$, we find
  that $\tilde{\alpha} (\gamma,\tilde{\Gamma}_0; \boldsymbol{a}) = \alpha (\gamma, \tilde{\Gamma}_0; q)$ as
  functions of $(\gamma, \tilde{\Gamma}_0)$ where $\alpha(\gamma, \tilde{\Gamma}_0; q)$ denotes
  the function defined within the formulation of Lemma \ref{lem:0027} for $\tilde{\Gamma}_0 = \Gamma_0$. Thus,
  locally, in a neighbourhood of $\{\gamma =0\} \times
  b^{- 1} (0),$ $\mathcal{C}_q = \mathcal{C}_{\boldsymbol{a}}$.\footnote{We do not claim that $\mathcal{C}_q = \mathcal{C}_{\boldsymbol{a}}$ \textit{globally}. }
\end{proof}

Of course, we may write the equation \eqref{e:152} directly in terms of $\lambda$ and $\sigma$ which yields the constraint 
		\begin{align}
		4 \lambda  - \frac{\sigma^2}{12 \lambda} - 2 \frac{\sigma}{\sqrt{3 \lambda}}- (3 a_2^2
		- 2 a_4) 3 \lambda - 1 - \gamma + \gamma \sum_{m = 2}^{\infty} \bar{c}_m
		a_m + a_4 \lambda - 6 \bar{a}_2\lambda + \sigma = 0 \label{e:5.5}
	\end{align}
on the parameters $(\gamma,\lambda,\sigma)$.

\begin{corollary}\label{cor:5.5}
  Let $| \gamma|$ be small and let $(\lambda, \sigma)$ be such that, for the parameters $\tilde{\alpha}$ and $\tilde{\Gamma}_0$ defined according to $\tilde{\alpha} := \sqrt{3 \lambda}\in (0,\infty)$ and $\tilde{\Gamma}_0 := \frac{\sigma}{2i \sqrt{3 \lambda}} \in \mathbb{R}$,  $| b
  (\tilde{\alpha}, \tilde{\Gamma}_0) |$ is sufficiently small and $ | \tilde{\Gamma}_0|  < 1$. Let $(\gamma,\lambda,\sigma)$ be in a range where Problem (a) has a locally unique solution $Z (\zeta)$, and let the scalar constraint
  \eqref{e:5.5} hold for $(\lambda,\sigma)$. Then 
  $Z (\zeta)$ satisfies the conformal mapping identity \eqref{e:write-up53}.
\end{corollary}

%\section{Conclusion}
 %\VK{Not really sure what to write here...}
 
 \begin{itemize}
 	\item[] \textbf{Funding}  S. T. acknowledges funding from the Simons Foundation international through the travel grant SFI-MPS-TSM00013918 
 	\item [] \textbf{Data availability } Data sharing is not applicable to this article as no new data were created or analyzed in this work.
 	\item[]\textbf{Declarations} 
 	\item[]\textbf{Conflict of interest} The authors certify that they have no conflicting interests of any kind related to the content of this manuscript
 \end{itemize}

%%%%%%%%%%%%%%%%%%%%%%%%%%%%
\bibliography{bibref}

\begin{thebibliography}{10}

\bibitem{Aref1983}
Hassan Aref.
\newblock Integrable, chaotic, and turbulent vortex motion in two-dimensional
  flows.
\newblock {\em Annual Review of Fluid Mechanics}, 15:345--389, 1983.

\bibitem{BakerSaffmanSheffield1976}
G.~R. Baker, P.~G. Saffman, and J.~S. Sheffield.
\newblock Structure of a linear array of hollow vortices of finite
  cross-section.
\newblock {\em Journal of Fluid Mechanics}, 74(3):469--476, 1976.

\bibitem{ChenWalshWheeler2026}
Robin~Ming Chen, Samuel Walsh, and Miles~H. Wheeler.
\newblock Desingularization and global continuation for hollow vortices.
\newblock {\em Annals of PDE}, 12(1):7, 2026.

\bibitem{CliftGraceWeber1978}
Roland Clift, John~R. Grace, and Martin~E. Weber.
\newblock {\em Bubbles, Drops, and Particles}.
\newblock Academic Press, New York, 1978.

\bibitem{Crapper1957}
G.~D. Crapper.
\newblock An exact solution for progressive capillary waves of arbitrary
  amplitude.
\newblock {\em Journal of Fluid Mechanics}, 2(6):532--540, 1957.

\bibitem{Crowdy1999}
Darren~G. Crowdy.
\newblock Circulation-induced shape deformations of drops and bubbles: Exact
  two-dimensional models.
\newblock {\em Physics of Fluids}, 11(10):2836--2845, 1999.

\bibitem{CrowdyNelsonKrishnamurthy2021}
Darren~G. Crowdy, R.~B. Nelson, and V.~S. Krishnamurthy.
\newblock {H}-states: exact solutions for a rotating hollow vortex.
\newblock {\em Journal of Fluid Mechanics}, 913:R5, 2021.

\bibitem{CrowdyRoenby2014}
Darren~G. Crowdy and Johan Roenby.
\newblock Hollow vortices, capillary water waves and double quadrature domains.
\newblock {\em Fluid Dynamics Research}, 46(3):031424, 2014.

\bibitem{CrowdyTanveer2025}
Darren~G. Crowdy and Saleh Tanveer.
\newblock A mathematical analogy between constant vorticity flows without
  surface tension and irrotational flows with surface tension.
\newblock {\em Proceedings of the Royal Society A: Mathematical, Physical and
  Engineering Sciences}, 481(2320):20240951, 2025.

\bibitem{Donnelly1991}
Russell~J. Donnelly.
\newblock {\em Quantized Vortices in Helium {II}}, volume~3 of {\em Cambridge
  Studies in Low Temperature Physics}.
\newblock Cambridge University Press, Cambridge, 1991.

\bibitem{Duren1970TheoryHpSpaces}
Peter~L. Duren.
\newblock {\em Theory of {$H^p$} Spaces}, volume~38 of {\em Pure and Applied
  Mathematics}.
\newblock Academic Press, New York and London, 1970.

\bibitem{DyachenkoHur2019b}
Sergey~A. Dyachenko and Vera~Mikyoung Hur.
\newblock {Stokes} waves with constant vorticity: folds, gaps and fluid
  bubbles.
\newblock {\em Journal of Fluid Mechanics}, 878:502--521, 2019.

\bibitem{DyachenkoHur2019a}
Sergey~A. Dyachenko and Vera~Mikyoung Hur.
\newblock {Stokes} waves with constant vorticity: {I}. numerical computation.
\newblock {\em Studies in Applied Mathematics}, 142(2):162--189, 2019.

\bibitem{Flajolet1990Singularity}
Philippe Flajolet and Andrew~M. Odlyzko.
\newblock Singularity analysis of generating functions.
\newblock {\em SIAM Journal on Discrete Mathematics}, 3(2):216--240, 1990.

\bibitem{HurVandenBroeck2020}
Vera~Mikyoung Hur and Jean-Marc Vanden-Broeck.
\newblock A new application of {Crapper}'s exact solution to waves in constant
  vorticity flows.
\newblock {\em European Journal of Mechanics - B/Fluids}, 83:190--194, 2020.

\bibitem{HurWheeler2020}
Vera~Mikyoung Hur and Miles~H. Wheeler.
\newblock Exact free surfaces in constant vorticity flows.
\newblock {\em Journal of Fluid Mechanics}, 896:R1, 2020.

\bibitem{CKT}
Valentin~D. Kunz, Saleh Tanveer, and Darren~G. Crowdy.
\newblock Extension of a mathematical analogy between vortical flows without
  surface tension and irrotational capillary flows.
\newblock {\em Submitted to Proc. Roy. Soc. A}, 2026.

\bibitem{LegendreZenit2025}
Dominique Legendre and Roberto Zenit.
\newblock Gas bubble dynamics.
\newblock {\em Reviews of Modern Physics}, 97(2):025001, 2025.

\bibitem{LehmannSchmied2003}
Kevin~K. Lehmann and Roman Schmied.
\newblock Energetics and possible formation and decay mechanisms of vortices in
  helium nanodroplets.
\newblock {\em Physical Review B}, 68(22):224520, 2003.

\bibitem{LlewellynSmithCrowdy2012}
Stefan~G. Llewellyn~Smith and Darren~G. Crowdy.
\newblock Structure and stability of hollow vortex equilibria.
\newblock {\em Journal of Fluid Mechanics}, 691:178--200, 2012.

\bibitem{McLeod1955}
E.~B. McLeod, Jr.
\newblock The explicit solution of a free boundary problem involving surface
  tension.
\newblock {\em Journal of Rational Mechanics and Analysis}, 4:557--567, 1955.

\bibitem{Newton2001}
Paul~K. Newton.
\newblock {\em The {N}-Vortex Problem: Analytical Techniques}, volume 145 of
  {\em Applied Mathematical Sciences}.
\newblock Springer, New York, 2001.

\bibitem{Tanveer1996}
Saleh Tanveer.
\newblock Some analytical properties of solutions to a two-dimensional steadily
  translating inviscid bubble.
\newblock {\em Proceedings of the Royal Society of London. Series A:
  Mathematical, Physical and Engineering Sciences}, 452(1949):1397--1410, 1996.

\bibitem{VandenBroeckKeller1980}
Jean-Marc Vanden-Broeck and Joseph~B. Keller.
\newblock Bubble or drop distortion in a straining flow in two dimensions.
\newblock {\em The Physics of Fluids}, 23(8):1491--1495, 1980.

\bibitem{WegmannCrowdy2000}
Rudolf Wegmann and Darren Crowdy.
\newblock Shapes of two-dimensional bubbles deformed by circulation.
\newblock {\em Nonlinearity}, 13(6):2131--2141, 2000.

\bibitem{YanEtAl2023}
Sheng-lin Yan, Xu-qing Wang, Li-tao Zhu, Xi-bao Zhang, and Zheng-hong Luo.
\newblock Mechanisms and modeling of bubble dynamic behaviors and mass transfer
  under gravity: A review.
\newblock {\em Chemical Engineering Science}, 277:118854, 2023.

\bibitem{ZannettiFerlautoLlewellynSmith2016}
Luca Zannetti, Michele Ferlauto, and Stefan~G. Llewellyn~Smith.
\newblock Hollow vortices in shear.
\newblock {\em Journal of Fluid Mechanics}, 809:705--715, 2016.

\end{thebibliography}
\bibliographystyle{plain}

%%%%%%%%%%%%%%%%%%%%%%%%%%%%
\appendix

\section{On the numerical verification of univalency}\label{App:Univ}

In this appendix, we shall discuss how univalency of solutions to Problem (a)
can be verified numerically as was mentioned in Remark \ref{rem.25}. So, let
us consider a conformal map $\tilde{Z} (\zeta)$ which lies close to the
solution $Z (\zeta)$ to Problem (a). For instance, $\tilde{Z}  (\zeta)$ could
be any of the approximative numerical solutions computed via the algorithm
detailed in \cite{CKT}. 

Consider the boundary curve $\partial
\tilde{\Omega} = \tilde{Z} (\partial \mathbb{D})$. Let $\zeta = e^{i \nu}$ for $| \nu | \leq \pi$ and consider a finite number
of ``anchor points'' $\nu_j \in (- \pi, \pi],~  j = 1, \ldots, n$. Without loss of generality, let us assume that $- \pi \leq \nu_1 < \nu_2 <
\cdots < \nu_n < \pi$ and let
\begin{equation}
	\Delta = \max \{ \nu_{j + 1} - \nu_j, (\nu_1 + 2 \pi) - \nu_n \}>0.
\end{equation}

 We assume that
the following a-priori hypotheses are true 
\begin{equation}
  \min_j | \tilde{Z}_{\nu} (e^{i \nu_j}) | \geq m \quad \text{and} \quad
  \min_{j, k} \left| \frac{\tilde{Z} (e^{i \nu_j}) - \tilde{Z} (e^{i
  \nu_k})}{\nu_j - \nu_k} \right| \geq m \label{e:155}
\end{equation}
with $m > 0$ for a set of anchor points sufficiently dense, and uniform, such that 	
	\begin{align}
		\min\left\{m - \frac{1}{2} \left(2\Delta + \Delta^{\frac{1}{2}}\right)\sup_{\xi} | \tilde{Z}_{\nu \nu} (e^{i \xi})|, ~ m - 3 \frac{\Delta \sup_{\xi } | \tilde{Z}_{\nu}(e^{i \xi})
			|}{\Delta^{\frac{1}{2}} - \Delta}\right\} >0. \label{e:A.3}
	\end{align}
 Importantly, given any numerically computed solution
$\tilde{Z}(\zeta)$ and finite set of anchor points, the conditions \eqref{e:155} and \eqref{e:A.3} can be verified explicitly. Indeed, if $\tilde{Z}(\zeta)$ were univalent, there exists a minimal $m^{\star}$ such that \eqref{e:155} holds irrespective of the anchor point distribution. If, on the other hand, $\tilde{Z}(\zeta)$ is \textit{not} univalent, then necessarily $m \to 0$ as $\Delta \to 0$. 

Our goal
is to show that the conditions \eqref{e:155} and \eqref{e:A.3} imply the inequality \eqref{e:66}, i.e.,
we will now prove that it is sufficient to numerically verify the condition
\eqref{e:66} at a finite number of anchor points for some ``sufficiently dense and uniform'' distribution of anchor-points.

Measuring distances modulo $2 \pi$, note that
\begin{equation}
  \min \{| \nu - \nu_j |, | \nu - \nu_{j + 1} \} \leq
  \frac{\Delta}{2}, \quad \text{for all} \quad \nu
\end{equation}
i.e. for any $\nu$, the corresponding nearest anchor point satisfies
\begin{equation}
  | \nu - \nu_j | \leq \frac{\Delta}{2} . \label{e:158}
\end{equation}
We will distinguish two cases. First, $|\nu - \nu'| < \Delta + \Delta^{\frac{1}{2}}$, and second, $|\nu - \nu'| \geq \Delta + \Delta^{\frac{1}{2}}$.

\textbf{Case 1} ($|\nu - \nu'| < \Delta + \Delta^{\frac{1}{2}}$)\textbf{.} Using Taylor's formula with integral remainder, we find that for any pair
$\nu, \nu'$, we have
\begin{equation}
  \left| \frac{\tilde{Z} (e^{i \nu}) - \tilde{Z} \left( {e^{i \nu}}'
  \right)}{\nu - \nu'} - \tilde{Z}_{\nu} (e^{i \nu}) \right| \leq \frac{| \nu
  - \nu' |}{2} \sup_{\xi} | \tilde{Z}_{\nu \nu} (e^{i \xi}) | . \label{e:159}
\end{equation}
From
mean-value inequality,
\begin{equation}
  | \tilde{Z}_{\nu} (e^{i \nu}) - \tilde{Z}_{\nu} (e^{i \nu_j}) | \leq | \nu -
  \nu_j | \sup_{\xi} | \tilde{Z}_{\nu \nu} (e^{i \xi}) | 
\end{equation}
and thus, using \eqref{e:155}, we find
\begin{equation}
  | \tilde{Z}_{\nu} (e^{i \nu}) | \geq m - | \nu - \nu_j |\sup_{\xi} | \tilde{Z}_{\nu \nu} (e^{i \xi}) | .
  \label{e:161}
\end{equation}
Combining \eqref{e:158}, \eqref{e:159}, \eqref{e:161}, recalling that $|\nu-\nu'| < \Delta + \Delta^{\frac{1}{2}}$, and using \eqref{e:A.3}, we thus obtain
\begin{equation}
  \left| \frac{\tilde{Z} (e^{i \nu}) - \tilde{Z} \left( {e^{i \nu}}'
  \right)}{\nu - \nu'} \right|  \geq m - \frac{1}{2} (\Delta + | \nu - \nu' |)
  \sup_{\xi} | \tilde{Z}_{\nu \nu} (e^{i \xi}) | >0
\end{equation}
which yields the inequality \eqref{e:66}. 

\textbf{Case 2} ($| \nu - \nu' | > \Delta + \Delta^{\frac{1}{2}}$)\textbf{.}  Choose two nearest anchor-points $\nu_m$ and $\nu_n$ of $\nu$ and $\nu'$, respectively, and let
\begin{equation}
  \nu - \nu_m = s, \quad \nu' - \nu_n = t
\end{equation}
so that
\begin{equation}
  | s - t | \leq \Delta \quad \text{and} \quad | \nu_m - \nu_n | >
  \Delta^{\frac{1}{2}} . \label{e:164}
\end{equation}
Computations show that
\begin{align*}
  \left| \frac{\tilde{Z} (e^{i \nu}) - \tilde{Z} (e^{i \nu'})}{\nu - \nu'} -
  \frac{\tilde{Z} (e^{i \nu_m}) - \tilde{Z} (e^{i \nu_n})}{\nu_m - \nu_n}
  \right| \leq & \left| \frac{s - t}{\nu_m - \nu_n + s - t} \right| \left|
  \frac{\tilde{Z} (e^{i \nu_m}) - \tilde{Z} (e^{i \nu_n})}{\nu_m - \nu_n}
  \right|\\
  + & \left| \frac{\tilde{Z} (e^{i \nu}) - \tilde{Z} (e^{i \nu_m})}{\nu_m -
  \nu_n + s - t} \right| + \left| \frac{\tilde{Z} (e^{i \nu'}) - \tilde{Z}
  (e^{i \nu_n})}{\nu_m - \nu_n + s - t} \right|\\
  \leq & 3 \frac{\Delta \sup | \tilde{Z}_{\nu} |}{\Delta^{\frac{1}{2}} -
  \Delta} \numberthis
\end{align*}
and therefore, we find
\begin{equation}
  \left| \frac{\tilde{Z} (e^{i \nu}) - \tilde{Z} (e^{i \nu'})}{\nu - \nu'}
  \right| \geq m - 3 \frac{\Delta \sup_{\xi} | \tilde{Z}_{\nu}(e^{i \xi})
  |}{\Delta^{\frac{1}{2}} - \Delta}
\end{equation}
from which, due to \eqref{e:A.3}, the inequality \eqref{e:66}
follows as well. \hfill  $\Box$

\section{Analyticity Properties of $\mathcal{I}$}

\subsection{Proof of Lemma \ref{lem:C-14}}\label{App.A}

In order to prove Lemma \ref{lem:C-14}, we first note that the condition that
$p$ has at most a simple pole at the origin guarantees that $\mathcal{I} [p]
(\eta)$ is well-defined, and it is clearly analytic in $\mathbb{D}_c$, since
the only possibly singular points of the integrand are located at $\eta_1,
\eta_2$ or the origin. Below, we will closely follow the proof of Lemma 20 within the  supplementary material of \cite{CrowdyTanveer2025}.

Let us now prove statement 1 of the Lemma. So, let $(\eta - \eta_j) p$
be analytic within $ \mathbb{D} \setminus \{
\eta_{j - 1} \}$, and let $\nu_j < 0$. Consider the functional
$\mathcal{I} [p]$ and introduce the M\"obius transformation
\begin{equation}
  \xi = \frac{\eta - \eta_j}{1 + \eta \eta_j}~ \text{ such that } ~ \eta =
  \frac{\xi + \eta_j}{1 - \xi \eta_j}.
\end{equation}
So,
\begin{equation}
  \mathcal{I} [p] = \left( \frac{\xi + \eta_j}{1 - \eta_j \xi} \right)^{-
  \frac{3}{2}} \frac{\xi^{- \nu_j}}{g_{j - 1} (\xi)} \int^{\xi}_{- \eta_j}
  \left( \frac{\xi' + \eta_j}{1 - \xi' \eta_j} \right)^{\frac{1}{2}} g_{j - 1}
  (\xi') (\xi')^{\nu_j} p(\eta (\xi')) \left( \frac{d}{d \xi'} \eta
  (\xi') \right) d \xi'  \label{e:0066}
\end{equation}
with
\[ g_{j - 1} (\xi) = \left( \frac{\eta (\xi) - \eta_{j - 1}}{1 + \eta (\xi)
   \eta_{j - 1}} \right)^{\nu_{j - 1}} . \]
The function $g_{j - 1}$ is holomorphic in a disc of radius $| \eta_j |$
centred at $\eta (\xi) = \eta_j$, i.e., it is holomorphic near $\xi = 0$, and
also near $\xi = - \eta_j .$ We re-write the integral in \eqref{e:0066} as
\begin{equation}
  \mathcal{I} [p] = \left( \frac{\xi + \eta_j}{1 - \eta_j \xi} \right)^{-
  \frac{3}{2}} \frac{\xi^{- \nu_j}}{g_{j - 1} (\xi)} \int^{\xi}_{- \eta_j}
  h_j (\xi') (\xi')^{\nu_j - 1} d \xi'  \label{e:034}
\end{equation}
where we have introduced the function
\begin{equation}
  h_j (\xi') = \left( \frac{\xi' + \eta_j}{1 - \xi' \eta_j}
  \right)^{\frac{1}{2}} g_{j - 1} (\xi') q_j (\xi') \left( \frac{d}{d
  \xi'} \eta (\xi') \right), \quad \text{with} \quad q_j(\xi') = \xi' p(\eta(\xi')).
\end{equation}
Since, by assumption, $q_j$ is holomorphic at $\xi'$=0, the
function $h_j (\xi')$ is analytic within a disc of radius $| \eta_j |$
centred at $\xi'$=0, with the nearest singularity located at $\xi' = - \eta_j$,
see the Figure \ref{Fig:04}. As $\xi' \rightarrow - \eta_j$,
we find the following asymptotic behaviour
\begin{equation}
  h_j (\xi') \sim (1 + \eta^2_j)^{- \frac{1}{2}} \left( \frac{d}{d \xi'} \eta
  (0) \right) g_{j - 1} (- \eta_j) q_j (0) (\xi' + \eta_j)^{\frac{1}{2}} .
  \label{e:0029}
\end{equation}
Therefore, upon writing
\begin{equation}
  h_j (\xi') = \sum_{k = 0}^{\infty} h_{ k, j}   (\xi')^k
\end{equation}
we obtain
\begin{equation}
  h_{ k, j}  \sim (- 1)^{k + 1} \frac{1}{2} (1 + \eta_j^2)^{- \frac{1}{2}}
  \eta' (0) q_j (0) g_{j - 1} (0) \eta_j^{\frac{1}{2} - k} \sqrt{\pi} k^{-
  \frac{3}{2}} \quad \text{as} \quad k \rightarrow \infty . \label{e:031}
\end{equation}
The passage from \eqref{e:0029} to \eqref{e:031} is an application of
 \cite{Flajolet1990Singularity}, Corollary 2.

To finish the proof of statement 1 of Lemma \ref{lem:C-14}, let us consider
the cases where $\nu_j \in \mathbb{Z}$ and $\nu_j \notin \mathbb{Z}$
separately. Let $\nu_j \notin \mathbb{Z}$. Then, we can integrate
\eqref{e:034} directly and find
\begin{equation}
  \mathcal{I} [p] = \left( \frac{\xi + \eta_j}{1 - \eta_j \xi} \right)^{-
  \frac{3}{2}} \frac{\xi^{- \nu_j}}{g_{j - 1} (\xi)} \left( \sum_{k =
  0}^{\infty} \frac{h_{ k}^j \xi^{\nu_j + k}}{k + \nu_j} - (-
  \eta_j)^{\nu_j} \sum_{k = 0}^{\infty} \frac{h_{ k}^j (- \eta_j)^k}{k +
  \nu_j} \right) \label{e:C-73}
\end{equation}
where the estimate \eqref{e:031} guarantees convergence of the series in
\eqref{e:C-73}. Thus, the integral $\mathcal{I} [p]$ is analytic at $\xi =
0$ (i.e., at $\eta = \eta_j$), if, and only if,
\begin{equation}
  \sum_{k = 0}^{\infty} \frac{h_{ k, j}   (- \eta_j)^k}{k + \nu_j} = 0.
\end{equation}
In turn, since $\nu_j \notin \mathbb{Z}$, this is equivalent to
$\mathcal{J}_j [p] = 0$.

Let us now consider $\nu_j \in \mathbb{Z}, ~ \nu_j < 0$. Then, we obtain
 \textit{at most} one singular term by computing the integral in
\eqref{e:034}. If such term occurs, it is a term involving $\xi^{- \nu_j} h
_{- \nu_j, j} (\log (\xi) - \log (- \eta_j))$, which is obtained from
integrating $\xi^{- \nu_j}  \int^{\xi}_{- \eta_j} h_{ - \nu_j, j} 
(\xi')^{- 1} d \xi'$. This term branches at $\xi = 0$, so holomorphicity at
$\xi = 0$ is equivalent to imposing $h _{- \nu_j, j} = 0$. This is again
equivalent to $\mathcal{J}_j [p] = 0$. Thus, we have established the
analyticity properties that were claimed in the first statement of Lemma
\ref{lem:C-14}.

\begin{figure}[h!]
	\centering 
		\includegraphics[width=.75\textwidth]{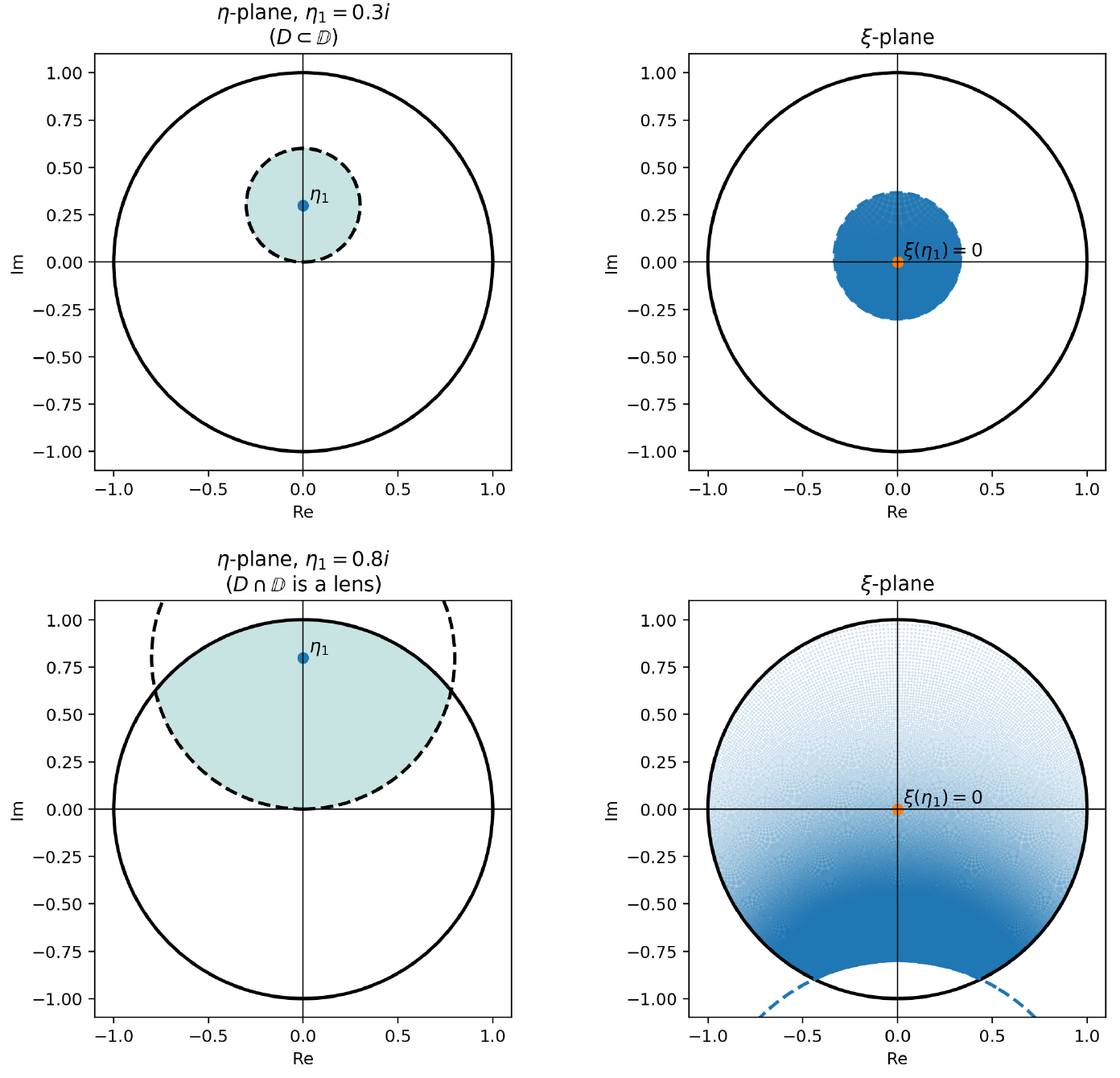}
		\caption{\small Intersection of the disc $D = \{ | \eta - \eta_1 | < |
			\eta_1 | \}$ with $\mathbb{D}$ (left), and $\xi (D \cap \mathbb{D})$
			(right).}
			\label{Fig:04}
\end{figure}

The second statement of Lemma \ref{lem:C-14} can be proved similarly, so we will be brief. Since $p$ is analytic at $\eta =
0$, the function $g (\eta) = \mu (\eta) \eta^{- \frac{3}{2}} p (\eta)$ is
analytic at $\eta = 0$, with nearest singularity located at either $\eta_1$,
$\eta_2$, or any of $p$'s singularities. We denote the latter, collectively,
by $\eta_{\text{sing}}$. Write, near $\eta = 0$,
\begin{equation}
  g (\eta) = \sum_{k = 0}^{\infty} g_k \eta^k
\end{equation}
so, for $| \eta |$ bounded away from $\min \{ | \eta_1 |, | \eta_2 |, | \eta_{\text{sing}} |
\}$, we have
\begin{equation}
  \mathcal{I} [p] (\eta) = \frac{1}{\mu (\eta) } \sum_{k = 0}^{\infty}
  \int_0^{\eta} g_k \left( \chi  \right)^{\frac{1}{2} + k} d \chi =
  \frac{1}{\mu (\eta)} \sum_{k = 0}^{\infty} \frac{g_k}{\frac{3}{2} + k}
  \eta^{\frac{3}{2} + k} = \frac{\eta^{\frac{3}{2}}}{\mu (\eta)} \sum_{k =
  0}^{\infty} \frac{g_k}{\frac{3}{2} + k} \eta^k \label{e:C-76}
\end{equation}
and the right-hand-side of \eqref{e:C-76} is analytic at $\eta = 0$. Moreover,
it vanishes at $\eta = 0$ if, and only if, $g_0 = 0$, so if, and only if, $g
(0) = 0$.

Finally, continuity of $\mathcal{I} [p ]$ on the boundary if $p $ is
continuous up to the boundary follows by inspection of $\mathcal{I} [p]$ since
all singularities of $\mu$ are bounded away from $\mathbb{T}_c$. This
completes the proof of Lemma \ref{lem:C-14}. \hfill $\Box$

\subsection{The integral operator $\mathcal{I} \circ
\mathcal{P}$}\label{App.A.2}

Let $g$ be continuous on $\partial \mathbb{D}$, so $\mathcal{P} [g]$ is
analytic in $\mathbb{D}$. Consider the function $\mathcal{I} [\mathcal{P}
[g]]$. Then, exactly as in the proof of Lemma 31 within the supplementary material of \cite{CrowdyTanveer2025}, 
we have
\begin{equation}
  \mathcal{I} [\mathcal{P} [g]] (\eta) = \frac{1}{2 \pi i} \oint_{| \chi | =
  1} K (\chi, \eta) g (\chi) d \chi . \label{e:074}
\end{equation}
The kernel
\begin{equation}
  K (\chi, \eta) = \frac{1}{\mu (\eta) \chi} \int_0^{\eta} \frac{\mu
  (\eta')}{\eta'} \frac{\chi + \eta'}{\chi - \eta'} d \eta'
\end{equation}
is analytic in $\eta \in \mathbb{D}_c $ and continuous for $\eta \in
\mathbb{T}_c \setminus \{ \chi \}$. At $\eta = \chi$, $K$ has a $\log (\chi -
\eta)$ singularity, which is integrable. So, \eqref{e:074} guarantees that
$\mathcal{I} [\mathcal{P} [g]] (\eta)$ is continuous on $\mathbb{T}_c$.
Moreover, on $\mathbb{T}_c$, we find that
\begin{equation}
  | \mathcal{I} [\mathcal{P}] [g] (\eta) | \leq c_{\mathcal{I}}^{\mathcal{P}}
  \sup_{\eta \in \partial \mathbb{D}} | g | \label{e:076}
\end{equation}
for a constant $c_{\mathcal{I}}^{\mathcal{P}}$ independent of the locations of
$\eta_1$ and $\eta_2$. All these statements can be proved by verbatim
adaptation of the proof of Lemma 31 within the supplementary material of \cite{CrowdyTanveer2025}, and we omit these details here for brevity.

\end{document}